\documentclass{article}

\usepackage{authblk}
\input{preamble}
\usepackage{thm-restate}
\usepackage[final, commandnameprefix=ifneeded]{changes}
\usepackage{lineno}
\newcommand{\kiran}[1]{{\color{blue} { \textbf Kiran}: #1}}

\title{Distributed Algorithms from Arboreal Ants  for the\\ Shortest Path Problem}

\author{Shivam Garg$^1$, \hspace{1pt} Kirankumar Shiragur$^1$,\hspace{1pt} Deborah M. Gordon, \hspace{1pt} Moses Charikar}
\affil{\texttt{\{shivamgarg, shiragur, dmgordon, charikar\}@stanford.edu}\\ \vspace{5pt} Stanford University

\vspace{-10pt}}
\date{}

\begin{document}
	\maketitle

\def\thefootnote{1}\footnotetext{These authors contributed equally to this work.}
\def\thefootnote{2}\footnotetext{Code for simulations is available at \href{https://github.com/shivamg13/Arboreal-Ants}{ this https URL.}}
\def\thefootnote{3}\footnotetext{Published in PNAS: \href{https://www.pnas.org/doi/10.1073/pnas.2207959120}{pnas.org/doi/10.1073/pnas.2207959120.}}

% \newpage
% \linenumbers
\begin{abstract}
Colonies of the arboreal turtle ant create networks of trails that link nests and food sources on the graph formed by branches and vines in the canopy of the tropical forest. Ants put down a volatile pheromone on edges as they traverse them. At each vertex, the next edge to traverse is chosen using a decision rule based on the current pheromone level. There is a bidirectional flow of ants around the network. In a field study, \textcite{chandrasekhar2019better} observed that the trail networks approximately minimize the number of vertices, thus solving a variant of the popular shortest path problem without any central control and with minimal computational resources. We propose a biologically plausible model, based on a variant of the reinforced random walk on a graph, which explains this observation and suggests surprising algorithms for the shortest path problem and its variants. Through simulations and analysis, we show that when the rate of flow of ants does not change, the dynamics converges to the path with the minimum number of vertices, as observed in the field. The dynamics converges to the shortest path when the rate of flow increases with time, so the  colony can solve the shortest path problem merely  by increasing the flow rate. We also show that to guarantee convergence to the shortest path, bidirectional flow and a decision rule dividing the flow in proportion to the pheromone level are necessary, but convergence to approximately short paths is possible with other decision rules.
\vspace{10pt}
\end{abstract}

\vspace{-7pt}
\section{Introduction}
\replaced{Biological systems, such as ant trail networks, are fascinating examples of distributed algorithms in nature \autocite{lynch1996distributed, navlakha2014distributed, feinerman2017individual, couzin2005effective} , often finding globally optimum solutions using simple local interactions among individuals, devoid of central control. The study of natural algorithms has led to synergistic exchange between biology and computer science \autocite{navlakha2011algorithms, chazelle2012natural}. The algorithmic lens has enhanced our understanding of biological phenomena such as how birds flock \autocite{chazelle2009natural}, how slime molds solve the shortest path problem \autocite{ nakagaki2000maze, bonifaci2012physarum, straszak2021iteratively} and how computation takes place in the brain  \autocite{papadimitriou2020brain, valiant2000circuits}. Moreover, the process of  evolution itself has been studied using an algorithmic lens \autocite{chastain2014algorithms, valiant2009evolvability, livnat2008mixability,  vishnoi2014speed}. Also, inspiration from nature has led to new algorithmic ideas such as ant-inspired algorithms for distributed density estimation \autocite{musco2016ant},  artificial neural networks in machine learning \autocite{lecun2015deep}, and algorithms for similarity search inspired by the fruit fly brain \autocite{dasgupta2017neural},  among others 
 \autocite{ dasgupta2018neural, shen2020habituation, afek2011biological, cardelli2012cell}.}{Biological systems, such as ant trail networks, are fascinating examples of distributed algorithms in nature \autocite{lynch1996distributed, navlakha2014distributed, feinerman2017individual, couzin2005effective} , often finding globally optimum solutions using simple local interactions among individuals, devoid of central control. The study of natural algorithms has led to synergistic exchange between biology and computer science \autocite{navlakha2011algorithms, chazelle2012natural}: the algorithmic lens has enhanced our understanding of biological phenomena such as how birds flock and how slime molds solve a maze \autocite{chazelle2009natural,  nakagaki2000maze, bonifaci2012physarum, papadimitriou2020brain, valiant2000circuits, valiant2009evolvability}; parallels have been drawn between biological processes and known algorithms, such as the mathematical process underlying sexual evolution and the multiplicative weights update algorithm in computer science \autocite{chastain2014algorithms, straszak2021iteratively}, and inspiration from nature has led to new algorithms  such as ant-inspired algorithms for distributed density estimation \autocite{musco2016ant, dasgupta2017neural, dasgupta2018neural, shen2020habituation, afek2011biological, cardelli2012cell, lecun2015deep}.}

Here we investigate how the trail networks of the arboreal turtle ant (Cephalotes goniodontus) can solve variants of the shortest path problem, a basic optimization problem on graphs  \autocite{ford1956network, bellman1958routing, dijkstra1959note}.
Textbook algorithms for this problem find optimum solutions using knowledge of the entire network \autocite{kleinberg2006algorithm, dasgupta2008algorithms, cormen2009introduction}.
Turtle ants nest and forage in the tree canopy of the tropical forest; their trail network is constrained to lie on a natural graph formed by tangled branches and vines, and no ant has any global information about the network. Observations of turtle ants in the field show that a colony's trail network approximately minimizes the number of vertices \autocite{chandrasekhar2019better}. 
We develop a model that gives a biologically plausible explanation for this observation, and outlines other intriguing phenomena as described in the next section.

% \newpage
\subsection{Summary of model and results}

\replaced{Turtle ant colonies form trails on a graph whose vertices correspond to junctions in the vegetation, and edges correspond to branches connecting these junctions. 
A colony's network of trails connects many nests and food sources. 
The trail network minimizes the number of vertices \autocite{chandrasekhar2019better} (compared to simulated random networks),  approximately solving a variant of the Steiner tree problem \autocite{winter1987steiner, charikar1999approximation, latty2011structure}, which is a generalization of the shortest path problem for multiple terminal vertices.
Here we focus on a section of this network, considering  two terminal vertices, such as a nest and a food source, and we seek to explain how a colony can find the path with minimum number of vertices between the two terminals. We model trails as a bidirectional flow of ants between the two terminal vertices; bidirectional flow is characteristic of the trails of this species \autocite{gordon2017local}. The flow in our model is similar to models of flow in traffic networks \autocite{wardrop1952road, beckmann1955studies, roughgarden2002bad}.}{Our model builds on previous work on biologically feasible local algorithms \autocite{chandrasekhar2018distributed}.  We model trails as a bidirectional flow of ants between two terminal vertices, such as a nest and a food source; bidirectional flow is characteristic of the trail networks of {\em C. goniodontus} \autocite{gordon2017local}. 
This is similar to models of  flow in traffic networks \autocite{wardrop1952road, beckmann1955studies, roughgarden2002bad}. Trails are formed on a graph, whose vertices correspond to junctions in the vegetation, and edges correspond to branches connecting these junctions.} 
Ants lay trail pheromone on edges as they traverse them, and the next edge to traverse is chosen based on the level of pheromone. The pheromone decays with time. Some fraction of flow leaks as it passes through each vertex, modeling the loss of ants due to exploration. \textcite{chandrasekhar2019better} hypothesized 
loss of ants at the vertices to be the reason why ants prefer paths with fewer vertices.
\begin{figure}[h]%{r}{0.5\textwidth}
  \begin{center}
\includegraphics[width=0.5\textwidth]{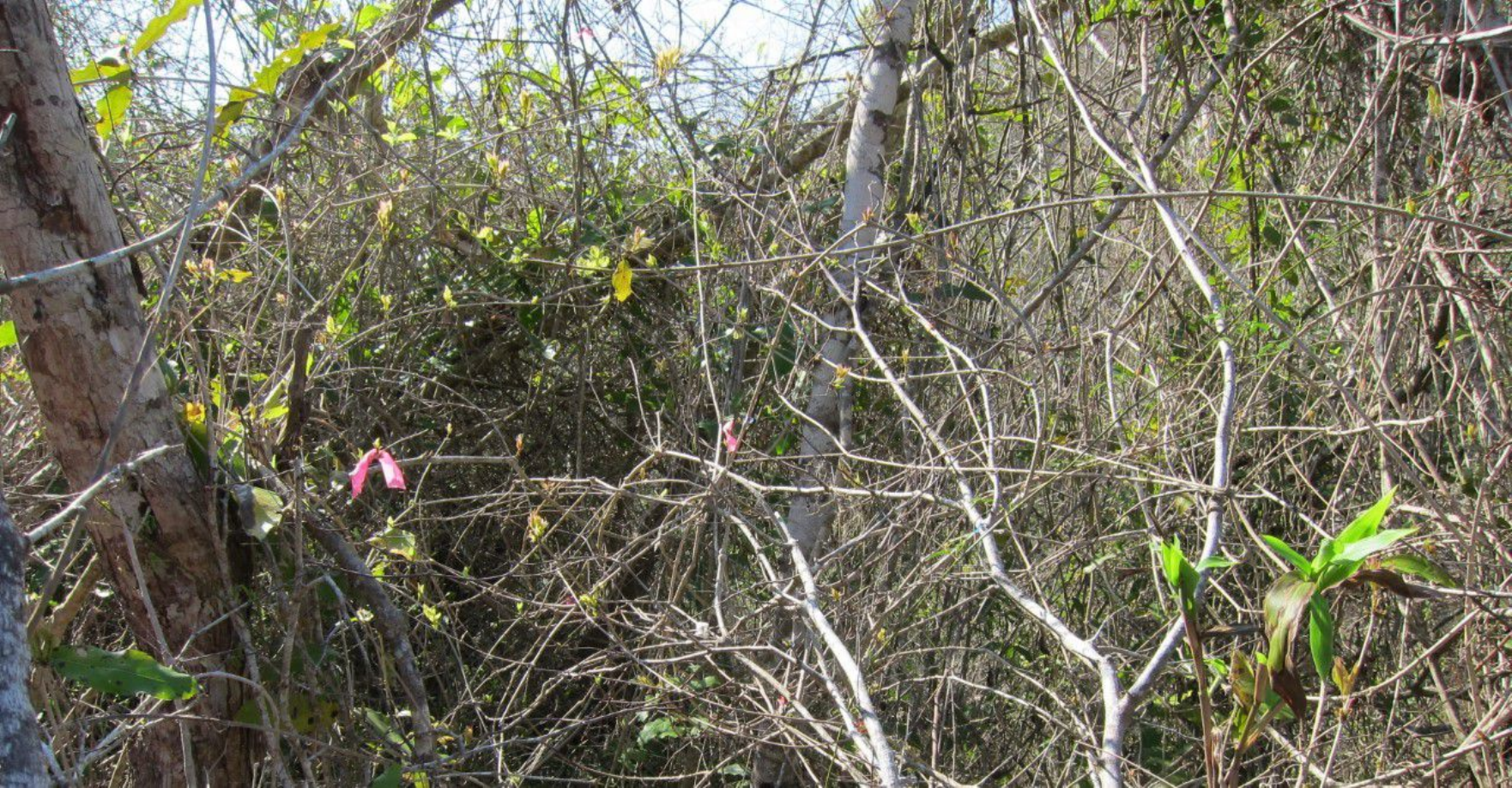}%{branches_pic.jpeg}
  \end{center}
\vspace{-8pt}
\captionof{figure}{\small{Tangled branches in which the turtle ants forage. }}
%\label{fig:ants2_graph}
\vspace{-5pt}
\end{figure}

Our model leads to 4 main results:
\begin{enumerate}
\item We first consider the linear decision rule, which at each vertex, divides the flow among the next set of edges in proportion to their pheromone level. We show through simulations and analysis that when the incoming rate of flow remains unchanged, the dynamics converges to the path with the minimum leakage. This is also the path with the minimum number of vertices when all vertices have equal leakage. This result describes a biologically plausible process that explains how colonies can find paths with the minimum number of vertices.

\item 
\replaced{
We show that when the rate of flow increases with time, in the absence of leakage, the dynamics converges to the shortest path.  Flow rate on ant trails can change over time 
\autocite{deneubourg1986random,gordon2012dynamics, bouchebti2015contact, reid2015army, bruce2017tall},  for example, in turtle ants the flow rate can increase in response to new food sources \autocite{gordon2012dynamics}.
In other ant species, it has been shown that ant trails converge to the shortest path in certain simple graphs \autocite{goss1989self}.
Our result shows a surprising link between these two phenomena:  ant colonies can use their ability to increase the flow rate to find the shortest path.}{When the rate of flow increases with time, in the absence of leakage, we show that the dynamics converges to the shortest path. Many studies of ant trail networks show that flow rate changes over time \autocite{deneubourg1986random,gordon2012dynamics, bouchebti2015contact, reid2015army, bruce2017tall}, while others show that these networks converge to the shortest path in certain simple graphs \autocite{goss1989self} . Our result shows a surprising link between these two phenomena: an increase in the flow rate of ants facilitates the discovery of the shortest path. }

\item We establish the utility of bidirectional flow by showing that it is necessary for convergence to the shortest or the minimum leakage path. In contrast, most flow-based problems considered in computer science and operations research have unidirectional flow  \autocite{harris1955fundamentals, ahuja1988network, schrijver2002history, roughgarden2002bad}.

\item  We investigate the effect of increasing flow and leakage with decision rules other than the linear rule. For a general family of decision rules, we show that the linear rule is its unique member with guaranteed convergence to the shortest and the minimum leakage path. However, for various non-linear rules, we show that the dynamics still often converges to a path with smaller length and less leakage, compared to the path found in the absence of increasing flow and leakage respectively. Thus the utility of increasing flow and leakage is not limited to the linear decision rule.
\end{enumerate}

\added{Our model builds on a previous model by  \textcite{chandrasekhar2018distributed},  adding components such as leakage and variation in flow rate. These components were not present in the model of \textcite{chandrasekhar2018distributed} and are crucial for the phenomena we discuss above.
\textcite{chandrasekhar2018distributed} investigated how ants find alternative paths, not necessarily the ones with the minimum number of vertices,  to route around ruptured links in a network. Here we ask how ants can find the path with the minimum number of vertices, and how the flow rate impacts the path found. We demonstrate that leakage at vertices can lead to convergence to the path with the minimum number of vertices, and increase in flow rate over time can lead to convergence to the shortest path.
}

Our work is different from traditional ant-colony optimization \autocite{dorigo1991distributed, dorigo2005ant, lopez2018ant}, in which the algorithms considered are not required to be biologically plausible. Models of ant colony optimization (ACO), inspired by ant behavior, solve combinatorial optimization problems, such as the traveling salesman problem \autocite{yang2008ant} and the shortest path problem \autocite{sudholt2012running, neumann2006runtime} . In ACO, individual agents, simulating ants, construct candidate solutions using heuristics, and then use limited communication, simulating trail pheromone, to lead other agents towards better solutions. The simulated ants have significantly more computational power than is biologically plausible for real ants. Unlike real ants, the simulated ants have the ability to remember, retrace and reinforce entire paths, and can use the quality of the global solution to determine the amount of ``pheromone'' to be laid.

Our model resembles the reinforced random walks introduced by Diaconis and others \autocite{diaconis1980finetti,davis1990reinforced,pemantle2007survey}, which have found applications in biology\autocite{stevens1997aggregation,coding2008random,smouse2010stochastic}.
Here, a single agent traverses a graph by choosing edges with probability proportional to their edge weight, with edge weight here analogous to the level of trail pheromone. This edge weight increases additively each time the edge is traversed. However, there are a few key differences between the model we study and the setup for reinforced random walks: 1) Our model involves many agents, modeled by a flow, and their behavior is affected by their collective action in putting down pheromone, while the model for reinforced random walks considers a single agent whose behavior is influenced by its past random choices. 2) Pheromone decays geometrically over time in our model, but edge weights do not decrease in reinforced random walks. 3) Our model has leakage at each vertex which is not present in the setup for reinforced random walks. Nevertheless, in a similar spirit to random walk-like processes studied before, the model we investigate is a Markov process of particular relevance to biology.

\section{The Model}\label{sec:model}

We consider bidirectional flow on a directed graph $G = (V, E)$ with source vertex and destination vertex $s$ and $d$ respectively.  For each vertex $v$, let $\fflow{v}(t)$ and $\bflow{v}(t)$ denote the forward and backward flow on $v$ present at time $t$. The forward flow moves along the direction of the edges, and the backward flow moves in the opposite direction. Let $p_{uv}(t)$ denote the pheromone level on edge $(u, v)$ at time $t$.  The pheromone level and flow together constitute the state of the system at any time $t$, which is updated as follows:

\begin{enumerate}
    \item \textbf{Flow movement}: At each time step,  the forward flow on a vertex  moves along its outgoing edges,  dividing itself based on a decision rule that depends on the pheromone levels on these edges. The simplest such rule divides the flow  proportional to the pheromone level on the outgoing edges. Formally, let $\fflow{uv}(t)$ denote the forward flow moving  along edge $(u, v)$ at time $t$. Then according to this rule, which we call the \emph{linear decision rule},
    \begin{equation}
    \fflow{uv}(t) = \fflow{u}(t)\frac{p_{uv}(t)}{\sum\limits_{z: (u, z) \in E} {p_{uz}(t)}}.
    \end{equation}

    The total forward flow on vertex $v$ is the sum of flow along its incoming edges, multiplied by a leakage factor:
        \begin{equation}
        \label{eq:no_leak_flow}
        \fflow{v}(t+1) = (1 - l_v) \sum\limits_{z: (z, v) \in E}
    \fflow{zv}(t)  
    \end{equation}
    
    The leakage parameter $l_v \in [0, 1]$ models the loss of ants due to exploration at each vertex.
    % In addition, there is possible leakage of flow at any vertex $v$ governed by a leakage parameter $l_v \in [0, 1]$, such that $l_v$ fraction of flow at $v$ leaks out. Leakage models the loss of ants due to exploration. Equation \ref{eq:no_leak_flow} corresponds to the case when leakage at vertex $v$, $l_v = 0$. In general, the total forward flow on vertex $v$ is given by  
    % \begin{equation}
    %     \label{eq:leak_flow}
    %     \fflow{v}(t+1) = (1 - l_v)  \sum\limits_{z: (z, v) \in E}
    % \fflow{zv}(t)  
    % \end{equation}

    Movement for backward flow takes place in exactly the same manner, with the direction of the flow reversed. Formally, let $\bflow{uv}(t)$ denote the backward flow moving along the edge $(u, v)$ at time $t$. Then
    
        \begin{equation}
    \bflow{uv}(t) =  \ \bflow{v}(t)\frac{p_{uv}(t)}{\sum\limits_{z: (z,v) \in E} {p_{zv}(t)}}
    \end{equation}
    The total backward flow at $u$ is the sum of backward flow along its outgoing edges, multiplied by the leakage parameter. 
    %For simplicity, in our model, the leakage parameter for any given vertex  is same for forward and backward flows.
    \begin{equation}
        \bflow{u}(t+1) = (1 - l_u) \sum\limits_{z: (u,z) \in E}
    \bflow{uz}(t)  
    \end{equation}
    
   At each time $t$, new forward and backward flow $\fflow{s}(t)$ and $\bflow{d}(t)$ appears on the source vertex and destination vertex $s$ and $d$ respectively.

    \item \textbf{Pheromone update}: At each time step, the pheromone level on an edge increases by the amount of flow on it, and decays by a multiplicative factor  of $\delta$ (similar to \textcite{chandrasekhar2018distributed}):
    \begin{equation}
    p_{uv}(t+1) = \delta(p_{uv}(t) + \fflow{uv}(t)  + \bflow{uv}(t))
    \end{equation}
    Note that unlike flow, the pheromone level present on an edge does not have any direction, and is influenced by flow from both the directions.

\end{enumerate}
\begin{figure}[t]
\centering
\includegraphics[scale=0.5]{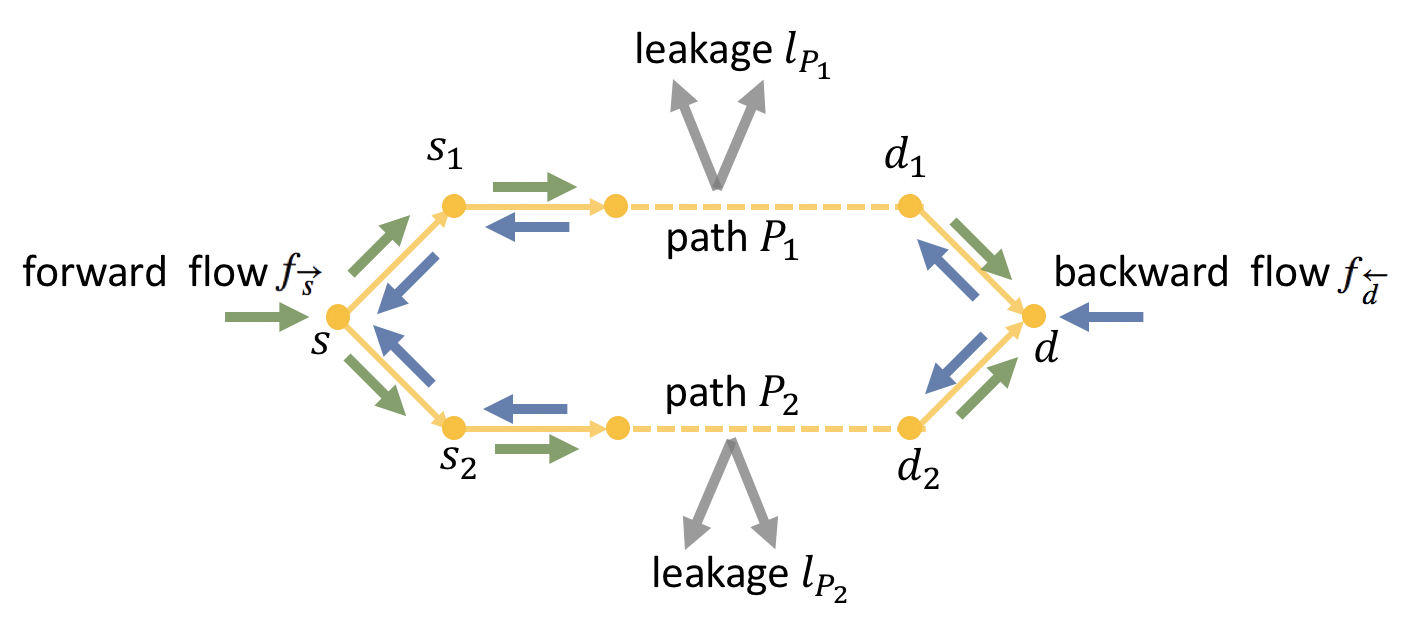}
\caption{Forward and backward flow in a graph with parallel paths.}
\label{fig:path_graph}
\end{figure}

%\moses{added definition environment}
\begin{definition}
\label{def:path_leakage}
For any path $P$ from source $s$ to destination $d$, we define its leakage $l_P$ as the fraction of flow that leaks out while moving through $P$, that is,  $l_P = 1 - \Pi_{{v} \in P \setminus \{s,d\}} (1 - l_v)$. \footnote{For convenience, we overload notation, and use $v \in P$ and $(u, v) \in P$ to denote the vertices and edges present on path $P$ respectively.} 
\end{definition}

\paragraph{Discussion of Modeling Assumptions.}
\added{
Our goal is to provide a simple and minimal biologically plausible model that explains how ants can find the path with the minimum number of vertices.  We discuss our modeling assumptions below.}

While we describe our model for unweighted graphs, it is general enough to capture the case with integral/rational edge lengths. For instance, an edge with length $k$
can be represented using $k$ unit length edges connected in series, with leakage at the vertices connecting these edges set to zero. 

\added{Our model uses a directed graph. Previous work shows that individual ants are unlikely to turn around on the trail, so that when an ant leaves a terminal, such as a nest or a food source, its distance from that terminal increases over time \autocite{gordon2012dynamics}. To account for this, \textcite{chandrasekhar2019better} assign direction to each edge relative to a terminal vertex, where the outbound direction goes away from the terminal vertex and the inbound direction goes towards it. Similarly, we use directed edges in our model.}

\added{
In our model, the backward flow at the destination vertex is not dependent on the forward flow reaching it at the previous time step. This is because an ant does not necessarily make a round trip from one terminal to the other and back. The flow between the two terminal vertices in our model represents only a section of the larger network, which includes many nests and food sources. Ants reaching a terminal vertex can go on to other parts of the network instead of turning back.}

\section{Results}\label{sec:results}
In this section, we discuss the convergence properties for the model defined above. We defer all the proofs and simulation details to the supplementary material.

\subsection{Convergence to the Minimum Leakage Path and the Shortest Path}
\label{sec:parallel_paths_ub}

\paragraph{Constant flow with time.}
Through simulations and analysis, we show that when the incoming forward and backward flow does not change with time, with the linear decision rule, the dynamics converges to the path with the minimum leakage (see Figure \ref{fig:simulation}).  

We run simulations for three families of directed graphs: 
\begin{enumerate}
    \item $G(n, p)$: The $G(n, p)$ model is a widely used random graph model which consists of graphs on $n$ vertices where each pair of vertices has an edge with probability $p$.
    \item $G(n, p)$ with a locality constraint: The standard $G(n, p)$ model allows edges between any two vertices (with probability $p$). However, the graphs formed by  branches and vines in the natural vegetation have a local physical structure in which the edges are more likely between nearby vertices. To capture this, we consider the $G(n, p)$ model with an additional locality constraint that an edge exists between vertex $i$ and $j$ only if $|i - j| \leq k$ for some parameter $k$. Here, the vertices are labelled from $1$ to $n$ with the source and the destination vertex labelled 1 and $n$ respectively.
    \item $n \times n$ grid graph.
\end{enumerate}
We generate a large number of instances with different values for $n$, $p$, $k$ and other parameters, and observe convergence to the minimum leakage path
in \emph{all} the simulated instances. More details about the simulations can be found in supplementary material Appendix \ref{sec:app_sim_details}.

\begin{figure}[t]
    \centering % <-- added
\begin{subfigure}{0.33\textwidth}
  \includegraphics[width=\linewidth]{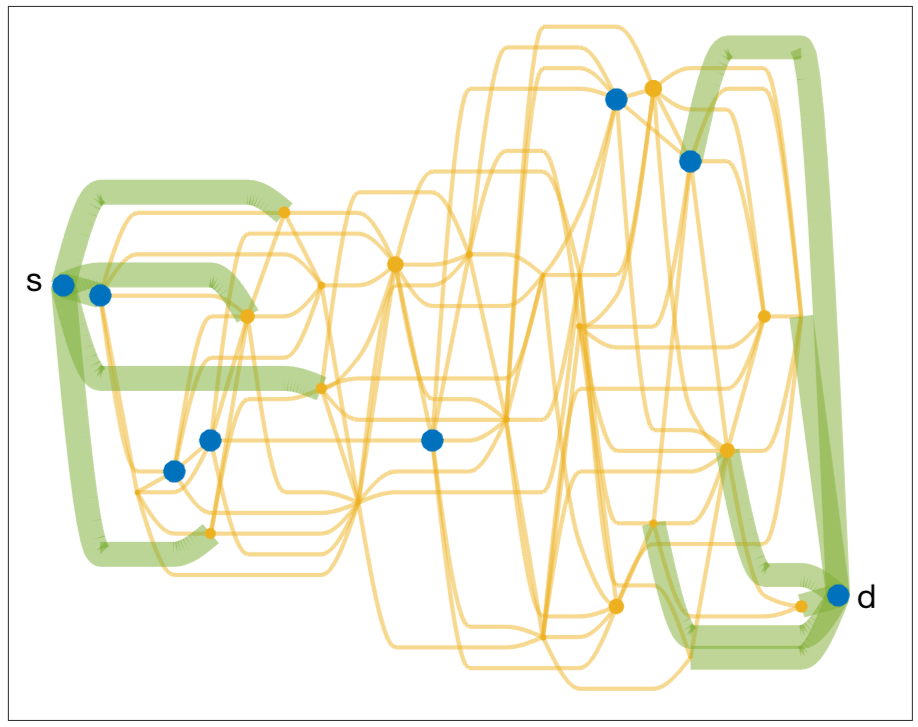}
  %\caption{image1}
  %\label{fig:1}
\end{subfigure}%\hfil % <-- added
\begin{subfigure}{0.33\textwidth}
  \includegraphics[width=\linewidth]{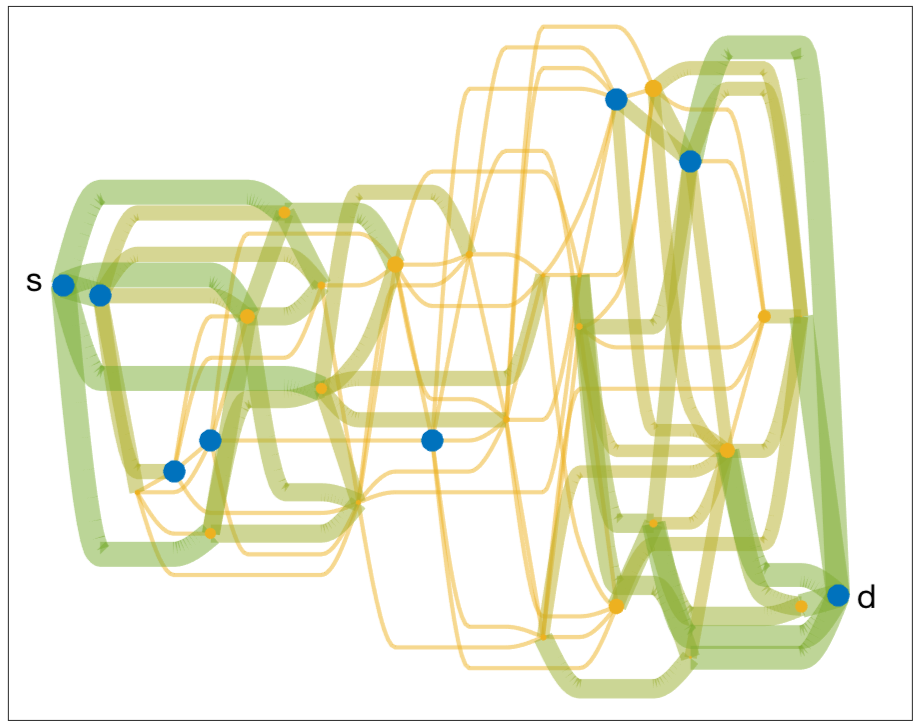}
  %\caption{image2}
  %\label{fig:2}
\end{subfigure}%\hfil % <-- added
\begin{subfigure}{0.33\textwidth}
  \includegraphics[width=\linewidth]{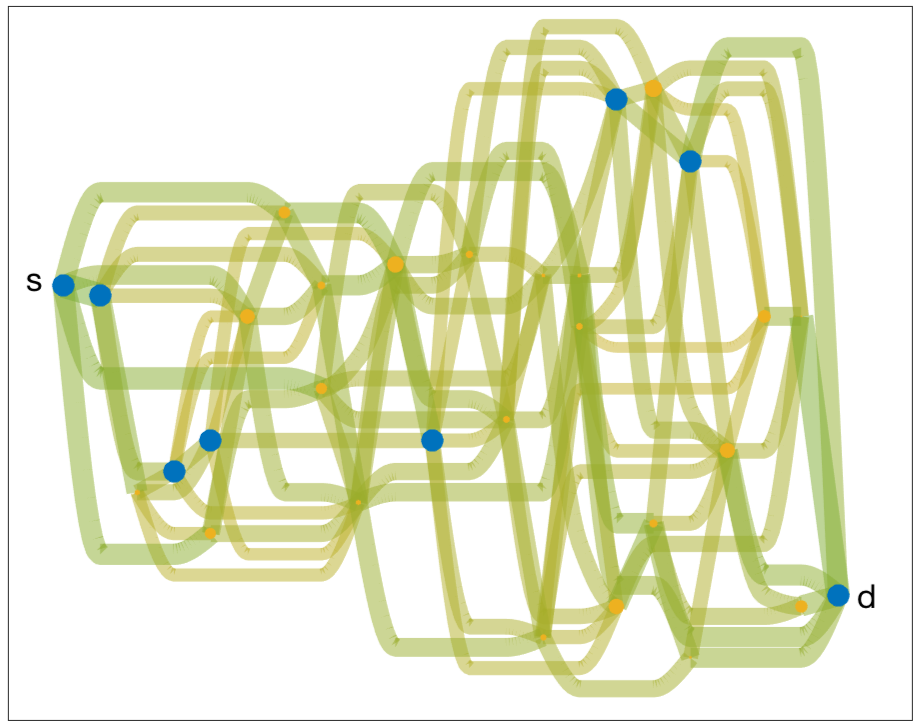}
  %\caption{image3}
  %\label{fig:3}
\end{subfigure}
%\medskip
\begin{subfigure}{0.33\textwidth}
  \includegraphics[width=\linewidth]{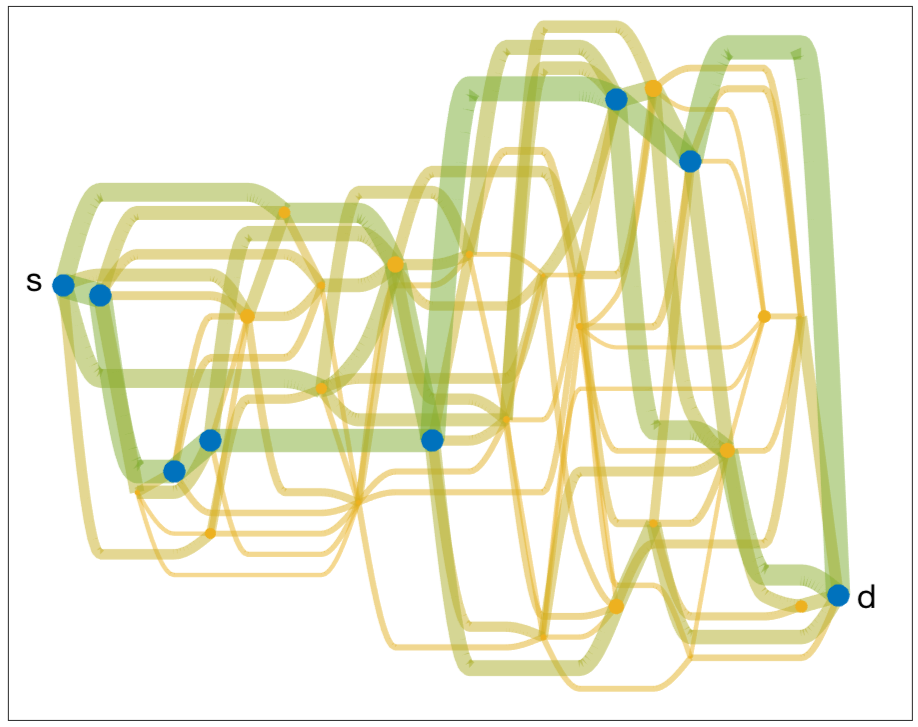}
  %\caption{image4}
  %\label{fig:4}
\end{subfigure}%\hfil % <-- added
\begin{subfigure}{0.33\textwidth}
  \includegraphics[width=\linewidth]{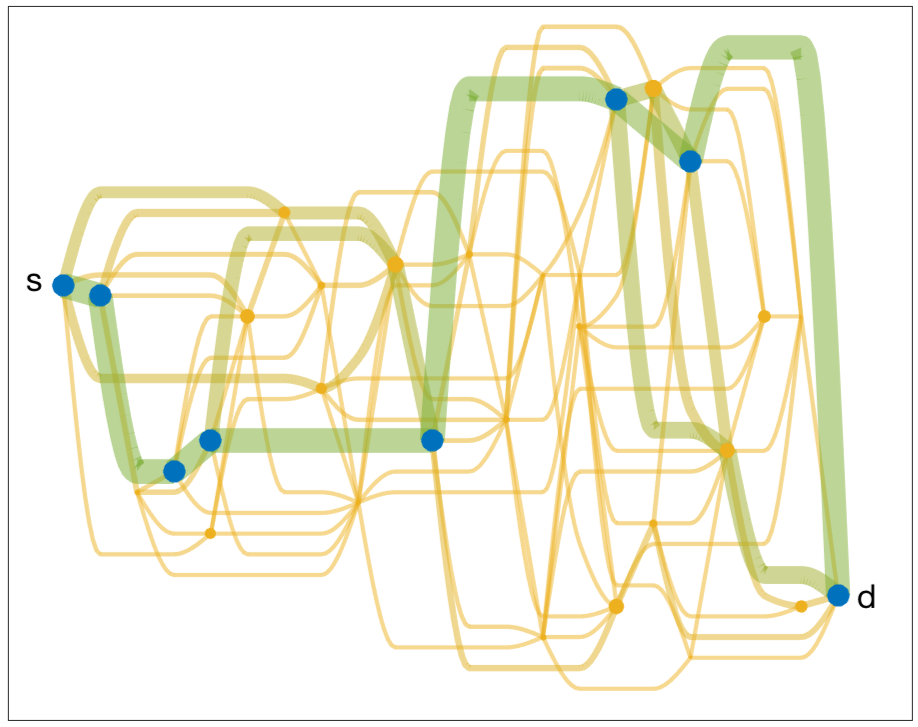}
  %\caption{image5}
  %\label{fig:5}
\end{subfigure}%\hfil % <-- added
\begin{subfigure}{0.33\textwidth}
  \includegraphics[width=\linewidth]{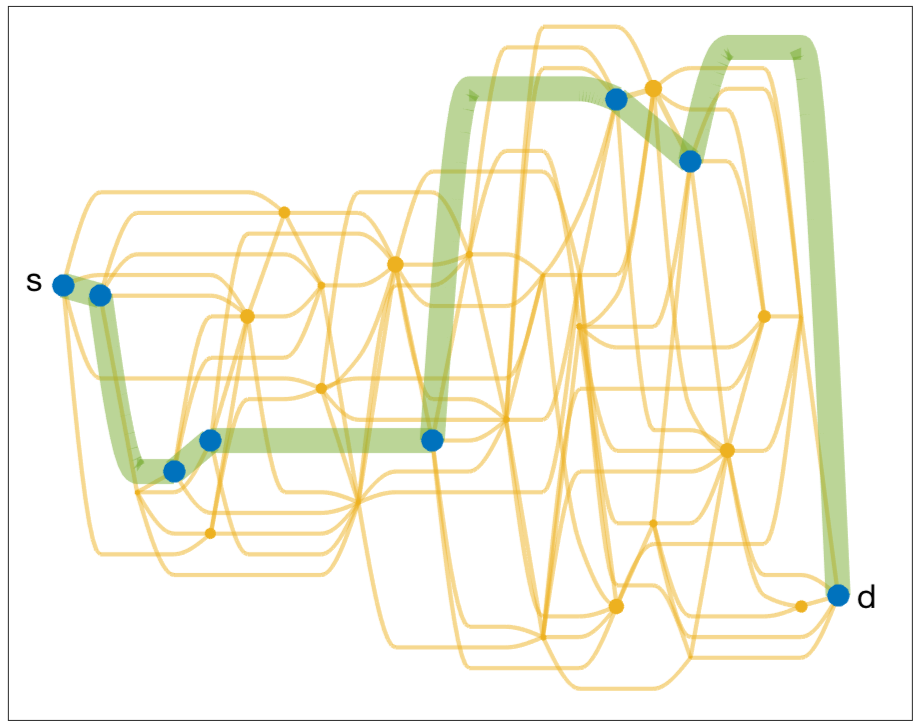}
  %\caption{image6}
  %\label{fig:6}
\end{subfigure}
\caption{Flow dynamics governed by the linear decision rule converges to the path with the minimum leakage (shown by blue vertices) when the incoming flow does not change with time. Larger dots represent vertices with smaller leakage, and thickness of the green edges corresponds to the flow level. %Notice that the paths with smaller leakage get pruned later.
\vspace{-8pt}
}
\label{fig:simulation}
\end{figure}

We complement our simulations on general graph models
with a provable convergence result  for graphs with two parallel paths, a case that has been experimentally investigated in the past \autocite{goss1989self, dorigo2001experimental}.

\begin{restatable}{theorem}{firstthm}
\label{thm:ub-leakage}
Consider a graph $G$ consisting of two parallel paths $P_1$ and $P_2$ from $s$ to $d$. Let the flow and the pheromone levels be updated according to the model in Section \ref{sec:model}, and let $P_1 $ be the path with the minimum leakage.  If (i) the incoming  flow values $\fflow{s}(t)$ and $\bflow{d}(t)$ are non-zero and unchanging with time, and (ii) the initial pheromone level ${p}_{uv}(0)$ is positive for all edges $(u, v) \in P_1$,  then the flow dynamics governed by the linear decision rule converges to a state where all the flow goes through $P_1$. %That is, $\nfp{uv}(t) \geq 1 - \epsilon$ and $\nbp{uv}(t) \geq 1 - \epsilon$, for all $t = \Omega\left(log \left(\frac{1}{\epsilon}\right)\right)$ and $(u, v) \in P_1$.
\end{restatable}

At a high level, the proof involves showing that relatively more pheromone accumulates on the path with less leakage as time progresses.
Although the update rules governing our model are simple, mathematically understanding its dynamics is surprisingly non-trivial. 
Even for the seemingly simple case of two parallel paths, the progression of pheromone levels can be highly non-monotone, and the proof needs a careful construction of an appropriate potential function. We give a sketch of the proof in Section \ref{sec:proof_overview}.

\replaced{To connect this result to the observation of \textcite{chandrasekhar2019better} that ants form trails with approximately the minimum number of vertices, we need to connect leakage to the number of vertices.
Note that the path with minimum leakage is also the path with the minimum number of vertices  when all the vertices have equal leakage. Moreover,  we can show that 
this connection between leakage and number of vertices degrades gracefully, and as long as the variation in leakage between different vertices is not too large, the path with the minimum leakage has approximately the minimum number of vertices. One way to formalize this is to assume that  for any pair of vertices $u$ and $v$, $\log(1 - l_u)$  and $\log(1-l_v)$ are within a $(1+\epsilon)$ factor of each other. Then we can show that the path with the minimum leakage has number of vertices at most $(1+\epsilon)$ times the path with the minimum number of vertices (see supplementary material Appendix \ref{app:connect_leakage_vertices} for the proof).
Thus, our result on convergence to the minimum leakage path suggests a plausible way in which ants can converge to the path with approximately the minimum number of vertices.}{Note that the path with minimum leakage is also the path with the minimum number of vertices  when all the vertices have equal leakage. Thus, this result suggests a plausible way in which ants converge to the path with the minimum number of vertices, as observed in \autocite{chandrasekhar2019better}.}

\added{In the case when there is a large variation in leakage between different vertices, the path with the minimum leakage may not have approximately the minimum number of vertices. However, even in this case, convergence to the minimum leakage path is consistent with the hypothesis of \textcite{chandrasekhar2019better} that turtle ants prefer paths that minimize their chances of getting lost.}

\paragraph{Increasing flow with time.}For the previous result, we assumed that the incoming flow does not change with time. Now we consider the effect of change in flow. We show that if the incoming forward and backward flow increases with time, in the absence of leakage, the dynamics governed by the linear decision rule converges to the shortest path.

We run the simulations for the same families of graphs considered for the last result. The differences in these set of simulations are the absence of leakage, and the incoming  flow increases by a fixed factor in each step. We generate a large number of graph instances with different values of the parameters, and observe convergence to the shortest path in \emph{all} the simulated instances. %More details about the simulations can be found in Appendix \ref{sec:app_sim_details}.

We also show provable convergence to the shortest path for graphs consisting of two parallel paths.

\begin{restatable}{theorem}{secondthm}
\label{thm:ub-changing-flow}
Consider a graph $G$ consisting of two parallel paths $P_1$ and $P_2$ from $s$ to $d$. Let the flow and the pheromone levels be updated according to the model in Section \ref{sec:model}, and let $P_1$ be the shorter path.   If  (i)  the initial pheromone level ${p}_{uv}(0)$ is positive for all edges $(u, v) \in P_1$,
(ii) leakage $l_{P_1} = l_{P_2} = 0$, and (iii) the incoming flow increases as follows: \begin{enumerate}
    \item \textbf{Multiplicative increase:} $\fflow{s}(t) = \alpha^t \fflow{s}(0)$ and $\fflow{d}(t) = \alpha^t \fflow{d}(0)$, for any  $\alpha > 1$,  or
       \item \textbf{Additive increase:} $\fflow{s}(t) = \fflow{s}(0) + \alpha t$ and $\fflow{d}(t) = \fflow{d}(0) +\alpha t$, for any  $\alpha > 0$,  
       
       \end{enumerate}
       then the flow dynamics governed by the linear decision rule converges to a state where all the flow goes through $P_1$ .
\end{restatable}
\added{For ease of analysis, we consider only the cases when the flow increases by a fixed additive or multiplicative factor. Our analysis suggests that the outcome may be similar when the rate of increase is not fixed; further work is needed to demonstrate this.}

For an intuitive explanation of this result, consider a graph with two parallel paths as shown in Figure \ref{fig:path_graph} such that path $P_1$ is shorter than $P_2$, and there is no leakage. Consider the forward flow on edges $(d_1, d)$ and $(d_2, d)$. Since $P_1$ is shorter, the forward flow on $(d_1, d)$ corresponds to the more recent forward flow that entered from $s$ compared to the forward flow on $(d_2, d)$. Since the flow is increasing with time, more recent flow is larger ensuring that relatively more pheromone accumulates on $(d_1, d)$ than $(d_2, d)$ as time progresses.
Similarly, relatively more pheromone accumulates on $(s, s_1)$ than $(s, s_2)$ as time progresses due to the increasing backward flow. Thus, as time progresses, relatively more pheromone accumulates on $P_1$. However, as in the case with leakage, increase in relative pheromone levels on $P_1$ is not monotone and we require a more careful proof. We give a sketch of the proof in Section \ref{sec:proof_overview}.

Previous studies have shown that ants are capable of finding the shortest path in certain simple graphs \autocite{goss1989self}, and that factors such as detection of new food sources can positively reinforce the rate of flow of ants \autocite{gordon2012dynamics, deneubourg1986random}. Our result shows a surprising connection between these two 
phenomena.

The main goal of our work is to demonstrate the intriguing connection between leakage and flow rate and the shortest path problem. 
 We do this using simulations on general graph models, and analysis on graphs with parallel paths. Based on our simulations, we conjecture that the above results showing provable convergence hold for general graphs.
\begin{conjecture}
\label{conjecture1}
The provable convergence results in Theorem \ref{thm:ub-leakage} and \ref{thm:ub-changing-flow} hold for general graphs.
\end{conjecture}

We can view leakage and increasing flow as two possibly conflicting forces, leading to paths with minimum leakage and length respectively. An interesting direction for future work would be to investigate the dynamics when both these forces are active simultaneously.

\subsection{General Rules and Fundamental Limits}
\label{sec:general_rules_lb}

In the previous subsection, we show that  bidirectional flow with linear decision rule converges to the path with the minimum leakage when the flow is fixed, and to the shortest path when the flow is increasing and there is no leakage. How crucial is the bidirectional nature of the flow? How does the dynamics behave when the decision rule is non-linear? Next, we study these questions.

\textbf{Necessity of bi-directional flow.}
We show that bi-directional flow is  necessary to find the shortest or the minimum leakage path. This result holds independent of the decision rule, leakage levels and the change in flow rate; we formally state this result next. Let $\mathcal{G}$ be the set of all decision rules that  distribute the forward (backward) flow at any vertex, onto the outgoing (incoming) edges only based on their pheromone levels.

\begin{restatable}{theorem}{firstprop}
\label{prop:one}
Consider any graph $G$ with two parallel paths between $s$ and $d$. Let there be unidirectional flow from $s$ to $d$. For any decision rule in $\mathcal{G}$, any setting of leakage parameters and with arbitrary incoming flow levels, there exists a setting of initial pheromone levels, such that the dynamics does not converge to the shortest or the minimum leakage path.
\end{restatable}
Thus, bi-directional flow is necessary for all pheromone based rules to guarantee convergence to the shortest or the minimum leakage path. The main idea behind this result is that with only unidirectional flow from $s$ to $d$, pheromone on edges incident on $s$ can not encode information about rest of the graph.

\newcommand{\ft}{f{(t)}}
\newcommand{\pt}{p{(t)}}
\newcommand{\pbt}{p{(t)}}
\newcommand{\bt}{b{(t)}}
\newcommand{\fz}{f{(0)}}
\newcommand{\pfz}{p{(0)}}
\newcommand{\pbz}{p{(0)}}
\newcommand{\bz}{b{(0)}}

\newcommand{\ftt}{f}
\newcommand{\pftt}{p}
\newcommand{\pbtt}{p}
\newcommand{\btt}{b}

\textbf{The conflict between leakage, change in flow and non-linearity.}
To understand the dynamics for other decision rules beyond the linear rule, we consider a  family of decision rules $\mathcal{F}$ satisfying certain assumptions. This family is a subset of the family $\mathcal{G}$ discussed above. These rules distribute the forward flow among the outgoing edges of any vertex, and backward flow among the the incoming edges based on the normalized pheromone levels at these edges (as defined below). That is, allocation of flow does not depend on the absolute pheromone levels. 
We also assume that these rules are monotone in the sense that increasing the normalized pheromone level at any edge does not decrease the proportion of flow entering that edge. 
Most decision rules considered in the past \autocite{chandrasekhar2018distributed} satisfy these conditions.

Here, the normalized pheromone level is defined as follows:
\begin{equation}
\nfp{uv}(t) = \frac{p_{uv}(t)}{\sum\limits_{z: (u, z) \in E} p_{uz}(t)} \ \text{, } \
\nbp{uv}(t) = \frac{p_{uv}(t)}{\sum\limits_{z: (z, v) \in E} p_{zv}(t)}
\end{equation}
Note that while the pheromone level on an edge has no associated direction, there is a forward and backward normalized pheromone level for each edge depending on whether the normalization is done with respect to the incoming edges or the outgoing edges.

More formally, we consider the family of rules  given by functions 
$g:[0,1/2] \rightarrow [0,1] ~ \in \setflb$.  For any vertex with out-degree (in-degree) $2$,   a rule in this family uses a function $g$, which takes as input the minimum of the  normalized pheromone levels on the two outgoing (incoming)  edges, and returns the fraction of the forward (backward) flow on this edge. In other words, $g(x)$ is the fraction of flow sent on edge with normalized pheromone level $x$, for $x < 0.5$. We assume that  $g$ is monotonically increasing, that is, $g(x) \leq g(x')$ for all  $x \leq x'$. Further, $g$ satisfies $g(0)=0$ and $g(1/2)=1/2$. This condition says that if one of the edges has $0$ normalized pheromone, there is no flow on it, and if both the edges have equal pheromone level, there is equal flow on them. For any vertex with out-degree (in-degree) $1$, the forward (backward) flow goes to the next (previous) vertex, just as in the case of the linear rule. For our results below, we need to define these rules for only degree $1$ and $2$ vertices.

\begin{comment}
\begin{itemize}
    \item The rule $g$ is symmetric, that is, the same rule apply to both $s$ and $d$.
    \item The rule $g:[0,1/2] \rightarrow [0,1]$ is pheromone ratio based, meaning at all time steps it takes the minimum of the two normalized pheromone levels on edges belonging to the two paths and returns the fraction of the flow on the minimum edge.
    \item The rule $g$ is monotonically non-decreasing, meaning $g(x) \leq g(x')$ for all $x,x' \in [0,1/2]$ and $x \leq x'$. Further, $g$ satisfies $g(0)=0$ and $g(1/2)=1/2$.
\end{itemize}
\end{comment}

Note that the linear decision rule studied in previous subsection belongs to $\setflb$, and corresponds to $g(x) = x$.

We study the dynamics when the decision rule belongs to the family of decision rules defined above
 and is non-linear. We show that for every non-linear rule $g \in \setflb$, there exists a setting of leakage parameters in which the dynamics fails to converge to the path with the minimum leakage.

\begin{restatable}{theorem}{secondprop}
\label{prop:two}
Consider any graph $G$ consisting of two parallel paths from $s$ to $d$. When the incoming flow is fixed, for every non-linear decision rule $g \in \setflb$, there exists a setting of leakage parameters and initial pheromone and flow levels dependent on $g$, such that the dynamics does not converge to the path with the minimum leakage. 
 
\end{restatable}
We show an analogous result for the increasing flow case with no leakage.
\begin{restatable}{theorem}{thirdprop}
\label{prop:nl-incflow}
Consider  any graph $G$ consisting of two parallel paths from $s$ to $d$ with a unique shortest path. When the leakage is zero for all the vertices, for every non-linear decision rule $g \in \setflb$, there exists a setting of initial pheromone and flow levels, with   incoming flow increasing by a fixed multiplicative factor at each time step, such that the dynamics does not converge to the shortest path. The multiplicative factor and initial pheromone and flow levels are chosen as a function of $g$. 

\end{restatable}

For an intuitive explanation for these results, consider the quadratic decision rule that distributes the flow in proportion to the square of the pheromone levels. Due to the square in the quadratic decision rule, given two edges incident on a vertex, this rule sends more than linearly proportional flow on the edge with the higher pheromone. Thus, if a path---that is not necessarily the shortest or the minimum leakage path---has relatively high pheromone initially, even more pheromone accumulates on it as time progresses, and the dynamics may not converge to the shortest or the minimum leakage path. Theorem \ref{prop:two} and \ref{prop:nl-incflow} formalize this intuition for any non-linear decision rule belonging to the family $\mathcal{F}$. 

From the last subsection, we can view  increasing flow and leakage as two conflicting forces, preferring the shortest and the minimum leakage path respectively.  The above results suggest that non-linearity in the decision rule can be viewed as another force, in conflict with the forces of leakage and increasing flow, preferring certain states that may not correspond to the shortest or the minimum leakage path.

\begin{figure}[h]
    \centering
    \begin{subfigure}[b]{0.485\textwidth}
        \centering
        \includegraphics[width=\textwidth]{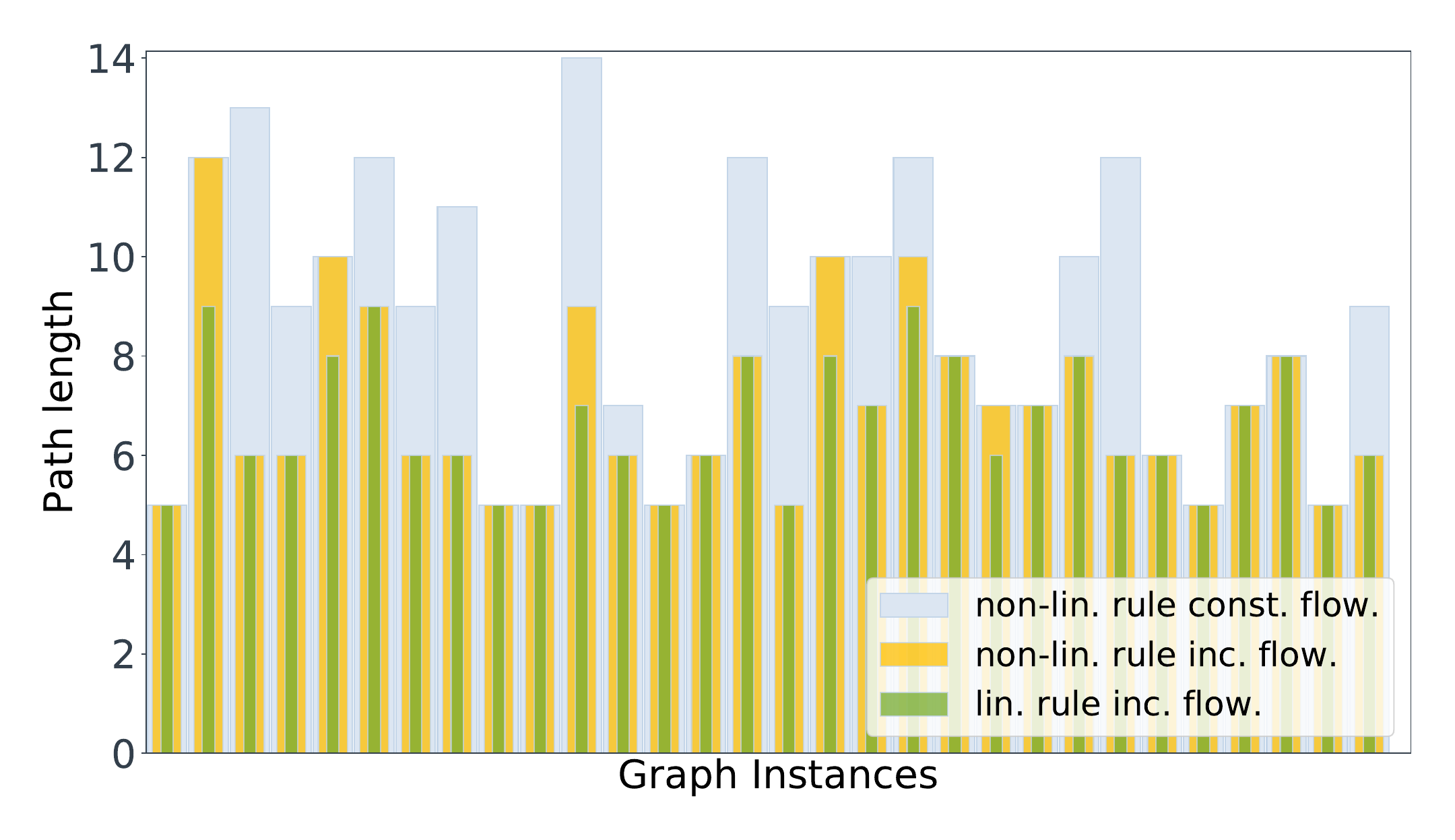}
        \caption{}
        \label{fig:non-lin1}
    \end{subfigure}
    \quad
    \begin{subfigure}[b]{0.485\textwidth}  
        \centering 
        \includegraphics[width=\textwidth]{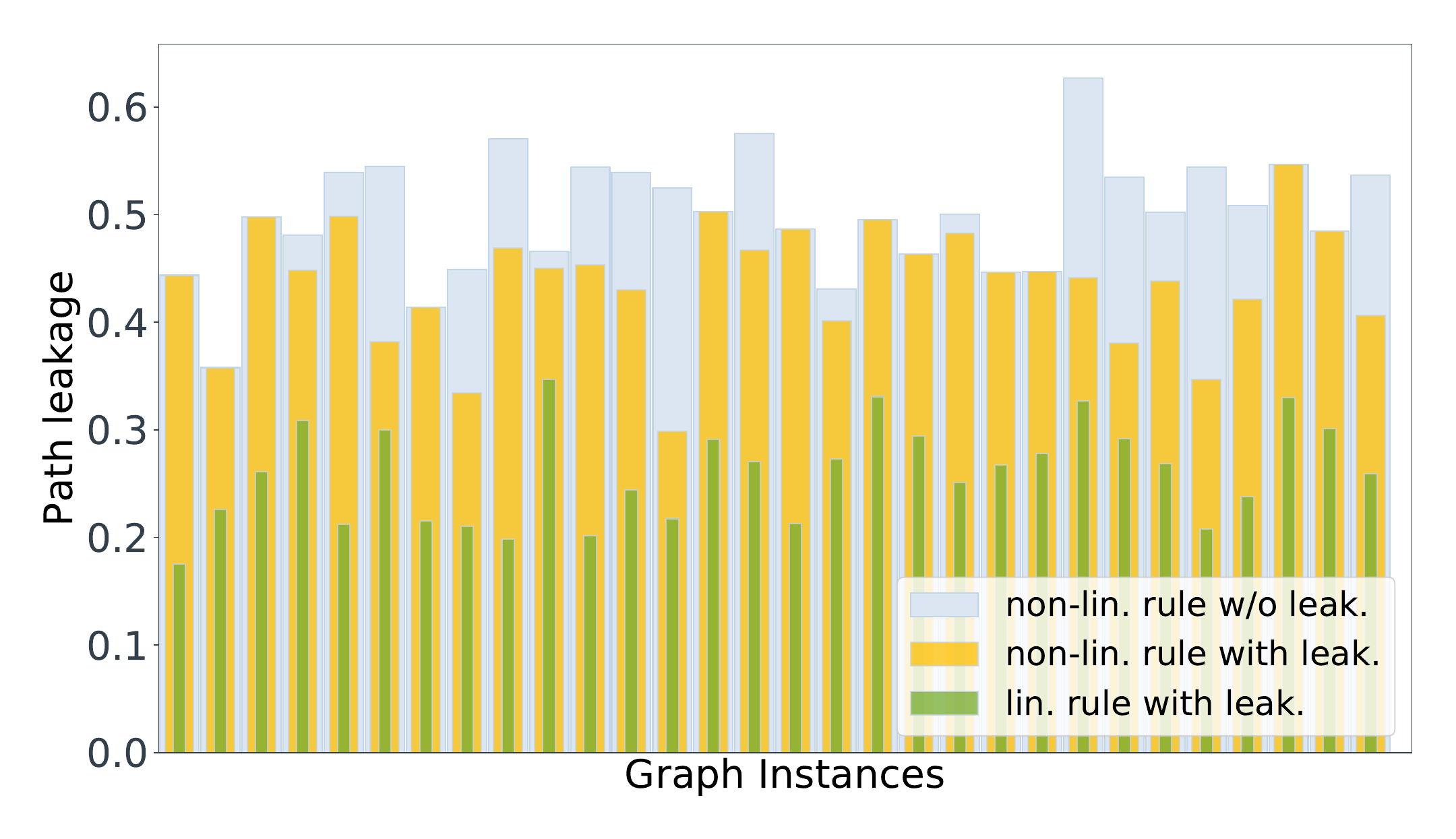}
        \caption{}
        \label{fig:non-lin2}
    \end{subfigure}
    \caption{Effect of leakage and increasing flow with a non-linear decision rule. (a) Path length obtained  with a non-linear decision rule with and without increasing flow, and the linear decision rule with increasing flow. (b) Path leakage obtained with a non-linear decision rule with and without leakage present at vertices, and linear decision rule with leakage present. \vspace{-5pt}}
    \label{fig:non-lin}
\end{figure}

\vspace{-8pt}
\paragraph{Usefulness of increasing flow and leakage not limited to the linear decision rule. }
Note that the above results only suggest that linear decision rule is necessary for \emph{guaranteed convergence} to the shortest or the minimum leakage path. Can it be the case that even with some non-linear decision rules, the forces of increasing flow and leakage still help in finding shorter or smaller leakage paths respectively, compared to the paths found in the absence of these forces? To understand this question, we ran simulations for various non-linear decision rules, some of which have been previously used to model ant behaviour \cite{chandrasekhar2018distributed, deneubourg1990self}.
We observe that within each graph family, for a large fraction of the graph instances, the path found in the presence of these forces has length (respectively leakage) smaller than or equal to the length (respectively leakage) of the path found in their absence. 

Figure \ref{fig:non-lin} demonstrates this for the quadratic decision rule for $G(n, p)$ graphs with the locality constraint. This rule  divides the flow in proportion to the square of the  pheromone levels. 

% Figure \ref{fig:non-lin} demonstrates this for the quadratic-with-offset decision rule \cite{ deneubourg1990self} for $G(n, p)$ graphs with the locality constraint. This rule adds a fixed positive constant to the pheromone levels, and divides the flow in proportion to the square of the offsetted pheromone levels. 

In Figure \ref{fig:non-lin1}, we show the path length obtained by the quadratic  decision rule with and without increase in flow, and by the linear rule with increase in flow, for 30 random graph instances. As discussed before, the linear rule finds the shortest path. But even with the quadratic rule, the path length obtained in the presence of increasing flow is smaller than or equal to the length obtained in its absence.

Similarly, Figure \ref{fig:non-lin2} shows the path leakage obtained by the quadratic decision rule with and without leakage present at the vertices, and by the linear decision rule with leakage present. There is a subtle distinction here between path leakage as an objective function and leakage as a process affecting the dynamics. For each graph instance, we assign leakage values to vertices (see supplementary material Appendix \ref{sec:app_sim_details} for details). This gives us a path leakage objective function which we measure in all the three cases. However, in the case of the quadratic rule without leakage, the leakage process is not applied at vertices during the dynamics. This gives a baseline to which we compare the path leakage objective when the leakage process is applied.
We observe that the quadratic rule with leakage applied leads to path leakage objective smaller than or equal to the baseline, while the linear rule with leakage applied minimizes the objective as discussed in Section \ref{sec:parallel_paths_ub}.

We observe similar results for other graph families and non-linear decision rules. However, the extent to which the forces of increasing flow and leakage are effective varies with the non-linear decision rule and graph family. 
For instance, we observe that compared to the quadratic rule, these forces are more effective for a non-linear rule closer to the linear rule, dividing the flow in proportion to the $1.1^{\text{th}}$ power of the pheromone levels.
Also, for most graph families and non-linear decision rules considered, 
there is a small fraction of instances where these forces end up increasing the path length (respectively path leakage).  Nonetheless, for \emph{all} the graph families and non-linear decision rules considered, for most (> 80\%) graph instances, the path found in the presence of these forces has length (respectively leakage) smaller than or equal to the length (respectively leakage) of the path found in their absence. Thus the usefulness of the forces of leakage and increasing flow is not limited to the linear decision rule. We include more details and discussion of these simulations in supplementary material Appendix \ref{sec:app_sim_details}.

\section{Proof Ideas}
\label{sec:proof_overview}

\subsection{Linear Decision Rule with Increasing Flow}
For a graph consisting of parallel paths $P_1$ and $P_2$, with $len_{P_1} < len_{P_2}$ (see Figure \ref{fig:path_graph}),  \Cref{thm:ub-changing-flow} says that the dynamics converges to $P_1$ when the incoming flow increases multiplicatively or additively with time. To prove this, we show that as time progresses, relatively more pheromone is accumulated on $P_1$ than $P_2$. 

Let $s_1$ and $s_2$ be the neighboring vertices of $s$, and $d_1$ and $d_2$ be the neighboring vertices of $d$, on path $P_1$ and $P_2$ respectively. 
Note that only the pheromone level on edges $(s, s_1), (s, s_2), (d_1, d), (d_2, d)$ affects the dynamics for this graph.   Consider the ratio of pheromone levels on $(s, s_1)$ and $(s, s_2)$ at time $(t+1)$:
\begin{align}
    \frac{p_{ss_1}(t+1)}{p_{ss_2}(t+1)} &= \frac{\delta(p_{ss_1}(t) + \fflow{ss_1}(t)  + \bflow{ss_1}(t))}{\delta(p_{ss_2}(t) + \fflow{ss_2}(t)  + \bflow{ss_2}(t))}\\
    &= \frac{p_{ss_1}(t) + \fflow{ss_1}(t)  + (1 - l_{P_1})\bflow{d_1d}(t- len_{P_1}+1)}{p_{ss_2}(t ) + \fflow{ss_2}(t)  + (1 - l_{P_2})\bflow{d_2d}(t - len_{P_2}+1)} \label{eq:proof_idea_int_8}\\
    &= \frac{p_{ss_1}(t) + \fflow{ss_1}(t)  + \bflow{d}(t- len_{P_1}+1) \nbp{d_1d}(t - len_{P_1}+1)}{p_{ss_2}(t ) + \fflow{ss_2}(t)  + \bflow{d}(t - len_{P_2}+1)\nbp{d_2d}(t - len_{P_2}+1)}
\end{align}
where in the last equation, we set leakage $l_{P_1} = l_{P_2} = 0$, and write the backward flow in terms of the normalized pheromone level. For simplicity, let us assume $\bflow{d}(t) = \alpha^t$, for some $\alpha > 1$. This gives
\begin{align}
    \frac{p_{ss_1}(t+1)}{p_{ss_2}(t+1)}
    &= \frac{p_{ss_1}(t) + \fflow{ss_1}(t)  + \alpha^{(t - len_{P_1}+1)} \nbp{d_1d}(t - len_{P_1}+1)}{p_{ss_2}(t ) + \fflow{ss_2}(t)  + \alpha^{(t - len_{P_2}+1)}\nbp{d_2d}(t - len_{P_2}+1)}
\end{align}
As $len_{P_1} < len_{P_2}$, we know $\alpha^{(t - len_{P_1}+1)} > \alpha^{(t - len_{P_2}+1)}$. These backward flow terms, $\alpha^{(t - len_{P_1}+1)}$ and $\alpha^{(t - len_{P_2}+1)}$, are the main reason why relatively more pheromone accumulates on $ss_1$ compared to $ss_2$ as time progresses. However, the ratio $\frac{p_{ss_1}(t)}{p_{ss_2}(t)}$ may not increase monotonically at each time step.
%, which makes analysis of progression of pheromone ratio tricky.  
To circumvent this issue, we carefully construct a potential function which increases monotonically with time.
Our potential function is  given by the minimum of the ratio of the pheromone levels $\frac{p_{ss_1}(t)}{p_{ss_2}(t)}$ and $\frac{p_{d_1d}(t)}{p_{d_2d}(t)}$ across the last $max(len_{P_1}, len_{P_2})$ time steps. Let $r_{ss_1}(t) \defeq \frac{p_{ss_1}(t)}{p_{ss_2}(t)}$, $r_{d_1 d}(t) \defeq \frac{p_{d_1d}(t)}{p_{d_2d}(t)}$, and $L \defeq max(len_{P_1}, len_{P_2})$. Our potential function is given by 
$$r_{min}(t) \defeq min\{r_{ss_1}(t), r_{ss_1}(t-1), \cdots, r_{ss_1}(t-L+1), r_{d_1d}(t), r_{d_1d}(t-1), \cdots, r_{d_1d}(t-L+1)   \}. $$
Using the definition of the linear decision rule, we know that $\frac{\fflow{ss_1}(t)}{\fflow{ss_2}(t)} = \frac{p_{ss_1}(t)}{p_{ss_2}(t)} = r_{ss_1}(t) \geq r_{min}(t)$. 
Further, it can be shown that $\nbp{d_1d}(t - len_{P_1}+1) \geq  \frac{r_{min}(t)  }{1+r_{min}(t)}$, and $\nbp{d_2d}(t - len_{P_2}+1) \leq  \frac{1 }{1+r_{min}(t)}$, which implies $\frac{\nbp{d_1d}(t - len_{P_1}+1)}{\nbp{d_2d}(t - len_{P_2}+1)} \geq r_{min}(t)$. These inequalities give us
\begin{align}
   \frac{p_{ss_1}(t+1)}{p_{ss_2}(t+1)} &\geq r_{min}(t)\frac{p_{ss_2}(t) + \fflow{ss_2}(t)  + \alpha^{(t - len_{P_1}+1)}\nbp{d_2d}(t- len_{P_2}+1)}{p_{ss_2}(t ) + \fflow{ss_2}(t)  + \alpha^{(t - len_{P_2}+1)}\nbp{d_2d}(t - len_{P_2}+1)}
   > r_{min}(t)
\end{align}
where we used $\alpha^{(t - len_{P_1}+1)} > \alpha^{(t - len_{P_2}+1)}$ for the last inequality. This gives us $\frac{p_{ss_1}(t+1)}{p_{ss_2}(t+1)} = r_{ss_1}(t+1) > r_{min}(t)$. Thus the backward flow terms $\alpha^{(t - len_{P_1}+1)}$ and $\alpha^{(t - len_{P_2}+1)}$ ensure that the pheromone ratio at the edges incident on $s$ at time $t+1$ is greater than $r_{min}(t)$, the minimum of the pheromone ratios at the edges incident on $s$ and $d$ across last $L$ time steps.    Similarly, the forward flow ensures $r_{d_1d}(t+1) > r_{min}(t)$. This implies that $r_{min}(t)$ never decreases, and strictly increases every $L$ time steps. We use this to show convergence to the shortest path. 

The proof for the case involving leakage with fixed flow (\Cref{thm:ub-leakage}) is similar and uses the same potential function. The only difference is that in this case, the potential function goes up due to the leakage terms $1-l_{P_1}$ and $1-l_{P_2}$ (Equation \ref{eq:proof_idea_int_8}), instead of the $\alpha^{(t - len_{P_1}+1)}$ and $\alpha^{(t - len_{P_2}+1)}$ terms.

% $\alpha^{(t - len_{P_1}+1)}$ and $\alpha^{(t - len_{P_2}+1)}$ terms would ensure  $\frac{p_{ss_1}(t+1)}{p_{ss_2}(t+1)} > \frac{p_{ss_1}(t+1)}{p_{ss_2}(t+1) }$, if $\frac{\nbp{d_1d}(t- len_{P_1}+1)}{\nbp{d_2d}(t - len_{P_2}+1)} \geq \frac{p_{ss_1}(t)}{p_{ss_2}(t)}$. However, this is not true in general, and the ratio of pheromone levels $p_{ss_1}$ and $p_{ss_2}$ does not increase monotonically. To get past this obstacle, similar to the proof for leakage with fixed flow case, we consider the potential function $r_{min}(t)$, given by the minimum of the ratio of the pheromone levels $\frac{p_{ss_1}(t)}{p_{ss_2}(t)}$ and $\frac{p_{d_1d}(t)}{p_{d_2d}(t)}$ across the last $max(len_{P_1}, len_{P_2})$ time steps. We show that this potential function never decreases and increases sufficiently every $max(len_{P_1}, len_{P_2})$ time steps. 

\subsection{ General Rules and Fundamental Limits}
%Here we provide the proof sketch for all the results stated in \Cref{sec:general_rules_lb}. 

In \Cref{prop:one}, we claim that for any pheromone based rule, bi-directional flow is necessary for convergence to the shortest or the minimum leakage path.  When there is unidirectional flow from $s$ to $d$, the pheromone levels on the edges incident on $s$ is only a function of their initial pheromone levels and  forward flow at $s$. It does not depend on the flow and pheromone levels on the rest of the graph. Therefore, in the case of two parallel paths, for a given decision rule and initial setting of the pheromone levels, if the dynamics converges to a particular path, then they will converge to the other path if we swap the initial pheromone levels on the two edges incident on $s$. Hence, we can always set the initial pheromone levels on the edges incident on $s$, such that the dynamics does not converge to the shortest or the minimum leakage path.

Next,  we provide a proof sketch for \Cref{prop:two}. The proof sketch for \Cref{prop:nl-incflow} is similar. Let $s_1$ and $s_2$ be the neighboring vertices of $s$, and $d_1$ and $d_2$ be the neighboring vertices of $d$, on path $P_1$ and $P_2$ respectively (see Figure \ref{fig:path_graph}). Let $P_1$ be the minimum leakage path.
We would show that the dynamics is not guaranteed to converge to path $P_1$ for non-linear $g \in \setflb$. For any non-linear $g \in \setflb$, consider some $r \in (0, 1/2)$  such that $g(r)\neq r$. Such an $r$ exists because $g$ is non-linear. Consider the two cases: 1) $g(r)< r$, 2) $g(r)>r$. 

Suppose $g$ is such that $g(r)< r$. Consider an instance where initial normalized pheromone levels $\nfp{ss_1}(0)$ and $\nbp{d_1d}(0)$ are at most $r$. And the initial 
flow values on the edges of path $P_1$ and $P_2$ are  at most $r$ and at least $1-r$ respectively. Let the incoming forward and backward flow be equal to $1$ at all times. The decision rule $g$ sends at most $r - c_{g, r}$  amount of forward and backward flow on path $P_1$ (because $g(r)< r$ and $g$ is monotone) and at least $1-r+c_{g, r}$ on path $P_2$, where $c_{g, r}$ is a positive constant dependent on $g$ and $r$. The positive constant $c_{g, r}$ ensures that we can set the leakage levels to be small enough, satisfying $l_{P_1} < l_{P_2}$, such that even after leakage, the flow levels on $P_1$ and $P_2$ remain at most $r$ and at least $1-r$ respectively, throughout the future. And normalized pheromone levels $\nfp{ss_1}$ and $\nbp{d_1d}$  remain at most $r$.
%Now, we can choose the leakage parameters to be small enough so that even after leakage,  the flow levels on path $P_1$ and $P_2$ remain at most $r$ and at least $1-r$ throughout the future and normalized pheromone levels $\nfp{ss_1}$ and $\nbp{d_1d}$  remain at most $r$.  
%Thus, the dynamics never converges to the  minimum leakage path $P_1$. 

Using a similar idea, in the case when $g(r) > r$, we can set the initial flow and pheromone levels and the leakage parameters, such that the normalized pheromone levels $\nfp{ss_2}$ , $\nbp{d_2d}$ and the flow levels on $P_2$ never fall below $r$ and the flow levels on $P_1$ remain at most $1-r$. Thus, the dynamics never converges to the minimum leakage path $P_1$.

\section{Discussion}
Like ant colonies, engineered systems such as molecular robots and swarm computing \autocite{rubenstein2012kilobot, lund2010molecular, brambilla2013swarm, werfel2014designing, hecker2015beyond} involve a large population of individuals lacking central control and equipped with minimal computational resources. Searching for a target \autocite{hoff2013distributed, o2007pervasive}, and in particular, finding the shortest path \autocite{szymanski2006distributed} is a basic task for such systems. Our algorithms based on leakage and increasing flow can also be applied to such swarms of robots equipped with the ability to release and detect pheromone \autocite{na2021bio, 7353405, kurabayashi2009realization, fujisawa2014designing, russell1999ant}, to solve the shortest path problem and its variants. 

Our result on convergence to the shortest path suggests that an ant colony has the ability to discover the shortest path merely by increasing its flow rate. An interesting direction for future research would be to empirically investigate the relationship
between flow rate of ants and path length, and understand whether ant colonies increase flow rates to find short paths.

Our results also open up avenues for further theoretical investigation. While the algorithms designed by humans are often set up so as to be amenable to analysis, nature is not constrained in this way. For seemingly simple models of biological systems, it has thus been notoriously difficult to devise mathematical guarantees on the quality of the solutions produced \autocite{tero2007mathematical, miyaji2007mathematical,miyaji2008physarum, bonifaci2012physarum, chazelle2012natural, chazelle2017challenges, bhattacharyya2013convergence,   chazelle2014convergence, 6375309} .  In our model, analyzing the dynamics is challenging as it involves understanding the progression of pheromone level with time, which is affected by the actions of  a large number of agents (modeled by flow), and can be highly non-monotone even for the simple case of graphs with parallel paths. The behavior of this model in extensive simulations suggests that it should be possible to significantly generalize the results we prove here. In particular, we conjecture that provable convergence to the shortest or the minimum leakage path holds for general graphs (Conjecture \ref{conjecture1}). Another direction for future research is to extend our analysis to the case when multiple terminal vertices are present in the graph, as trail networks in nature usually include many  nests and food sources. 

In summary, our model for how ant trails change over time contributes to the synergistic exchange between biology and computer science, providing a plausible explanation for how turtle ant colonies can find paths that minimize the number of vertices, and suggesting a surprising algorithm for the shortest path discovery, by increasing the flow rate, applicable to distributed engineering systems.

\section*{Acknowledgements}
 Shivam Garg was supported by NSF awards AF-1813049 and AF-1704417, and a Stanford Interdisciplinary Graduate Fellowship. Kirankumar Shiragur was supported by a Stanford Data Science Scholarship and a Dantzig-Lieberman Operations Research Fellowship.  Deborah M. Gordon was supported by the  Templeton Fund.
 Moses Charikar was supported by a Simons Investigator Award, a Google Faculty Research Award and an Amazon Research Award.

\printbibliography
\newpage
\begin{center}
\textbf{\LARGE Supplementary Material}
\end{center}
\vspace{30pt}
\appendix
\newcommand{\lf}{\fflow{s}}
\newcommand{\rf}{\bflow{d}}
\newcommand{\phe}{p_{uv}}
\newcommand{\fluv}{\fflow{uv}}
\newcommand{\flvu}{\bflow{uv}}
\newcommand{\phup}{C_{\lf,\rf,\delta}}
\newcommand{\fsos}{\bflow{ss_1}}
\newcommand{\fsts}{\bflow{ss_2}}
\newcommand{\fsso}{\fflow{ss_1}}
\newcommand{\fsst}{\fflow{ss_2}}
\newcommand{\psso}{p_{ss_1}}
\newcommand{\psst}{p_{ss_2}}
\newcommand{\pddo}{p_{dd_1}}
\newcommand{\pddt}{p_{dd_2}}
\newcommand{\npsso}{\nfp{ss_1}}
\newcommand{\npsst}{\nfp{ss_2}}
\newcommand{\npddo}{\nbp{d_1d}}
\newcommand{\npddt}{\nbp{d_2d}}
\newcommand{\lpo}{l_{P_1}}
\newcommand{\lpt}{l_{P_2}}
\DeclarePairedDelimiter\ceil{\lceil}{\rceil}
\DeclarePairedDelimiter\floor{\lfloor}{\rfloor}

\section{Proof of Convergence}
\label{sec:appendix-ub-proofs}
Here we provide proof for all the results stated in \Cref{sec:parallel_paths_ub}.
\subsection{Fixed flow with different leakage on each path (\Cref{thm:ub-leakage})}
Here, we provide a proof of \Cref{thm:ub-leakage}. We restate it below.
% \begin{theorem*}
% Consider a graph $G$ consisting of two parallel paths $P_1$ and $P_2$ from $s$ to $d$. Let $P_1 $ be the path with the minimum leakage such that $l_{P_1} < l_{P_2}$.  If (1) the incoming forward and backward flow values $\fflow{s}(t)$ and $\bflow{d}(t)$ are non-zero and unchanging with time, and (2) the initial pheromone level ${p}_{uv}(0)$ is positive for any edge $(u, v) \in P_1$, then the flow dynamics governed by the linear decision rule converges to a state in which the forward and backward normalized pheromone levels for any edge $(u, v) \in P_1$ tends to $1$. That is, $\nfp{uv}(t) \geq 1 - \epsilon$ and $\nbp{uv}(t) \geq 1 - \epsilon$, for all $t = \Omega\left(log \left(\frac{1}{\epsilon}\right)\right)$ and $(u, v) \in P_1$.
% \end{theorem*}
\firstthm*
 Let $s_1$, $s_2$ be the neighboring vertices of $s$ that belong to paths $P_1$ and $P_2$ respectively. Similarly let $d_1$ and $d_2$ be the corresponding neighbors for $d$. Let $\lf,\rf >0$ be some fixed forward and backward flow values, such that $\fflow{s}(t) = \lf $ and $\bflow{d}(t) = \rf $ for all $t$.
 Define $r_{ss_1}(t) \defeq \frac{p_{ss_1}(t)}{p_{ss_2}(t)}$ and $r_{d_1 d}(t) \defeq \frac{p_{d_1 d}(t)}{p_{d_2 d}(t)}$ to be the relative pheromone levels at $(s,s_1)$ and $(d_1,d)$ respectively. For notational simplicity, we define $m \defeq len_{P_1}$ and $n \defeq len_{P_2}$ to be the lengths of path $P_1$ and $P_2$,  and $L \defeq max(m, n)$. We also define $\alpha \defeq 1 - l_{P_1}$ and $\beta \defeq 1 - l_{P_2}$.

Our potential function at any time $t \geq L$  is given by the minimum of the relative pheromone levels $r_{ss_1}(t)$ and $r_{d_1d}(t)$ across the last $L$ time steps:
\begin{equation}
    r_{min}(t) \defeq min\{r_{ss_1}(t), r_{ss_1}(t-1), \cdots, r_{ss_1}(t-L+1), r_{d_1d}(t), r_{d_1d}(t-1), \cdots, r_{d_1d}(t-L+1)   \}
\end{equation}

We divide our proof into 3 steps:
\begin{itemize}
    \item \textbf{Step 1}: In this step, we show that $r_{min}(t)$ is non-decreasing at every time step and increases by a factor of $ \gamma(t)$ every $L$ time steps,  for all $t \geq L$. Here, $\gamma(t)$ is some appropriately defined function which is greater than 1 for all $t \geq L$.
    \item \textbf{Step 2}: In this step, we will give a lower bound $\gamma_l > 1$, on $\gamma(t)$, to show that $r_{min}(t)$ increases sufficiently every $L$ time steps. 
    \item \textbf{Step 3}: We will find the rate of convergence based on the rate of increase shown in step 2.
\end{itemize}
\begin{comment}
\textbf{High level idea}: For each flow currently present at any edge of the graph, we associate a number $r$ which is the ratio of pheromone levels when this flow entered the graph. Let $\rmin$ be the minimum value of $r$ among all the flow currently present in the graph. Claim: $\rmin$ is non-decreasing, and it strictly increases every $max(n,m)$ time steps, where $n$ and $m$ are the path lengths.

Let $p_{uv}(t)$ be the flow from node $u$ to $v$ at time $t$, and  let $p_l$ and $\rf$ be fixed flow levels from left and right respectively. Define $r_{uv}(t)$ to be the ratio of pheromone levels at the terminal edges when the flow $p_{uv}(t)$ had entered the graph. For example, let $(u,v)$ be an edge on the path along $(B,D)$, then $r_{uv}(t)$ is the ratio of pheromone levels at $BD$ and $BC$ when the flow $p_{uv}(t)$ entered $BD$, that is,  $p_{uv}(t) = \frac{r_{uv}(t)}{1+r_{uv}(t)} * \rf$. 

Let $\rmin(t) = \min_{(u,v)} r_{uv}(t)$ be the minimum pheromone level ratio among the flow present on all the edges at time $t$. Here, all the edges are directed. We will show that $r_{\text{min}(t)}$ is non-decreasing with $t$ and strictly increases every $max(m,n)$ time steps.
\end{comment}
Now, we give proofs for each of these three steps.

\paragraph{Step 1.}
We will use the following lemma for the proof.
\begin{lemma}\label{lem:helper_min_ratio}
For any time $t \geq L$, the following is true,
$$\frac{ \bflow{d_1d}(t-m+1)}{ \bflow{d_2d}(t-n+1)} \geq \min(r_{d_1d}(t-m+1), r_{d_1 d}(t-n+1)) 
\geq r_{min}(t)~.$$
\end{lemma}
\begin{proof}
Under the linear decision rule we know that,
\begin{align}
 \bflow{d_1d}(t-m+1) &=   \frac{\rf r_{d_1 d}(t-m+1)}{r_{d_1 d}(t-m+1)+1} \\
 \bflow{d_2d}(t-n+1) &=  \frac{\rf}{r_{d_1 d}(t-n+1)+1} 
\end{align}
Let $r\defeq \min(r_{d_1d}(t-m+1), r_{d_1d}(t-n+1))$, then note that $\frac{ r_{d_1d}(t-m+1)}{r_{d_1d}(t-m+1)+1} \geq \frac{r}{r+1}$ and $\frac{1}{r_{d_1d}(t-n+1)+1} \leq \frac{1}{r+1}$. Therefore,
$$\frac{ \bflow{d_1d}(t-m+1)}{ \bflow{d_2d}(t-n+1)} \geq  r=\min(r_{d_1 d}(t-m+1), r_{d_1d}(t-n+1))~.$$
and from the definition of $r_{min}(t)$, we know that $\min(r_{d_1 d}(t-m+1), r_{d_1 d}(t-n+1)) \geq r_{min}(t)$.
\begin{comment}
\begin{itemize}
    \item 
\end{itemize}
Case 1: $r_{d d_1}(t-n+1) > r_{d d_1}(t-m+1)$. In this case, we want to show that 
\begin{align}
\frac{ \fsos(t)}{ \fsts(t)} &\geq \frac{\lpo}{\lpt} r_{d d_1}(t-m+1)\\
\Leftrightarrow \
\frac{\lpo r_{d d_1}(t-n+1) (r_{d d_1}(t-m+1)+1)}{\lpt (r_{d d_1}(t-n+1)+1)} & \geq \frac{\lpo}{\lpt} r_{d d_1}(t-m+1)\\
\Leftrightarrow \
r_{d d_1}(t-n+1) (r_{d d_1}(t-m+1)+1) &\geq r_{d d_1}(t-m+1)(r_{d d_1}(t-n+1)+1) \\
\Leftrightarrow \
r_{d d_1}(t-n+1) &\geq r_{d d_1}(t-m+1)
\end{align}
Case 2: $r_{d d_1}(t-n+1) \leq r_{d d_1}(t-m+1)$. In this case, we want to show that 
\begin{align}
\frac{ \fsos(t)}{ \fsts(t)} &\geq \frac{\lpo}{\lpt} r_{d d_1}(t-n+1)\\
\Leftrightarrow \
\frac{\lpo r_{d d_1}(t-n+1) (r_{d d_1}(t-m+1)+1)}{\lpt (r_{d d_1}(t-n+1)+1)} & \geq \frac{\lpo}{\lpt} r_{d d_1}(t-n+1)\\
\Leftrightarrow \
r_{d d_1}(t-n+1) (r_{d d_1}(t-m+1)+1) &\geq r_{d d_1}(t-n+1)(r_{d d_1}(t-n+1)+1) \\
\Leftrightarrow \
r_{d d_1}(t-m+1) &\geq r_{d d_1}(t-n+1)
\end{align}
As desired, this gives us 
\begin{align}
\frac{ \fsos(t)}{ \fsts(t)} \geq \frac{\lpo}{\lpt}\text{min}(r_{d d_1}(t-n+1), r_{d d_1}(t-m+1))
\end{align}
\end{comment}
\end{proof}
Now we show that $r_{min}(t)$ is non-decreasing at every time step and increases by a factor of $ \gamma(t)$ every $L$ time steps, for all $t \geq L$. To show this, we would show that $r_{ss_1}(t+1) \geq r_{min}(t)\gamma_s(t)$, and $r_{d_1 d}(t+1) \geq r_{min}(t) \gamma_d(t)$.  Here, $\gamma_s(t)$ and $\gamma_d(t)$ are appropriately defined functions such that $\gamma_s(t)) > 1$ and $\gamma_d(t)) > 1$ for all $t \geq L$.

This would give us $r_{min}(t+L) \geq r_{min}(t) \gamma(t)$ where $\gamma(t) = min\{\gamma_s(t+L-1),\gamma_s(t+L-2), \cdots, \gamma_s(t), \gamma_d(t+L-1),\gamma_d(t+L-2), \cdots, \gamma_d(t) \}$.
Below, we give the proof for $r_{ss_1}(t+1) \geq r_{min}(t) \gamma_s(t)$. The proof for  $r_{d_1 d}(t+1) \geq r_{min}(t) \gamma_d(t)$ is  similar.
\begin{lemma}\label{lem:leakage1}
For all time $t \geq L$, $r_{ss_1}(t+1) \geq r_{min}(t) \gamma_s(t)$,
for some appropriately defined  function $\gamma_s(t)$, such that $\gamma_s(t) > 1$ for all $t \geq L$.
\end{lemma}

\begin{proof}

 The pheromone levels on edges $(s,s_1)$ and $(s,s_2)$ at time $t+1$ is provided by the following expressions. 
\begin{align*}
    \psso(t+1) = \delta \left(\psso(t) + 
    \fsso(t) +  \fsos(t)\right), \quad 
    \psst(t+1) = \delta \left(\psst(t) + \fsst(t) +  \fsts(t)\right)~.\\
\end{align*}
Therefore, for $t \geq L$, we get
 \begin{align}
\frac{\psso(t+1)}{\psst(t+1)} &= \frac{ \psso(t) + 
    \fsso(t) +  \fsos(t)}{ \psst(t) + \fsst(t) +  \fsts(t)}\\
    &= \frac{ \psso(t) + 
    \fsso(t) +  (1-l_{P_1}) \ \bflow{d_1d}(t-m+1)}{ \psst(t) + \fsst(t) +  (1 - l_{P_2}) \ \bflow{d_2d}(t-n+1)}\\
    &= \frac{ \psso(t) + 
    \fsso(t) +  \alpha \ \bflow{d_1d}(t-m+1)}{ \psst(t) + \fsst(t) +  \beta \ \bflow{d_2d}(t-n+1)}~.\label{eq:9-s1}
\end{align}
By the definition of $r_{ss_1}(t)$ and linear decision rule we know that,
\begin{align*}
\frac{\psso(t)}{\psst(t)} = \frac{  
    \fsso(t)}{ \fsst(t)} = r_{ss_1}(t)  \geq r_{min}(t)~.
\end{align*}
From \Cref{lem:helper_min_ratio}, we know that
\begin{align*}
    \frac{\bflow{d_1d}(t-m+1)}{\bflow{d_2d}(t-n+1)} \geq r_{min}(t)~.
\end{align*}
For $t \geq L$, define 
\begin{align}
a(t) \defeq \frac{r_{ss_1(t)}}{r_{min}(t)}, \ b(t) \defeq \frac{\bflow{d_1d}(t-m+1)}{\bflow{d_2d}(t-n+1)}\frac{1}{r_{min}(t)}~.
\end{align}
Note that, from the definition of $r_{min}(t)$ and using \Cref{lem:helper_min_ratio}, we know that $a(t) \geq 1$ and $b(t) \geq 1$, for all $t \geq L$.
Now, we can write
\begin{align}
\label{eq:9-s1_1}
\frac{\psso(t)}{\psst(t)} = \frac{
    \fsso(t)}{ \fsst(t)} = a(t) r_{min}(t), \ \frac{\bflow{d_1d}(t-m+1)}{\bflow{d_2d}(t-n+1)} = b(t) r_{min}(t)
\end{align}
Substituting this in \Cref{eq:9-s1}, we get 
\begin{align}
\label{eq:9-s2}
\frac{\psso(t+1)}{\psst(t+1)} &= 
    r_{min}(t)\frac{a(t) \psst(t) + 
    a(t)\fsst(t) +  \alpha b(t)  \ \bflow{d_2d}(t-n+1)}{ \psst(t) + \fsst(t) +  \beta \ \bflow{d_2d}(t-n+1)}\\
    &= r_{min}(t) \gamma_s(t)
\end{align}
where we define 
\begin{align}
    \gamma_s(t) \defeq \frac{a(t) \psst(t) + 
    a(t)\fsst(t) +  \alpha b(t)  \ \bflow{d_2d}(t-n+1)}{ \psst(t) + \fsst(t) +  \beta \ \bflow{d_2d}(t-n+1)}.
\end{align}
Since $a(t) \geq 1$, $b(t) \geq 1$ and $\alpha > \beta$, we get that $\gamma_s(t) > 1$ for all $t \geq L$. 

This completes the proof of the Lemma.
\end{proof}

\paragraph{Step 2.}
In this step, we will give a lower bound $\gamma_l > 1$, on $\gamma(t)$, to show that $r_{min}(t)$ increases sufficiently every $L$ time steps. If the pheromone levels on the edges are too high as compared to the flow, it will take more time for the relative pheromone levels to change. We will first show that there exists a time $T_1$, such that for $t \geq T_1$, the pheromone levels and flow are comparable. Our lower bound $\gamma_l \leq \gamma(t)$ will hold for all $t \geq T_1$.

\begin{lemma}\label{lem:pbound1}
In the flow dynamics governed by the linear decision rule, the pheromone level on any edge $e=(u,v)$ is always  bounded as follows:
\begin{equation*}
 p_{uv}(t) \leq \frac{2(\lf+\rf)}{1-\delta},
\end{equation*}
for all 
\begin{align*}
t \geq T_1 \defeq \max_{(u,v) \in E}\left( \frac{log \left(\frac{p_{uv}(0)}{\lf + \rf}\right)} {log \left(\frac{1}{\delta}\right)}\right)  ~.\quad 
\end{align*}
\end{lemma}
\begin{proof}
Consider any edge $e=(u,v)$ and time $t>0$,
\begin{align*}
    p_{uv}(t) & = \delta(\phe(t-1) + \fflow{uv}(t-1)  + \bflow{uv}(t-1))
     \leq \delta \left(\phe(t-1) + \lf+\rf \right)~,\\
    & \leq (\lf+\rf) \sum_{i=1}^{t}\delta ^{i} + \delta^{t} \phe(0)
\end{align*}
For $t \geq \frac{log \left(\frac{p_{uv}(0)}{\lf + \rf}\right)} {log \left(\frac{1}{\delta}\right)}$, we get \begin{align}
    p_{uv}(t)  \leq \frac{(\lf+\rf)}{1-\delta} + (\lf + \rf)
    \leq \frac{2(\lf+\rf)}{1-\delta}.
\end{align}
Since $T_1 \geq \frac{log \left(\frac{p_{uv}(0)}{\lf + \rf}\right)} {log \left(\frac{1}{\delta}\right)}$, we get that for $t \geq T_1$,
\begin{align}
    p_{uv}(t) 
    \leq \frac{2(\lf+\rf)}{1-\delta}.
\end{align}
\end{proof}

In \Cref{eq:9-s1_1}, we defined $a(t)$ and $b(t)$ such that \begin{align*}
\frac{\psso(t)}{\psst(t)} = \frac{
    \fsso(t)}{ \fsst(t)} = a(t) r_{min}(t), \ \frac{\bflow{d_1d}(t-m+1)}{\bflow{d_2d}(t-n+1)} = b(t) r_{min}(t)
\end{align*} 
Similarly, for $t \geq L$, we define $c(t) \defeq \frac{r_{d_1d}(t-n+1)}{r_{min}(t)}$.  From the definition of $r_{min}(t)$, we know that $c(t) \geq 1$ for all $t \geq L$. We will show a upper bound on $c(t)$ in terms of $b(t)$ which will be useful later. 
\begin{lemma}
\label{lem:helper-s1}
For all $t \geq L$, $1+c(t) r_{min}(t) \leq b(t)(1+r_{min}(t))$.
\end{lemma}
\begin{proof}
We know 
\begin{align}
    r_{d_1 d}(t-n+1) = \frac{\nbp{d_1 d}(t-n+1)}{\nbp{d_2 d}(t-n+1)} = c(t) r_{min}(t)
\end{align}
which gives
\begin{align}
     \frac{1 - \nbp{d_2 d}(t-n+1)}{\nbp{d_2 d}(t-n+1)} = c(t) r_{min}(t)
\end{align}
which gives $\frac{1}{\nbp{d_2d}(t-n+1)} = 1 + c(t)r_{min}(t)$. 
We also know 
\begin{align}
    \frac{\nbp{d_1 d}(t-m+1)}{\nbp{d_2 d}(t-n+1)} = b(t) r_{min}(t)~.
\end{align}
Substituting $\frac{1}{\nbp{d_2d}(t-n+1)} = 1 + c(t)r_{min}(t)$, we get 
\begin{align}
    \nbp{d_1 d}(t-m+1) = \frac{b(t) r_{min}(t)}{1 + c(t)r_{min}(t)}~.
\end{align}
Finally, from the definition of $r_{min}(t)$, we know 
\begin{align}
    \frac{\nbp{d_1 d}(t-m+1)}{\nbp{d_2 d}(t-m+1)} = r_{d_1 d}(t-m+1) \geq r_{min}(t).
\end{align}
Since $\nbp{d_2 d}(t-m+1) = 1 - \nbp{d_1 d}(t-m+1)$, this gives
\begin{align}
    {\nbp{d_1 d}(t-m+1)}  \geq \frac{r_{min}(t)}{r_{min}(t) + 1}.
\end{align}
Substituting the value $\nbp{d_1 d}(t-m+1) = \frac{b(t) r_{min}(t)}{1 + c(t)r_{min}(t)}$, we get 
\begin{align}
 \frac{b(t) r_{min}(t)}{1 + c(t)r_{min}(t)} \geq \frac{r_{min}(t)}{r_{min}(t) + 1}.
\end{align}
which implies
\begin{align}
    b(t)(1+r_{min}(t)) \geq 1+c(t) r_{min}(t)  ~.
\end{align}
This finishes the proof of the Lemma.
\end{proof}
Now, we come to the main part of step 2 where we show $\gamma(t) > \gamma_l$ for all $t \geq T_1 + L$, for some $\gamma_l > 1$. From step 1, we know that $\gamma(t) = min\{\gamma_s(t+L-1),\gamma_s(t+L-2), \cdots, \gamma_s(t), \gamma_d(t+L-1),\gamma_d(t+L-2), \cdots, \gamma_d(t) \}$. Therefore if we show that $\gamma_s(t) \geq \gamma_{s_l}$ and $\gamma_d(t) \geq \gamma_{d_l}$ for all $t \geq T_1 + L$, for some $\gamma_{s_l}, \gamma_{d_l} > 1$, we can set $\gamma_l = min(\gamma_{s_l}, \gamma_{d_l})$, and we will be done.

Below we will prove $\gamma_s(t) \geq \gamma_{s_l}$. The proof for $\gamma_d(t) \geq \gamma_{d_l}$ is similar.

\begin{lemma}\label{lem:leakage2}
For all time $t \geq L+T_1$, $r_{ss_1}(t+1) \geq r_{min}(t) \gamma_s(t)$ where  $\gamma_s(t) \geq \gamma_{s_l}$ for some fixed constant $\gamma_{s_l} > 1$.
\end{lemma}

\begin{proof}
From \Cref{lem:leakage1}, we know that 
\begin{align}
\frac{\psso(t+1)}{\psst(t+1)} &= 
    r_{min}(t)\gamma_s(t)
\end{align}
where 
\begin{align*}
    \gamma_s(t) \defeq \frac{a(t) \psst(t) + 
    a(t)\fsst(t) +  \alpha b(t)  \ \bflow{d_2d}(t-n+1)}{ \psst(t) + \fsst(t) +  \beta \ \bflow{d_2d}(t-n+1)}.
\end{align*}
We want to show that $\gamma_s(t) \geq \gamma_{s_l}$ for some fixed constant $\gamma_{s_l} > 1$.
\begin{comment}
We consider two cases:

\textbf{Case 1:} $a(t) \geq 2$.\\
In this case, 

\begin{align}
    \gamma_s(t) &= 
    \frac{a(t) \psst(t) + 
    a(t)\fsst(t) +  b(t) \alpha \ \bflow{d_2d}(t-n+1)}{ \psst(t) + \fsst(t) +  \beta \ \bflow{d_2d}(t-n+1)}\\ \label{eq:10-s1}
    &\geq \frac{2 \psst(t) + 
    2\fsst(t) +   \alpha \ \bflow{d_2d}(t-n+1)}{ \psst(t) + \fsst(t) +  \beta \ \bflow{d_2d}(t-n+1)}\\
    &\geq    \frac{ \min\left(2, \frac{\alpha}{\beta}\right) \psst(t) + 
     \min\left(2, \frac{\alpha}{\beta}\right)\fsst(t) +     \min\left(2, \frac{\alpha}{\beta}\right) \beta \ \bflow{d_2d}(t-n+1)}{ \psst(t) + \fsst(t) +  \beta \ \bflow{d_2d}(t-n+1)}\\
     &=  \ \min\left(2, \frac{\alpha}{\beta}\right)\\
     &=  \ \min\left(2, \left( 1 + \frac{\alpha - \beta}{\beta} \right) \right)~.
\end{align}
where we used $a(t) \geq 2$, $b(t) \geq 1$ for inequality \ref{eq:10-s1}.

\textbf{Case 2:} $a(t) < 2$.\\
In this case, 
\end{comment}
Using $a(t) \geq 1$ and $b(t) \geq 1$, we get 
\begin{align}
\gamma_s(t) 
    &\geq \frac{ \psst(t) + 
    \fsst(t) +  b(t) \alpha \ \bflow{d_2d}(t-n+1)}{ \psst(t) + \fsst(t) +  b(t)\beta \ \bflow{d_2d}(t-n+1)}\\\label{eq:10-s5}
    &=  1 + \frac{\alpha - \beta}{\frac{ \psst(t)}{b(t)\bflow{d_2d}(t-n+1)} + \frac{ \fsst(t)}{b(t)\bflow{d_2d}(t-n+1)} + \beta} 
\end{align}
 Now, to lower bound $\gamma_s(t)$, we need to upper bound $\frac{ \psst(t)}{b(t)\bflow{d_2d}(t-n+1)}$ and  $\frac{ \fsst(t)}{b(t)\bflow{d_2d}(t-n+1)}$. Using the definition of normalized pheromone level, we can write
\begin{align}
    \frac{ \psst(t)}{b(t)\bflow{d_2d}(t-n+1)} &= \left( \frac{p_{ss_1}(t)+p_{ss_2}(t)}{\bflow{d}} \right) \left( \frac{ \nfp{ss_2}(t)}{b(t) \nbp{d_2 d}(t-n+1)} \right) \\\label{eq:10-s2}
    &\leq  \left( \frac{4(\fflow{s}+\bflow{d})}{\bflow{d}(1 - \delta)} \right) \left( \frac{ \nfp{ss_2}(t)}{b(t) \nbp{d_2 d}(t-n+1)} \right)
\end{align}
where we used the upper bound for pheromone level shown in \Cref{lem:pbound1}. Similarly, we can write 
\begin{align}
\label{eq:10-s3}
    \frac{ \fsst(t)}{b(t)\bflow{d_2d}(t-n+1)} &= \left( \frac{\fflow{s}}{\bflow{d}} \right) \left( \frac{ \nfp{ss_2}(t)}{b(t) \nbp{d_2 d}(t-n+1)} \right).
\end{align}
Combining \Cref{eq:10-s2} and  \Cref{eq:10-s3}, we  get
\begin{align}
    \frac{ \psst(t)}{b(t)\bflow{d_2d}(t-n+1)} + \frac{ \fsst(t)}{b(t)\bflow{d_2d}(t-n+1)} &\leq 
    \left( \frac{4(\fflow{s}+\bflow{d})}{\bflow{d}(1 - \delta)}  + \frac{\fflow{s}}{\bflow{d}}\right)\left( \frac{ \nfp{ss_2}(t)}{b(t) \nbp{d_2 d}(t-n+1)} \right)\\
    &\leq \left( \frac{5(\fflow{s}+\bflow{d})}{\bflow{d}(1 - \delta)}  \right)\left( \frac{ \nfp{ss_2}(t)}{b(t) \nbp{d_2 d}(t-n+1)} \right)\\\label{eq:10-s4}
     &=  C_{\fflow{s}, \bflow{d}, \delta} \left( \frac{ \nfp{ss_2}(t)}{b(t) \nbp{d_2 d}(t-n+1)} \right)
\end{align}
where we define $ C_{\fflow{s}, \bflow{d}, \delta} \defeq \left( \frac{5(\fflow{s}+\bflow{d})}{\bflow{d}(1 - \delta)}  \right) $.
Now, we upper bound  $ \frac{ \nfp{ss_2}(t)}{b(t) \nbp{d_2 d}(t-n+1)} $. From the definition of normalized pheromone level, we get 
\begin{align*}
     \frac{ \nfp{ss_2}(t)}{b(t) \nbp{d_2 d}(t-n+1)}  &= \frac{1}{b(t)} \left(\frac{1+r_{d_1 d}(t-n+1)}{1+r_{ss_1}(t)} \right)
\end{align*}
We know $r_{d_1 d}(t-n+1) = c(t) r_{min}(t)$ where $c(t) \geq 1$, and we earlier defined $r_{ss_1}(t) = a(t) r_{min}$ where $a(t) \geq 1$. This gives us 
\begin{align*}
     \frac{ \nfp{ss_2}(t)}{b(t) \nbp{d_2 d}(t-n+1)}  &= \frac{1}{b(t)} \left(\frac{1+c(t) r_{min}(t)}{1+a(t) r_{min}(t)} \right)\\
     &\leq 
     \frac{1}{b(t)} \left(\frac{1+c(t) r_{min}(t)}{1+ r_{min}(t)} \right)
\end{align*}
In \Cref{lem:helper-s1}, we show that $(1+c(t) r_{min}(t)) \leq b(t)(1+r_{min}(t))$, which gives us 
\begin{align*}
     \frac{ \nfp{ss_2}(t)}{b(t) \nbp{d_2 d}(t-n+1)}  &\leq 1
\end{align*}
Substituting this in \Cref{eq:10-s4}, we get
\begin{align}
    \frac{ \psst(t)}{b(t)\bflow{d_2d}(t-n+1)} + \frac{ \fsst(t)}{b(t)\bflow{d_2d}(t-n+1)} \leq  C_{\fflow{s}, \bflow{d}, \delta}
\end{align}
Substituting this in \Cref{eq:10-s5}, we get 
\begin{align}
    \gamma_s(t)
    &\geq  1 + \frac{\alpha - \beta}{ C_{\fflow{s}, \bflow{d}, \delta} + \beta} 
\end{align}

Define $\gamma_{s_l} \defeq  1 + \frac{\alpha - \beta}{C_{\fflow{s}, \bflow{d}, \delta} + \beta} $. Since $\gamma_{s_l} > 1$, this completes the proof of the lemma.
\end{proof}
Using a very similar proof, we can bound $\gamma_d(t)$ by $\gamma_{d_l}$. Finally, setting $\gamma_l = min(\gamma_{s_l}, \gamma_{d_l})$, we will get a lower bound $\gamma_l > 1$ on $\gamma(t)$, for all $t \geq T + L$. This finishes the proof for step 2.

\paragraph{Step 3.} In step 1, we show that $r_{min}(t)$ is non-decreasing for $t \geq L$, therefore $r_{min}(L + T_1) \geq r_{min}(L)$. And from step 2, we know that $r_{min}(t)$ increases at least by a factor of $\gamma_l$ every $L$ time steps, for all $t \geq L+T_1$. This  gives us $r_{min}(t) \geq r_{min}(L)  \gamma_l^{\lfloor\frac{t - L - T_1}{L}\rfloor}$. Let $T_2 = L \frac{ log\left(\frac{2}{\epsilon r_{min}(L) }\right)}{log(\gamma_l)}$. For  $t \geq L+T_1+T_2$ time steps, we would get that $r_{min} \geq \frac{2}{\epsilon}$, which implies $\nfp{ss_1}(t) \geq 1 - \epsilon$ and $\nbp{d_1 d}(t) \geq 1 - \epsilon$. Since the pheromone level on all edges on $P_1$ is always non-zero, we trivially know that $\nfp{uv}(t) = 1$ for $(u, v) \in P_1 \setminus {(s,s_1)}$ and $\nbp{uv}(t) = 1$ for $(u, v) \in P_1 \setminus {(d_1, d)}$. This gives us that for $t \geq L + T_1 + T_2$, normalized pheromone level on all edges on $P_1$ are at least $1- \epsilon$, where
\begin{align*}
T_1 &= \max_{(u,v) \in E}\left( \frac{log \left(\frac{p_{uv}(0)}{\lf + \rf}\right)} {log \left(\frac{1}{\delta}\right)}\right)\\
    T_2 &= L \frac{ log\left(\frac{2}{\epsilon r_{min}(L) }\right)}{log(\gamma_l)}
\end{align*}
This completes the proof of \Cref{thm:ub-leakage}.

\subsection{Increasing flow with no leakage on both paths (\Cref{thm:ub-changing-flow})}
Here, we provide a proof of \Cref{thm:ub-changing-flow}. We restate it below.
% \begin{theorem*}
% Consider a graph $G$ consisting of two parallel paths $P_1$ and $P_2$ from $s$ to $d$. Let $P_1$ be the shorter path such that $len_{P_1} < len_{P_2}$.  If  (1)  the initial pheromone levels ${p}_{uv}(0)$ is positive for any edge $(u, v) \in P_1$,
% and (2) leakage $l_{P_1} = l_{P_2} = 0$, then the following is true: \begin{enumerate}
%     \item \textbf{Exponential increase:} When $\fflow{s}(t) = \alpha^t \fflow{s}(0)$ and $\fflow{d}(t) = \alpha^t \fflow{d}(0)$, for any  $\alpha > 1$,   $\nfp{uv}(t) \geq 1 - \epsilon$ and $\nbp{uv}(t) \geq 1 - \epsilon$, for all  $t = \Omega \left(log \left(\frac{1}{\epsilon}\right)  \right)$ and  $(u, v) \in P_1$.
%       \item \textbf{Linear increase:} When $\fflow{s}(t) = \fflow{s}(0) + \alpha t$ and $\fflow{d}(t) = \fflow{d}(0) +\alpha t$, for any  $\alpha > 0$,  we get  $\nfp{uv}(t) \geq 1 - \epsilon$ and $\nbp{uv}(t) \geq 1 - \epsilon$, for all  $t = \Omega \left(  \left(\frac{1}{\epsilon} \right)^{C_{\delta, len_{P_2}}}\right) $ and $(u, v) \in P_1$. Here, $C_{\delta, len_{P_2}}$ is some constant dependent on $\delta$ and the path length.
%       \end{enumerate}
% \end{theorem*}
\secondthm*
\vspace{11pt}
 Let $s_1$, $s_2$ be the neighboring vertices of $s$ that belong to paths $P_1$ and $P_2$ respectively. Similarly let $d_1$ and $d_2$ be the corresponding neighbors for $d$. 
 Define $r_{ss_1}(t) \defeq \frac{p_{ss_1}(t)}{p_{ss_2}(t)}$ and $r_{d_1 d}(t) \defeq \frac{p_{d_1 d}(t)}{p_{d_2 d}(t)}$ to be the relative pheromone levels at $(s,s_1)$ and $(d_1,d)$ respectively. For notational simplicity, we define $m \defeq len_{P_1}$ and $n \defeq len_{P_2}$ to be the lengths of path $P_1$ and $P_2$,  and $L \defeq max(m, n)$. 

Our potential function at any time $t \geq L$  is given by the minimum of the relative pheromone levels $r_{ss_1}(t)$ and $r_{d_1d}(t)$ across the last $L$ time steps:
\begin{equation}
    r_{min}(t) \defeq min\{r_{ss_1}(t), r_{ss_1}(t-1), \cdots, r_{ss_1}(t-L+1), r_{d_1d}(t), r_{d_1d}(t-1), \cdots, r_{d_1d}(t-L+1)   \}
\end{equation}

The proof is similar to the proof for the case of leakage with constant flow.  
We divide our proof into 3 steps:
\begin{itemize}
    \item \textbf{Step 1}: In this step, we show that $r_{min}(t)$ is non-decreasing at every time step and increases by a factor of $ \gamma(t)$ every $L$ time steps,  for all $t \geq L$. Here, $\gamma(t)$ is some appropriately defined function which is greater than 1 for all $t \geq L$.
    \item \textbf{Step 2}: In this step, we will lower bound  $\gamma(t)$ to show that $r_{min}(t)$ increases sufficiently every $L$ time steps. 
    \item \textbf{Step 3}: We will find the rate of convergence based on the rate of increase shown in step 2.
\end{itemize}
\begin{comment}
\textbf{High level idea}: For each flow currently present at any edge of the graph, we associate a number $r$ which is the ratio of pheromone levels when this flow entered the graph. Let $\rmin$ be the minimum value of $r$ among all the flows currently present in the graph. Claim: $\rmin$ is non-decreasing, and it strictly increases every $max(n,m)$ time steps, where $n$ and $m$ are the path lengths.

Let $p_{uv}(t)$ be the flow from node $u$ to $v$ at time $t$, and  let $p_l$ and $\rf$ be fixed flow levels from left and right respectively. Define $r_{uv}(t)$ to be the ratio of pheromone levels at the terminal edges when the flow $p_{uv}(t)$ had entered the graph. For example, let $(u,v)$ be an edge on the path along $(B,D)$, then $r_{uv}(t)$ is the ratio of pheromone levels at $BD$ and $BC$ when the flow $p_{uv}(t)$ entered $BD$, that is,  $p_{uv}(t) = \frac{r_{uv}(t)}{1+r_{uv}(t)} * \rf$. 

Let $\rmin(t) = \min_{(u,v)} r_{uv}(t)$ be the minimum pheromone level ratio among flows present on all the edges at time $t$. Here, all the edges are directed. We will show that $r_{\text{min}(t)}$ is non-decreasing with $t$ and strictly increases every $max(m,n)$ time steps.
\end{comment}
Now, we give proofs for each of these three steps.

\paragraph{Step 1.}
We will use the following lemma for the proof.
\begin{lemma}\label{lem:inc-helper_min_ratio}
For any time $t \geq L$, the following is true,
$$\frac{ \nbp{d_1d}(t-m+1)}{ \nbp{d_2d}(t-n+1)} \geq \min(r_{d_1d}(t-m+1), r_{d_1 d}(t-n+1)) \geq r_{min}(t)~.$$
\end{lemma}
\begin{proof}
Under the linear decision rule we know that,
\begin{align}
 \nbp{d_1d}(t-m+1) &=   \frac{ r_{d_1d}(t-m+1)}{r_{d_1d}(t-m+1)+1} \\
 \nbp{d_2d}(t-n+1) &=  \frac{1}{r_{d_1d}(t-n+1)+1} \\
\end{align}
Let $r\defeq \min(r_{d_1d}(t-m+1), r_{ d_1d}(t-n+1))$, then note that $\frac{ r_{d_1d}(t-m+1)}{r_{d d_1}(t-m+1)+1} \geq \frac{r}{r+1}$ and $\frac{1}{r_{d_1d}(t-n+1)+1} \leq \frac{1}{r+1}$. Therefore,
$$\frac{ \nbp{d_1d}(t-m+1)}{ \nbp{d_2d}(t-n+1)} \geq  r=\min(r_{d_1d}(t-m+1), r_{d_1d}(t-n+1))~.$$
and from the definition of $r_{min}(t)$, we know that $\min(r_{d_1d}(t-m+1), r_{d_1 d}(t-n+1)) \geq r_{min}(t)$.
\begin{comment}
\begin{itemize}
    \item 
\end{itemize}
Case 1: $r_{d d_1}(t-n+1) > r_{d d_1}(t-m+1)$. In this case, we want to show that 
\begin{align}
\frac{ \fsos(t)}{ \fsts(t)} &\geq \frac{\lpo}{\lpt} r_{d d_1}(t-m+1)\\
\Leftrightarrow \
\frac{\lpo r_{d d_1}(t-n+1) (r_{d d_1}(t-m+1)+1)}{\lpt (r_{d d_1}(t-n+1)+1)} & \geq \frac{\lpo}{\lpt} r_{d d_1}(t-m+1)\\
\Leftrightarrow \
r_{d d_1}(t-n+1) (r_{d d_1}(t-m+1)+1) &\geq r_{d d_1}(t-m+1)(r_{d d_1}(t-n+1)+1) \\
\Leftrightarrow \
r_{d d_1}(t-n+1) &\geq r_{d d_1}(t-m+1)
\end{align}
Case 2: $r_{d d_1}(t-n+1) \leq r_{d d_1}(t-m+1)$. In this case, we want to show that 
\begin{align}
\frac{ \fsos(t)}{ \fsts(t)} &\geq \frac{\lpo}{\lpt} r_{d d_1}(t-n+1)\\
\Leftrightarrow \
\frac{\lpo r_{d d_1}(t-n+1) (r_{d d_1}(t-m+1)+1)}{\lpt (r_{d d_1}(t-n+1)+1)} & \geq \frac{\lpo}{\lpt} r_{d d_1}(t-n+1)\\
\Leftrightarrow \
r_{d d_1}(t-n+1) (r_{d d_1}(t-m+1)+1) &\geq r_{d d_1}(t-n+1)(r_{d d_1}(t-n+1)+1) \\
\Leftrightarrow \
r_{d d_1}(t-m+1) &\geq r_{d d_1}(t-n+1)
\end{align}
As desired, this gives us 
\begin{align}
\frac{ \fsos(t)}{ \fsts(t)} \geq \frac{\lpo}{\lpt}\text{min}(r_{d d_1}(t-n+1), r_{d d_1}(t-m+1))
\end{align}
\end{comment}
\end{proof}
Now we show that $r_{min}(t)$ is non-decreasing at every time step and increases by a factor of $ \gamma(t)$ every $L$ time steps, for all $t \geq L$. To show this, we would show that $r_{ss_1}(t+1) \geq r_{min}(t)\gamma_s(t)$, and $r_{d_1 d}(t+1) \geq r_{min}(t) \gamma_d(t)$.  Here, $\gamma_s(t)$ and $\gamma_d(t)$ are appropriately defined functions such that $\gamma_s(t)) > 1$ and $\gamma_d(t)) > 1$ for all $t \geq L$.

This would give us $r_{min}(t+L) \geq r_{min}(t) \gamma(t)$ where $\gamma(t) = min\{\gamma_s(t+L-1),\gamma_s(t+L-2), \cdots, \gamma_s(t), \gamma_d(t+L-1),\gamma_d(t+L-2), \cdots, \gamma_d(t) \}$.
Below, we give the proof for $r_{ss_1}(t+1) \geq r_{min}(t) \gamma_s(t)$. The proof for  $r_{d_1 d}(t+1) \geq r_{min}(t) \gamma_d(t)$ is  similar.
\begin{lemma}\label{lem:inc-leakage1}
For all time $t \geq L$, $r_{ss_1}(t+1) \geq r_{min}(t) \gamma_s(t)$,
for some appropriately defined  function $\gamma_s(t)$, such that $\gamma_s(t) > 1$ for all $t \geq L$.
\end{lemma}

\begin{proof}

 The pheromone levels on edges $(s,s_1)$ and $(s,s_2)$ at time $t+1$ is provided by the following expressions. 
\begin{align*}
    \psso(t+1) = \delta \left(\psso(t) + 
    \fsso(t) +  \fsos(t)\right) \quad 
    \psst(t+1) = \delta \left(\psst(t) + \fsst(t) +  \fsts(t)\right)~.\\
\end{align*}
Therefore, for $t \geq L$, we get
 \begin{align}
\frac{\psso(t+1)}{\psst(t+1)} &= \frac{ \psso(t) + 
    \fsso(t) +  \fsos(t)}{ \psst(t) + \fsst(t) +  \fsts(t)}\\
    &= \frac{ \psso(t) + 
    \fsso(t) +  (1-l_{P_1}) \ \bflow{d_1d}(t-m+1)}{ \psst(t) + \fsst(t) +  (1 - l_{P_2}) \ \bflow{d_2d}(t-n+1)}\\\label{eq:inc-9-s1-n}
    &= \frac{ \psso(t) + 
    \fsso(t) +   \ \bflow{d_1d}(t-m+1)}{ \psst(t) + \fsst(t) +   \ \bflow{d_2d}(t-n+1)}\\
    &= \frac{ \psso(t) + 
    \fsso(t) +   \ \bflow{d}(t-m+1) \nbp{d_1d}(t-m+1)}{ \psst(t) + \fsst(t) +   \ \bflow{d}(t-n+1)\nbp{d_2d}(t-n+1)}~.\label{eq:inc-9-s1}
\end{align}
where we used $l_{P_1} = l_{P_2} = 0$ in \Cref{eq:inc-9-s1-n} and express the flow in terms of normalized pheromone level in \Cref{eq:inc-9-s1}. By the definition of $r_{ss_1}(t)$ and linear decision rule we know that,
\begin{align*}
\frac{\psso(t)}{\psst(t)} = \frac{  
    \fsso(t)}{ \fsst(t)} = r_{ss_1}(t)  \geq r_{min}(t)~.
\end{align*}
From \Cref{lem:inc-helper_min_ratio}, we know that
\begin{align*}
    \frac{\nbp{d_1d}(t-m+1)}{\nbp{d_2d}(t-n+1)} \geq r_{min}(t)~.
\end{align*}
For $t \geq L$, define 
\begin{align}
a(t) \defeq \frac{r_{ss_1(t)}}{r_{min}(t)}, \ b(t) \defeq \frac{\nbp{d_1d}(t-m+1)}{\nbp{d_2d}(t-n+1)}\frac{1}{r_{min}(t)}~.
\end{align}
Note that, from the definition of $r_{min}(t)$ and using \Cref{lem:inc-helper_min_ratio}, we know that $a(t) \geq 1$ and $b(t) \geq 1$, for all $t \geq L$.
Now, we can write
\begin{align}
\label{eq:inc-9-s1_1}
\frac{\psso(t)}{\psst(t)} = \frac{
    \fsso(t)}{ \fsst(t)} = a(t) r_{min}(t), \ \frac{\nbp{d_1d}(t-m+1)}{\nbp{d_2d}(t-n+1)} = b(t) r_{min}(t)
\end{align}
Substituting this in \Cref{eq:inc-9-s1}, we get 
\begin{align}
\label{eq:inc-9-s2}
\frac{\psso(t+1)}{\psst(t+1)} &= 
    r_{min}(t)\frac{a(t) \psst(t) + 
    a(t)\fsst(t) +   b(t)  \ \ \bflow{d}(t-m+1) \nbp{d_2d}(t-n+1)}{ \psst(t) + \fsst(t) +   \ \bflow{d}(t-n+1)\nbp{d_2d}(t-n+1)}\\
    &= r_{min}(t) \gamma_s(t)
\end{align}
where we define 
\begin{align}
    \gamma_s(t) \defeq \frac{a(t) \psst(t) + 
    a(t)\fsst(t) +   b(t)  \ \ \bflow{d}(t-m+1) \nbp{d_2d}(t-n+1)}{ \psst(t) + \fsst(t) +   \ \bflow{d}(t-n+1)\nbp{d_2d}(t-n+1)} 
   .
\end{align}
Since the flow is flow is increasing with time, $\bflow{d}(t-m+1) > \bflow{d}(t-n+1)$, as $n > m$.  Also, $a(t) \geq 1$, $b(t) \geq 1$. Therefore, we get that $\gamma_s(t) > 1$ for all $t \geq L$. 

This completes the proof of the Lemma.
\end{proof}

\paragraph{Step 2.}
In this step, we will lower bound  $\gamma(t)$ to show that $r_{min}(t)$ increases sufficiently every $L$ time steps. For the multiplicative increase case, we will show that $\gamma(t) \geq \left(1 + \frac{c_1}{\alpha^{len_{P_2}}} \right)$ for some constant $c_1 > 0$. For the additive increase case, we will show that $\gamma(t) \geq \left(1 + \frac{c_2}{t} \right)$ for some constant $c_2 > 0$.  If the pheromone levels on the edges are too high as compared to the flow, it will take more time for the relative pheromone levels to change. So we will first show that there exists a time $T_1$, such that for $t \geq T_1$, the pheromone levels and flow are comparable. Our lower bounds for $\gamma(t)$  will hold for all $t \geq T_1$.

\begin{lemma}\label{lem:inc-pbound1}
In the flow dynamics governed by the linear decision rule, the pheromone level on any edge $e=(u,v)$ is always  bounded as follows:
\begin{equation*}
 p_{uv}(t) \leq \frac{2(\lf(t)+\rf(t))}{1-\delta} ~.\quad 
\end{equation*}
for all 
\begin{align*}
t \geq T_1 \defeq \max_{(u,v) \in E}\left( \frac{log \left(\frac{p_{uv}(0)}{\lf(0) + \rf(0)}\right)} {log \left(\frac{1}{\delta}\right)}\right)
\end{align*}
\end{lemma}
\begin{proof}
Consider any edge $e=(u,v)$ and time $t>0$,
\begin{align*}
    p_{uv}(t) & = \delta(\phe(t-1) + \fflow{uv}(t-1)  + \bflow{uv}(t-1))
     \leq \delta \left(\phe(t-1) + \lf(t-1)+\rf(t-1) \right)~,\\
    & \leq  \sum_{i=1}^{t}\delta ^{i} (\lf(i-1)+\rf(i-1)) + \delta^{t} \phe(0)\\
    &\leq (\lf(t)+\rf(t)) \sum_{i=1}^{t}\delta ^{i}  + \delta^{t} \phe(0)
\end{align*}
where we used that $\fflow{s}(t)$ and $\bflow{d}(t)$ are monotonically increasing in the last inequality.
For $t \geq \frac{log \left(\frac{p_{uv}(0)}{\lf(0) + \rf(0)}\right)} {log \left(\frac{1}{\delta}\right)}$, we get \begin{align}
    p_{uv}(t)  \leq \frac{(\lf(t)+\rf(t))}{1-\delta} + (\lf(0) + \rf(0))
    \leq \frac{2(\lf(t)+\rf(t))}{1-\delta}.
\end{align}
where we used that $\fflow{s}(t)$ and $\bflow{d}(t)$ are monotonically increasing.
Since $T_1 \geq \frac{log \left(\frac{p_{uv}(0)}{\lf(0) + \rf(0)}\right)} {log \left(\frac{1}{\delta}\right)}$, we get that for $t \geq T_1$,
\begin{align}
    p_{uv}(t) 
    \leq \frac{2(\lf(t)+\rf(t))}{1-\delta}.
\end{align}
\end{proof}

In \Cref{eq:inc-9-s1_1}, we defined $a(t)$ and $b(t)$ such that \begin{align*}
\frac{\psso(t)}{\psst(t)} = \frac{
    \fsso(t)}{ \fsst(t)} = a(t) r_{min}(t), \ \frac{\nbp{d_1d}(t-m+1)}{\nbp{d_2d}(t-n+1)} = b(t) r_{min}(t)
\end{align*} 
Similarly, for $t \geq L$, we define $c(t) \defeq \frac{r_{d_1d}(t-n+1)}{r_{min}(t)}$.  From the definition of $r_{min}(t)$, we know that $c(t) \geq 1$ for all $t \geq L$. We will show a relationship between $b(t)$ and $c(t)$ which will be useful later.
\begin{lemma}
\label{lem:inc-helper-s1}
For all $t \geq L$, $1+c(t) r_{min}(t) \leq b(t)(1+r_{min}(t))$.
\end{lemma}
\begin{proof}
We know 
\begin{align}
    r_{d_1 d}(t-n+1) = \frac{\nbp{d_1 d}(t-n+1)}{\nbp{d_2 d}(t-n+1)} = c(t) r_{min}(t)
\end{align}
which gives
\begin{align}
     \frac{1 - \nbp{d_2 d}(t-n+1)}{\nbp{d_2 d}(t-n+1)} = c(t) r_{min}(t)
\end{align}
which gives $\frac{1}{\nbp{d_2d}(t-n+1)} = 1 + c(t)r_{min}(t)$. 
We also know 
\begin{align}
    \frac{\nbp{d_1 d}(t-m+1)}{\nbp{d_2 d}(t-n+1)} = b(t) r_{min}(t)~.
\end{align}
Substituting $\frac{1}{\nbp{d_2d}(t-n+1)} = 1 + c(t)r_{min}(t)$, we get 
\begin{align}
    \nbp{d_1 d}(t-m+1) = \frac{b(t) r_{min}(t)}{1 + c(t)r_{min}(t)}~.
\end{align}
%%%%%%
Finally, from the definition of $r_{min}(t)$, we know 
\begin{align}
    \frac{\nbp{d_1 d}(t-m+1)}{\nbp{d_2 d}(t-m+1)} = r_{d_1 d}(t-m+1) \geq r_{min}(t).
\end{align}
Since $\nbp{d_2 d}(t-m+1) = 1 - \nbp{d_1 d}(t-m+1)$, this gives
\begin{align}
    {\nbp{d_1 d}(t-m+1)}  \geq \frac{r_{min}(t)}{r_{min}(t) + 1}.
\end{align}
Substituting the value $\nbp{d_1 d}(t-m+1) = \frac{b(t) r_{min}(t)}{1 + c(t)r_{min}(t)}$, we get 
\begin{align}
 \frac{b(t) r_{min}(t)}{1 + c(t)r_{min}(t)} \geq \frac{r_{min}(t)}{r_{min}(t) + 1}.
\end{align}
which implies
\begin{align}
    b(t)(1+r_{min}(t)) \geq 1+c(t) r_{min}(t)  ~.
\end{align}
This finishes the proof of the Lemma.

%%%%%%

\end{proof}

Now, we come to the main part of step 2 where we lower bound $\gamma(t)$  for all $t \geq T_1 + L$. From step 1, we know that $\gamma(t) = min\{\gamma_s(t+L-1),\gamma_s(t+L-2), \cdots, \gamma_s(t), \gamma_d(t+L-1),\gamma_d(t+L-2), \cdots, \gamma_d(t) \}$. Therefore, to lower bound $\gamma(t)$, we need to lower bound $\gamma_s(t)$ and $\gamma_d(t)$.

Below we will prove a lower bound for  $\gamma_s(t)$.

\begin{lemma}\label{lem:inc-leakage2}
For all time $t \geq L+T_1$, $r_{ss_1}(t+1) \geq r_{min}(t) \gamma_s(t)$ where  
\begin{align*}
    \gamma_s(t) \geq \ 1 + \frac{(1 - \delta)\left(\bflow{d}(t-m+1) - \bflow{d}(t-n+1)\right)}{  6(\fflow{s}(t)+\bflow{d}(t))   } 
\end{align*}
\end{lemma}

\begin{proof}
From \Cref{lem:inc-leakage1}, we know that 
\begin{align}
\frac{\psso(t+1)}{\psst(t+1)} &= 
    r_{min}(t)\gamma_s(t)
\end{align}
where 
\begin{align*}
   \gamma_s(t) \defeq \frac{a(t) \psst(t) + 
    a(t)\fsst(t) +   b(t)  \ \ \bflow{d}(t-m+1) \nbp{d_2d}(t-n+1)}{ \psst(t) + \fsst(t) +   \ \bflow{d}(t-n+1)\nbp{d_2d}(t-n+1)}
\end{align*}
%We want to show that $\gamma_s(t) \geq \gamma_{s_l}$ for some fixed constant $\gamma_{s_l} > 1$. 
\begin{comment}
We consider two cases:

\textbf{Case 1:} $a(t) \geq 2$.\\
In this case, $\min\left(2, \frac{\alpha}{\beta}\right)$
\begin{align}
    \gamma_s(t) &= 
    \frac{a(t) \psst(t) + 
    a(t)\fsst(t) +   b(t)  \ \ \bflow{d}(t-m+1) \nbp{d_2d}(t-n+1)}{ \psst(t) + \fsst(t) +   \ \bflow{d}(t-n+1)\nbp{d_2d}(t-n+1)}\\ \label{eq:inc-10-s1}
    &\geq \frac{2 \psst(t) + 
    2\fsst(t) +    \ \ \bflow{d}(t-m+1) \nbp{d_2d}(t-n+1)}{ \psst(t) + \fsst(t) +   \ \bflow{d}(t-n+1)\nbp{d_2d}(t-n+1)}\\
     &\geq \ \min\left(2, \frac{\bflow{d}(t-m+1)}{\bflow{d}(t-n+1)}\right) \frac{ \psst(t) + 
    \fsst(t) +    \ \ \bflow{d}(t-n+1) \nbp{d_2d}(t-n+1)}{ \psst(t) + \fsst(t) +   \ \bflow{d}(t-n+1)\nbp{d_2d}(t-n+1)}
     \\
     &= \ \min\left(2, \frac{\bflow{d}(t-m+1)}{\bflow{d}(t-n+1)}\right)\\
     &=  \ \min\left(2, \left( 1 + \frac{\bflow{d}(t-m+1) - \bflow{d}(t-n+1)}{\bflow{d}(t-n+1)} \right) \right)~.
\end{align}
where we used $a(t) \geq 2$, $b(t) \geq 1$ for inequality \ref{eq:inc-10-s1}.

\textbf{Case 2:} $a(t) < 2$.\\
In this case,
\end{comment}
Using $a(t) \geq 1$ and $b(t) \geq 1$, we get
\begin{align}
\gamma_s(t) 
    &\geq \frac{ \psst(t) + 
    \fsst(t) +   b(t)  \ \ \bflow{d}(t-m+1) \nbp{d_2d}(t-n+1)}{ a(t)\psst(t) + a(t)\fsst(t) +   \ b(t)\bflow{d}(t-n+1)\nbp{d_2d}(t-n+1)}\\\label{eq:inc-10-s5}
    &=  1 + \frac{\bflow{d}(t-m+1) - \bflow{d}(t-n+1)}{\frac{ \psst(t)}{b(t)\nbp{d_2d}(t-n+1)} + \frac{ \fsst(t)}{b(t)\nbp{d_2d}(t-n+1)} + \bflow{d}(t-n+1)} 
\end{align}

 Now, to lower bound $\gamma_s(t)$, we need to upper bound $\frac{ \psst(t)}{b(t)\nbp{d_2d}(t-n+1)}$ and  $\frac{ \fsst(t)}{b(t)\nbp{d_2d}(t-n+1)}$. Using the definition of normalized pheromone level, we can write
\begin{align}
    \frac{ \psst(t)}{b(t)\nbp{d_2d}(t-n+1)} &= \left( {p_{ss_1}(t)+p_{ss_2}(t)} \right) \left( \frac{ \nfp{ss_2}(t)}{b(t) \nbp{d_2 d}(t-n+1)} \right) \\\label{eq:inc-10-s2}
    &\leq  \left(\frac {4(\fflow{s}(t)+\bflow{d}(t))}{1-\delta} \right) \left( \frac{ \nfp{ss_2}(t)}{b(t) \nbp{d_2 d}(t-n+1)} \right)
\end{align}
where we used the upper bound for pheromone level shown in \Cref{lem:inc-pbound1}. Similarly, we can write 
\begin{align}
\label{eq:inc-10-s3}
    \frac{ \fsst(t)}{b(t)\nbp{d_2d}(t-n+1)} &=  {\fflow{s}(t)}  \left( \frac{ \nfp{ss_2}(t)}{b(t) \nbp{d_2 d}(t-n+1)} \right).
\end{align}
Combining \Cref{eq:inc-10-s2} and  \Cref{eq:inc-10-s3}, we  get
\begin{align}
    \frac{ \psst(t)}{b(t)\nbp{d_2d}(t-n+1)} + \frac{ \fsst(t)}{b(t)\nbp{d_2d}(t-n+1)} &\leq 
    \left( \frac{4(\fflow{s}(t)+\bflow{d}(t))}{(1 - \delta)}  + {\fflow{s}(t)}\right)\left( \frac{ \nfp{ss_2}(t)}{b(t) \nbp{d_2 d}(t-n+1)} \right)\\\label{eq:inc-10-s4}
    &\leq \left( \frac{5(\fflow{s}(t)+\bflow{d}(t))}{(1 - \delta)}  \right)\left( \frac{ \nfp{ss_2}(t)}{b(t) \nbp{d_2 d}(t-n+1)} \right)
     %&=  C_{\fflow{s}, \bflow{d}, \delta} \left( \frac{a(t) \nfp{ss_2}(t)}{b(t) \nbp{d_2 d}(t-n+1)} \right)
\end{align}
Now, we upper bound  $\left( \frac{ \nfp{ss_2}(t)}{b(t) \nbp{d_2 d}(t-n+1)} \right)$. From the definition of normalized pheromone level, we get 
\begin{align*}
     \frac{ \nfp{ss_2}(t)}{b(t) \nbp{d_2 d}(t-n+1)}  &= \frac{1}{b(t)} \left(\frac{1+r_{d_1 d}(t-n+1)}{1+r_{ss_1}(t)} \right)
\end{align*}
We know $r_{d_1 d}(t-n+1) = c(t) r_{min}(t)$ where $c(t) \geq 1$, and we earlier defined $r_{ss_1}(t) = a(t) r_{min}$ where $a(t) \geq 1$. This gives us 
\begin{align*}
     \frac{ \nfp{ss_2}(t)}{b(t) \nbp{d_2 d}(t-n+1)}  &= \frac{1}{b(t)} \left(\frac{1+c(t) r_{min}(t)}{1+a(t) r_{min}(t)} \right)\\
     &\leq 
     \frac{1}{b(t)} \left(\frac{1+c(t) r_{min}(t)}{1+ r_{min}(t)} \right)
\end{align*}
In \Cref{lem:inc-helper-s1}, we show that $(1+c(t) r_{min}(t)) \leq b(t)(1+r_{min}(t))$, which gives 
\begin{align*}
 \frac{ \nfp{ss_2}(t)}{b(t) \nbp{d_2 d}(t-n+1)}  &\leq 1
\end{align*}

Substituting this in \Cref{eq:inc-10-s4}, we get
\begin{align}
    \frac{ \psst(t)}{b\bflow{d_2d}(t-n+1)} + \frac{ \fsst(t)}{b\bflow{d_2d}(t-n+1)} \leq  \frac{5(\fflow{s}(t)+\bflow{d}(t))}{(1 - \delta)}  
\end{align}
Substituting this in \Cref{eq:inc-10-s5}, we get 
\begin{align}
    \gamma_s(t)
    &\geq  1 + \frac{\bflow{d}(t-m+1) - \bflow{d}(t-n+1)}{ \left( \frac{5(\fflow{s}(t)+\bflow{d}(t))}{(1 - \delta)}  \right) + \bflow{d}(t-n+1)} 
\end{align}

Since the flow is monotonically increasing, we can further simplify this to get
\begin{align}
    \gamma_s(t) &\geq  \ 1 + \frac{\bflow{d}(t-m+1) - \bflow{d}(t-n+1)}{ \left( \frac{6(\fflow{s}(t)+\bflow{d}(t))}{(1 - \delta)}  \right) }  \\
    &=  \ 1 + \frac{(1 - \delta)\left(\bflow{d}(t-m+1) - \bflow{d}(t-n+1)\right)}{  6(\fflow{s}(t)+\bflow{d}(t))   } 
\end{align}
\end{proof}
Using a similar argument, we can bound $\gamma_d(t)$ as
\begin{align*}
    \gamma_d(t) \geq  \ 1 + \frac{(1 - \delta)\left(\fflow{s}(t-m+1) - \fflow{s}(t-n+1)\right)}{  6(\fflow{s}(t)+\bflow{d}(t))   } .
\end{align*}
Now, using these lower bounds on $\gamma_s(t)$ and $\gamma_l(t)$, we can lower bound $\gamma(t)$ for the multiplicative increase and additive increase case.
\begin{lemma}
\label{lem:inc-exp-gamma-lb}
Consider the multiplicative increase case, that is, when $\fflow{s}(t) = \alpha^t\fflow{s}(0)$ and $\bflow{d}(t) = \alpha^t\fflow{d}(0)$ for some $\alpha > 1$. For  all $t \geq L + T_1$.
\begin{align*}
    \gamma(t) \geq C_{\fflow{s}(0),\bflow{d}(0),\delta,\alpha,m}  > 1
\end{align*}
where $C_{\fflow{s}(0),\bflow{d}(0),\delta,\alpha,m}$ is some constant dependent on $\fflow{s}(0), \bflow{d}(0),\delta, \alpha$ and $m$.
\end{lemma}
\begin{proof}
From \Cref{lem:inc-leakage2}, we know that 
\begin{align*}
    \gamma_s(t) &\geq  \ 1 + \frac{(1 - \delta)\left(\bflow{d}(t-m+1) - \bflow{d}(t-n+1)\right)}{  6(\fflow{s}(t)+\bflow{d}(t))   } \\
    &= 1 + \frac{(1 - \delta)\left(\bflow{d}(0)\alpha^{t-m+1} - \bflow{d}(0){\alpha^{t-n+1}}\right)}{  6(\fflow{s}(0)+\bflow{d}(0))\alpha^t   } \\
    &= 1 + \frac{(1 - \delta)\left(\bflow{d}(0)\alpha^{t-m+1} - \bflow{d}(0){\alpha^{t-n+1}}\right)}{  6(\fflow{s}(0)+\bflow{d}(0))\alpha^t   } \\
    &= 1 + \frac{\bflow{d}(0)(1 - \delta)\left(\alpha^{n-m} - {1}\right)}{  6(\fflow{s}(0)+\bflow{d}(0))\alpha^{n-1}  } \\
    &= 1 + \frac{\bflow{d}(0)(1 - \delta)}{  6(\fflow{s}(0)+\bflow{d}(0)) }\left(\frac{1}{\alpha^{m-1}} - \frac{1}{\alpha^{n-1}} \right)
    %&= \min\left(2, \left(1 + \frac{\bflow{d}(0)(1 - \delta)\left(\alpha^{n-m} - {1}\right)}{  11(\fflow{s}(0)+\bflow{d}(0))\alpha^{n-1}  } \right) \right)\\
    %&= \min\left(2, \left(1 + \frac{\bflow{d}(0)(1 - \delta)\left(1 - \frac{1}{\alpha^{n-m}}\right)}{  11(\fflow{s}(0)+\bflow{d}(0))\alpha^{m-1}  } \right) \right)\\
   %&\geq \min\left(2, \left(1 + \frac{\bflow{d}(0)(1 - \delta)\left(\alpha-1\right)}{  11(\fflow{s}(0)+\bflow{d}(0))\alpha^{m}  } \right) \right)
\end{align*}
Using $n > m$, this gives us
\begin{align*}
    \gamma_s(t) &\geq 1 + \frac{\bflow{d}(0)(1 - \delta)\left(\alpha - 1\right)}{  6(\fflow{s}(0)+\bflow{d}(0))\alpha^{m}  }
\end{align*}
Similarly, we get 
\begin{align*}
\gamma_d(t) \geq 1 + \frac{\fflow{s}(0)(1 - \delta)\left(\alpha - {1}\right)}{  6(\fflow{s}(0)+\bflow{d}(0))\alpha^{m}  } 
\end{align*}
We define 
\begin{align*}
C_{\fflow{s}(0),\bflow{d}(0),\delta,\alpha,m} \defeq  1 + \frac{min(\fflow{s}(0),\bflow{d}(0))(1 - \delta)\left(\alpha - {1}\right)}{  6(\fflow{s}(0)+\bflow{d}(0))\alpha^{m}  } 
\end{align*}
It is easy to see that $\gamma_d(t) \geq C_{\fflow{s}(0),\bflow{d}(0),\delta,\alpha,m}$, $\gamma_s(t) \geq C_{\fflow{s}(0),\bflow{d}(0),\delta,\alpha,m}$, and $ C_{\fflow{s}(0),\bflow{d}(0),\delta,\alpha,m} > 1$, for all $t \geq L + T_1$.
We know that $\gamma(t) = min\{\gamma_s(t+L-1),\gamma_s(t+L-2), \cdots, \gamma_s(t), \gamma_d(t+L-1),\gamma_d(t+L-2), \cdots, \gamma_d(t) \}$. This implies $\gamma(t) \geq C_{\fflow{s}(0),\bflow{d}(0),\delta,\alpha,n,m}$ which completes the proof.
\end{proof}

\begin{lemma}
\label{lem:inc-lin-gamma-lb}
In the additive increase case, that is, when $\fflow{s}(t) = \fflow{s}(0) + \alpha t$ and $\bflow{d}(t) = \fflow{d}(0) + \alpha t$ for some $\alpha > 0$,
\begin{align*}
   \gamma(t) \geq 1 + \frac{\alpha(1 - \delta)}{  6(\fflow{s}(0)+\bflow{d}(0) + 2 \alpha L + 2 \alpha t )   }  
\end{align*}
 for all $t \geq L + T_1$.
\end{lemma}
\begin{proof}
From \Cref{lem:inc-leakage2}, we know that 
\begin{align*}
    \gamma_s(t) &\geq  \ 1 + \frac{(1 - \delta)\left(\bflow{d}(t-m+1) - \bflow{d}(t-n+1)\right)}{  6(\fflow{s}(t)+\bflow{d}(t))   } \\
    &= 1 + \frac{(1 - \delta)\left(\alpha(n-m)\right)}{  6(\fflow{s}(0)+\bflow{d}(0) + 2 \alpha t)   } \\
    &\geq 1 + \frac{\alpha(1 - \delta)}{  6(\fflow{s}(0)+\bflow{d}(0) + 2 \alpha t) } 
\end{align*}
Similarly, we get 
\begin{align*}
\gamma_d(t) \geq 1 + \frac{\alpha(1 - \delta)}{  6(\fflow{s}(0)+\bflow{d}(0) + 2 \alpha t)   }  
\end{align*}

We know that $\gamma(t) = min\{\gamma_s(t+L-1),\gamma_s(t+L-2), \cdots, \gamma_s(t), \gamma_d(t+L-1),\gamma_d(t+L-2), \cdots, \gamma_d(t) \}$. This gives us 
\begin{align*}
   \gamma(t) \geq 1 + \frac{\alpha(1 - \delta)}{  6(\fflow{s}(0)+\bflow{d}(0) + 2 \alpha L + 2 \alpha t )   }  
\end{align*}
\end{proof}

\paragraph{Step 3.} 
In this step, we use the lower bounds on $\gamma(t)$ shown in step 2, to find the rate of convergence.
\begin{lemma}
Consider the multiplicative increase case, that is, when $\fflow{s}(t) = \alpha^t\fflow{s}(0)$ and $\bflow{d}(t) = \alpha^t\fflow{d}(0)$ for some $\alpha > 1$.
For all $t \geq L + T_1 + T_2$, $\nfp{uv}(t) \geq 1 - \epsilon$ and $\nbp{uv}(t) \geq 1 - \epsilon$, for all $(u,v) \in P_1$. Here,
\begin{align*}
T_1 &= \max_{(u,v) \in E}\left( \frac{log \left(\frac{p_{uv}(0)}{\lf(0) + \rf(0)}\right)} {log \left(\frac{1}{\delta}\right)}\right)\\
    T_2 &= L \frac{ log\left(\frac{2}{\epsilon r_{min}(L) }\right)}{log\left(C_{\fflow{s}(1),\bflow{d}(1),\delta,\alpha,m}\right)} \leq \frac{12  (\fflow{s}(0)+\bflow{d}(0))}{(1-\delta)min(\fflow{s}(0),\bflow{d}(0))} \frac{\alpha^m L}{\alpha - 1}  { log\left(\frac{2}{\epsilon r_{min}(L) }\right)}
\end{align*}
 
\end{lemma}
\begin{proof}
In step 1, we show that $r_{min}(t)$ is non-decreasing for $t \geq L$, we know that $r_{min}(L + T_1) \geq r_{min}(L)$. And from step 2, we know that $r_{min}(t)$ increases at least by a factor of $C_{\fflow{s}(0),\bflow{d}(0),\delta,\alpha,n,m}$ every $L$ time steps, for all $t \geq L+T_1$. For clarity of writing, we use $C$ to denote $C_{\fflow{s}(0),\bflow{d}(0),\delta,\alpha,m}$ in the rest of this proof. This  gives us 
\begin{align*}
r_{min}(t) \geq r_{min}(L)  C^{\lfloor\frac{t - L - T_1}{L}\rfloor}
\end{align*}
Let $T_2 = L \frac{ log\left(\frac{2}{\epsilon r_{min}(L) }\right)}{log(C)}$. For  $t \geq L+T_1+T_2$ time steps, we would get that $r_{min} \geq \frac{2}{\epsilon}$, which implies $\nfp{ss_1}(t) \geq 1 - \epsilon$ and $\nbp{d_1 d}(t) \geq 1 - \epsilon$. Since the pheromone level on all edges on $P_1$ is always non-zero, we trivially know that $\nfp{uv}(t) = 1$ for $(u, v) \in P_1 \setminus {(s,s_1)}$ and $\nbp{uv}(t) = 1$ for $(u, v) \in P_1 \setminus {(d_1, d)}$. This gives us that for $t \geq L + T_1 + T_2$, normalized pheromone level on all edges on $P_1$ are at least $1- \epsilon$.

To get the desired upper bound on $T_2$, we use $\frac{x}{1+x} \leq log(1+x)$ for all $x \geq -1$. Using this we can write 
\begin{align*}
\frac{1}{log(C)} &= \frac{1}{log\left( 1 + \frac{min(\fflow{s}(0),\bflow{d}(0))(1 - \delta)\left(\alpha - {1}\right)}{  6(\fflow{s}(0)+\bflow{d}(0))\alpha^{m}  }  \right)}\\
&\leq 1 + \frac{  6(\fflow{s}(0)+\bflow{d}(0))\alpha^{m}  } {min(\fflow{s}(0),\bflow{d}(0))(1 - \delta)\left(\alpha - {1}\right)}\\
&\leq \frac{  12(\fflow{s}(0)+\bflow{d}(0))\alpha^{m}  } {min(\fflow{s}(0),\bflow{d}(0))(1 - \delta)\left(\alpha - {1}\right)}
\end{align*}
Substituting this in the expression for $T_2$, we get the desired bound.
\end{proof}
 Note that the expression for $T_2$ involves a $\frac{\alpha^m}{\alpha - 1}$ term. So if $\alpha$ is a constant, this would give exponential dependence on the length of the shortest path $m$. But by setting $\alpha \approx 1+ \frac{1}{m}$, we would get $\frac{\alpha^m}{\alpha - 1} \approx m$ (for $m$ large enough), making the dependence on the length of the shortest path linear. 
\begin{lemma}
Consider the additive increase case, that is, when $\fflow{s}(t) = \fflow{s}(0) + \alpha t$ and $\bflow{d}(t) = \fflow{d}(0) + \alpha t$ for some $\alpha > 0$.
For all $t \geq (L + T_1 + T_2) \left(\frac{2}{r_{min}(L)\epsilon} \right)^{{C_{\delta,L}}}$, $\nfp{uv}(t) \geq 1 - \epsilon$ and $\nbp{uv}(t) \geq 1 - \epsilon$, for all $(u,v) \in P_1$. Here,
\begin{align*}
T_1 &= \max_{(u,v) \in E}\left( \frac{log \left(\frac{p_{uv}(0)}{\lf(0) + \rf(0)}\right)} {log \left(\frac{1}{\delta}\right)}\right)\\
    T_2 &= \frac{2\alpha L + \fflow{s}(0) + \bflow{d}(0)}{2 \alpha}
\end{align*}
\end{lemma}
\begin{proof}
 From step 2, we know that $r_{min}(t)$ increases at least by a factor of $\gamma(t)$  every $L$ time steps for $t \geq L+T_1$, where
\begin{align*}
   \gamma(t) \geq 1 + \frac{\alpha(1 - \delta)}{  6(\fflow{s}(0)+\bflow{d}(0) + 2 \alpha L + 2 \alpha t )   }  
\end{align*}
Note that for $t \geq T_2$,
\begin{align*}
   \gamma(t) &\geq 1 + \frac{\alpha(1 - \delta)}{  6(4 \alpha t )   }  \\
   &= 1 + \frac{C_{ \delta}}{t}
\end{align*}
where we define $C_{ \delta}
\defeq \frac{(1-\delta)}{24}$. Using $1+x \geq exp({\frac{x}{2}})$ for $0<x<1$, we get 
\begin{align*}
   \gamma(t) &\geq  exp\left({\frac{C_\delta}{2t}}\right)
\end{align*}
This gives us that $r_{min}(t)$ increases at least by a factor of $exp\left({\frac{C_\delta}{2t}}\right)$ every $L$ time steps for $t \geq L+T_1+T_2$. 
In step 1, we show that $r_{min}(t)$ is non-decreasing for $t \geq L$, so we know that $r_{min}(L + T_1+T_2) \geq r_{min}(L)$.  Define $C_{\delta,L} = \frac{2L}{C_\delta}$ and  $T \defeq L+T_1+T_2$. Also, suppose $t$ is of the form $T+kL$ for some positive integer $k$.  Then this gives us

\begin{align*}
r_{min}(t) &\geq r_{min}(L) \left( exp\left(\frac{C_\delta}{2T} \right) exp\left(\frac{C_\delta}{2(T+L)} \right) \cdots exp\left(\frac{C_\delta}{2(t-L)} \right) \right)\\
&= r_{min}(L)  exp\left(\frac{C_\delta}{2L} \sum_{i=0}^{\frac{t-T}{L}-1} \frac{1}{\frac{T}{L}+i} \right)\\
&\geq r_{min}(L)exp\left(\frac{C_\delta}{2L} log\left(\frac{t}{T}\right) \right)\\
&= r_{min}(L)exp\left(\frac{1}{C_{\delta, L}} log\left(\frac{t}{T}\right) \right)
%&\geq r_{min}(L)  exp\left(C_{\delta,L} \sum_{i=1}^{\frac{t-T}{L}} \frac{1}{i} \right)\\
%&\geq r_{min}(L)exp\left(C_{\delta, L} log(\frac{t-T}{L}) \right)
\end{align*}
where we used $\sum_{i=0}^{x-1} \frac{1}{c+i} \geq log(\frac{x+c}{c})$ for the last inequality. 
From here, we get that for $t \geq T \left(\frac{2}{r_{min}(L)\epsilon} \right)^{{C_{\delta,L}}}$, $r_{min}(t) \geq \frac{2}{\epsilon}$. 
In the calculations above, we also assumed that $t$ is of the form $T+kL$ for some positive integer $k$. But as we know that $r_{min}(t)$ is monotonically increasing for all $t \geq L$, $r_{min}(t) \geq \frac{2}{\epsilon}$ holds for all $t \geq T \left(\frac{2}{r_{min}(L)\epsilon} \right)^{{C_{\delta,L}}}$.
This  implies $\nfp{ss_1}(t) \geq 1 - \epsilon$ and $\nbp{d_1 d}(t) \geq 1 - \epsilon$. Since the pheromone level on all edges on $P_1$ is always non-zero, we trivially know that $\nfp{uv}(t) = 1$ for $(u, v) \in P_1 \setminus {(s,s_1)}$ and $\nbp{uv}(t) = 1$ for $(u, v) \in P_1 \setminus {(d_1, d)}$. This gives us that for $t \geq T \left(\frac{2}{r_{min}(L)\epsilon} \right)^{{C_{\delta,L}}}$, normalized pheromone level on all edges on $P_1$ are at least $1- \epsilon$.
\end{proof}

This completes the proof of \Cref{thm:ub-changing-flow}.

\subsection{Connecting leakage and number of vertices}
\label{app:connect_leakage_vertices}

Here, we show that for any graph (not necessarily the one with parallel paths), the path with minimum leakage also has approximately the minimum number of vertices as long as variation in leakage between different vertices is not too large. (The leakage we consider in the lemma below are only for non-terminal vertices. By convention, we do not have leakage at terminal vertices in our model.)

\begin{lemma}
Suppose for all pairs of  vertices $u$ and $v$, $\log(1-l_u)$ and $log(1-l_v)$ are within a $(1+\epsilon)$ factor of each other, where $l_v \in (0, 1)$ for all   $v$. Then the path with the minimum leakage has number of vertices at most $(1+\epsilon)$ times the path with the minimum number of vertices.
\end{lemma}
\begin{proof}
Let $\max_v ~ (1-l_v) = \alpha$ for some $\alpha \in (0, 1)$. Then, since $\log(1-l_u)$ and $\log(1-l_v)$ are within a factor of $(1+\epsilon)$ for all pairs of vertices $u$ and $v$, we get 
\[
\min_v ~ (1-l_v) \geq \alpha^{1+\epsilon}.
\]
Let $P_1$ be the path with the minimum number of vertices between $s$ and $d$. And let the number of vertices on $P_1$ between equal to $m$ (not counting $s$ and $d$). Let $P_2$ be the path with the minimum leakage among all paths between $s$ and $d$, and let the number of vertices on $P_2$ be equal to $n$ (not counting $s$ and $d$). 

Then, by the bounds on leakage, we get bounds on path leakage for $P_1$ and $P_2$,
\begin{align*}
l_{P_1} &\leq 1- \alpha^{m(1+\epsilon)},\\
l_{P_2} &\geq 1 - \alpha^{n}.
\end{align*}

Since $l_{P_2} \leq l_{P_1}$ ($P_2$ is the path with the minimum leakage), we get
\[
\alpha^n \geq \alpha^{m(1+\epsilon)},
\]
which implies
\[
m(1+\epsilon)\log \left(\frac{1}{\alpha} \right) \geq n \log \left(\frac{1}{\alpha} \right),
\]
which gives
\[
n \leq m(1+\epsilon).
\]
Therefore, the path with the minimum leakage has number of vertices at most $(1+\epsilon)$ times the path with the minimum number of vertices.
\end{proof}

We do not consider the degenerate case of $l_v = 0$ or $l_v = 1$ in the above lemma as in that case, the condition on leakage would imply that either all vertices have leakage $0$ in which case all paths have leakage $0$, or all vertices have leakage $1$ in which case all paths have leakage $1$. Our claims on the effect of leakage are only relevant and interesting when we are not in these degenerate cases.

\newcommand{\epspp}{\epsp'}
\newcommand{\lo}{\alpha}
\newcommand{\lt}{\beta}
\newcommand{\cue}{c_{g, r}}
\section{Characterization of other rules}
% Here we prove the results stated in \Cref{sec:general_rules_lb} and for convenience we recall the definition of $\setflb$. Let $\setflb$ be the set of all symmetric, continuous, monotone, pheromone ratio based rules on the 2 parallel path graph; formally any $g \in \setflb$ satisfy the following:
% \begin{itemize}
%     \item The rule $g$ is symmetric, that is, the same rule applies to both $s$ and $d$.
%     \item The rule $g:[0,1/2] \rightarrow [0,1]$ is pheromone ratio based, meaning at all time steps it takes the minimum of the two normalized pheromone levels on edges belonging to the two paths and returns the fraction of the flow on the minimum edge.
%     \item The rule $g$ is monotonically non-decreasing, meaning $g(x) \leq g(x')$ for all $x,x' \in [0,1/2]$ and $x \leq x'$. 
%     %Further, $g$ is a continuous function, and satisfies $g(0)=0$ and $g(1/2)=1/2$.
% \end{itemize}

\subsection{Necessity of bidirectional flow (\Cref{prop:one})}
\firstprop*
\begin{proof}[Proof of \Cref{prop:one}]
Let $P_1$ and $P_2$ denote be two parallel paths between $s$ and $d$, and $s_1$ and $s_2$ be neighboring vertices of $s$ on $P_1$ and $P_2$ respectively. Without loss of generality, let $P_1$ be the minimum leakage or the shortest path, and let the flow be unidirectional from $s$ to $d$.

%In the setting of one sided flow, for the dynamics on $G$ governed by the pheromone based rule $g$ to converge to any single path it should be true that,
%$$\psso(t) \rightarrow 0 \text{ or } \psst(t) \rightarrow 0 \text{ as }t \rightarrow \infty~.$$
As the flow is unidirectional, we get that  the pheromone levels on the edges incident on $s$ is only a function of
their initial pheromone levels and forward flow at $s$. That is, for some function $F$, $\psso(t)$ and $\psst(t)$ can be written as
\begin{align*}
\psso(t) &= F(\psso(0),\psst(0),\fflow{s}(0), \fflow{s}(1),\dots \fflow{s}(t-1))\\
\psst(t) &= F(\psst(0),\psso(0),\fflow{s}(0), \fflow{s}(1),\dots \fflow{s}(t-1))
\end{align*}
Given $(\psso(0),\psst(0))$, and $\fflow{s}(t)$ for all $t \geq 0$, suppose the dynamics converges to $P_1$.  Now if we swap the initial pheromone level on the two edges incident to $s$, then the dynamics would converge to $P_2$. Therefore, for one of these two initial pheromone settings, the dynamics does not converge to the minimum leakage or the shortest path.
\end{proof}

\renewcommand{\epsp}{\epsilon}

\subsection{Necessity of the linear rule for convergence to the minimum leakage path (\Cref{prop:two})}
\secondprop*
\begin{proof}[Proof of \Cref{prop:two}]
For notational convenience, we define $\alpha=1-l_{P_1}$ and $\beta=1-l_{P_2}$. Without loss of generality we assume $P_{1}$ to be the path of minimum leakage and therefore we let $\alpha \geq \beta$. Let $n$ and $m$ be the number of vertices between $s$ and $d$ on paths $P_1$ and $P_2$ respectively. We name the vertices from left to right on path $P_1$ by $v_{0}$ to $v_{n+1}$ and $P_2$ by $u_0$ to $u_{m+1}$, with the convention $v_0=u_0=s$ and $v_{n+1}=u_{m+1}=d$. We also let $s_1=v_1,s_2=u_1,d_1=v_n$ and $d_2=u_m$. Let $\fflow{s}$ and $\bflow{d}$ denote the fixed incoming forward and backward flow. We divide the analysis into two cases.
For any non-linear $g \in \setflb$, there exists an $r \in (0,1/2)$ such that $g(r)\neq r$. For such an $r$, one of the following two conditions holds:
\begin{itemize}
    \item $g(r)< r$.
    \item $g(r)>r$.
\end{itemize}

\paragraph{Case 1:} Suppose $g(r)< r$, then pick the following initial configuration:

\begin{itemize}
\item Assign pheromone value on the edges $(s,s_1),(s,s_2), (d_1, d),(d_2, d)$ such that the normalized pheromone level on $(s,s_1),(d_1, d)$ is $\leq r$, that is $\npsso(0),\nbp{d_1d}(0) \leq r$. Further assign the flow values on the edges such that they satisfy the following,
\begin{equation}
     \fflow{v_{i}v_{i+1}}(0) \leq \fflow{s}\cdot r \cdot\prod_{  j \leq i} (1-l_{v_{j}}) \text{  and }\bflow{v_{i}v_{i+1}}(0) \leq \bflow{d}\cdot r\cdot\prod_{j \geq i+1} (1-l_{v_{j}})
\end{equation}
\begin{equation}
    \fflow{u_{i}u_{i+1}}(0) \geq \fflow{s}\cdot (1-r)\cdot\prod_{j \leq i} (1-l_{u_{j}}) \text{ and }\bflow{u_{i}u_{i+1}}(0) \geq \bflow{d}\cdot (1-r)\cdot\prod_{j \geq i+1} (1-l_{u_{j}})~.
\end{equation}
(Since we do not assume any leakage at the terminal vertices while defining the path leakage (see definition \ref{def:path_leakage}), we set $l_{v_0} = l_{u_0} = l_{v_{n+1}} = l_{u_{m+1}}=0$ in the expressions above.  )
\item Let the leakage value at each vertex be such that the parameters $\alpha$ and $\beta$ satisfy the following inequality, 
$\frac{\lo}{\lt} \leq 1+ \frac{\cue}{4 \cdot r} \cdot  \min\left(\frac{\fflow{s}}{\bflow{d}},\frac{\bflow{d}}{\bflow{s}}\right)$, where $\cue \defeq r-g(r)>0$. As the forward and backward flow is fixed, the above constraint on the leakage parameters only depend on $g$ and therefore satisfy the conditions of the theorem.
\end{itemize}
We show by induction that the above inequalities on flow and pheromone levels hold for all time $t\geq 0$ and therefore the system does not converge to the minimum leakage path.

\paragraph{Hypothesis:} At time $t\geq 0$, pheromone values on the edges $(s,s_1),(s,s_2), (d_1, d),(d_2, d)$ are such that the normalized pheromone level on $(s,s_1),(d_1, d)$ is $\leq r$, that is $\npsso(t),\nbp{d_1d}(t) \leq r$. Further the flow values on the edges satisfy the following,
\begin{equation}
    \fflow{v_{i}v_{i+1}}(t) \leq \fflow{s}\cdot r\cdot\prod_{j \leq i} (1-l_{v_{j}}) \text{ and }\bflow{v_{i}v_{i+1}}(t) \leq \bflow{d}\cdot r\cdot\prod_{j \geq i+1} (1-l_{v_{j}})
\end{equation}
\begin{equation}
    \fflow{u_{i}u_{i+1}}(t) \geq \fflow{s} \cdot(1-r)\cdot\prod_{j \leq i} (1-l_{u_{j}}) \text{ and }\bflow{u_{i}u_{i+1}}(t) \geq \bflow{d} \cdot(1-r)\cdot\prod_{j \geq i+1} (1-l_{u_{j}})~.
\end{equation}

\paragraph{Base case:} The conditions trivially hold at time $0$ because of the initial setting of flow and pheromone levels described above. 

\paragraph{Induction Step:} To prove the hypothesis for time $t+1$, all we need to show is that, $\npsso(t+1),\npddo(t+1) \leq r$, $\fflow{v_0v_1}(t+1)=\fflow{ss_1}(t+1) \leq \fflow{s} \cdot r$ and $\bflow{v_{n}v_{n+1}}(t+1)=\bflow{d_1d}(t+1) \leq \bflow{d} \cdot r$; all the remaining inequalities follow from these basic inequalities. Also note that from the definition of our case, that is $g(r)<r$, we get that the inequalities $\npsso(t+1),\npddo(t+1) \leq r$ further imply the following,
$$\fflow{ss_1}(t+1) = \fflow{s} \cdot g(\npsso(t+1)) \leq  \fflow{s}\cdot g(r) = \fflow{s}\cdot r~,$$ 
and similarly, 
$$\bflow{d_1d}(t+1) < \bflow{d} \cdot r~.$$
In the above we used monotonically non-decreasing property of decision rule $g$. Therefore it is enough to show that $\npsso(t+1),\npddo(t+1) \leq r$. As the proof for the bound on $\npddo(t+1)$ is analogous to that of $\npsso(t+1)$, in the remainder we focus our attention towards the proof for $\npsso(t+1)$. Recall the definition of $\npsso(t+1)$,
\begin{align}
\npsso(t+1)=\frac{\psso(t+1)}{\psso(t+1)+\psst(t+1)} &= \frac{ \psso(t) + 
    \fsso(t) +  \fsos(t)}{ \psst(t) + \psso(t)+ \fsst(t) +\fsso(t) +  \fsts(t)+\fsos(t)}
\end{align}
We know by the induction step that,
\begin{align}\label{eq:P41}
\frac{\psso(t)}{\psst(t)+\psso(t)} < r~. 
\end{align}
Also note that,
\begin{equation}\label{eq:P42}
    \frac{\fsos(t)}{ \fsts(t)+\fsos(t)}=\frac{\frac{\fsos(t)}{\fsts(t)}}{ 1+\frac{\fsos(t)}{\fsts(t)}} \leq \frac{\frac{r \alpha}{(1-r) \beta}}{1+\frac{r \alpha}{(1-r) \beta}}=\frac{r \alpha}{\beta+r(\alpha-\beta)}  \leq \frac{\alpha}{\beta} \cdot r~.
\end{equation}
In the above we used the monotonically non-decreasing property of $x/(1+x)$ and the conditions provided by the induction step at time $t$, that is $\fsos(t)\leq \bflow{d} \cdot r \cdot \alpha$ and $\fsts(t)\geq \bflow{d} \cdot (1-r)\cdot \beta$; which implies $\frac{\fsos(t)}{\fsts(t)} \leq \frac{r \alpha}{(1-r) \beta}$. In the fourth inequality, we used $\alpha-\beta \geq 0$ and $r\geq 0$.

Let $\cue \defeq r-g(r)>0$ and note that we have the following upper bound on the flow value,
\begin{equation}\label{eq:P43}
    \fsso(t) = \fflow{s} \cdot g(\npsso(t)) \leq  \fflow{s}\cdot g(r) = \fflow{s}\cdot(r-\cue)~.
\end{equation}
In the above we used $\npsso(t) \leq r$ and monotonicity of decision rule $g$.
Using these bounds, we provide an upper bound for the normalized pheromone level.
\begin{align}
\npsso(t+1)&= \frac{ \psso(t) + 
    \fsso(t) +  \fsos(t)}{ \psst(t) + \psso(t)+ \fsst(t) +\fsso(t) +  \fsts(t)+\fsos(t)},\\
    &\leq \frac{ r\cdot (\psst(t) + \psso(t)) +\fsso(t)+r\cdot\frac{\alpha}{\beta}\cdot (\fsts(t)+\fsos(t))}{ \psst(t) + \psso(t)+ \fsst(t) +\fsso(t) +  \fsts(t)+\fsos(t)},\\
    & = r+\frac{r\cdot\left(\frac{\alpha}{\beta}-1\right)\cdot (\fsts(t)+\fsos(t))+\fsso(t)-r(\fsso(t) +  \fsst(t))}{ \psst(t) + \psso(t)+ \fsst(t) +\fsso(t) +  \fsts(t)+\fsos(t)},\\
    & \leq r+\frac{2r\bflow{d}\cdot\left(\frac{\alpha}{\beta}-1\right)+ \fflow{s}\cdot (r-\cue)-r\cdot \fflow{s}}{ \psst(t) + \psso(t)+ \fsst(t) +\fsso(t) +  \fsts(t)+\fsos(t)},\\
    & \leq r+\frac{2r\bflow{d}\cdot\left(\frac{\alpha}{\beta}-1\right)-\cue \fflow{s}}{ \psst(t) + \psso(t)+ \fsst(t) +\fsso(t) +  \fsts(t)+\fsos(t)},\\
    & \leq r-\frac{\frac{1}{2}\cue \cdot 
    \fflow{s}}{ \psst(t) + \psso(t)+ \fsst(t) +\fsso(t) +  \fsts(t)+\fsos(t)},\\
    &\leq r~.
\end{align} 
In the second inequality we used \Cref{eq:P41,eq:P42}. In the third equality, we rearranged the terms. In the fourth inequality, we used $\fsts(t),\fsos(t) \leq \bflow{d}$,  $\fsso(t) +  \fsst(t) = \fflow{s}$ and \Cref{eq:P43}. In the fifth inequality, we simplified the expression. The sixth inequality follows because $\alpha$, $\beta$, $\fflow{s}$ and $\bflow{d}$ satisfy $\frac{\lo}{\lt} \leq 1+ \frac{\cue}{4 \cdot r} \cdot  \min\left(\frac{\fflow{s}}{\bflow{d}},\frac{\bflow{d}}{\fflow{s}}\right)$. Therefore, the previous derivation gives us,
$$\npsso(t+1) \leq r~,$$ 
and we conclude the first case.

\paragraph{Case 2:} Suppose $g(r)> r$, then pick the following initial configuration:
\begin{itemize}
\item Assign pheromone value on the edges $(s,s_1),(s,s_2), (d_1, d),(d_2, d)$ such that the normalized pheromone level on $(s,s_2),(d_2, d)$ is $\geq r$, that is $\npsst(0),\nbp{d_2d}(0) \geq r$. Further assign the flow values on the edges such that they satisfy the following,
\begin{equation}
    \fflow{v_{i}v_{i+1}}(0) \leq \fflow{s}\cdot (1-r)\cdot\prod_{j \leq i} (1-l_{v_{j}}) \text{ and }\bflow{v_{i}v_{i+1}}(0) \leq \bflow{d}\cdot (1-r)\cdot\prod_{j \geq i+1} (1-l_{v_{j}})
\end{equation}
\begin{equation}
    \fflow{u_{i}u_{i+1}}(0) \geq \fflow{s} \cdot r\cdot\prod_{j \leq i} (1-l_{u_{j}}) \text{ and }\bflow{u_{i}u_{i+1}}(0) \geq \bflow{d}\cdot r\cdot\prod_{j \geq i+1} (1-l_{u_{j}})~.
\end{equation}
(Since we do not assume any leakage at the terminal vertices while defining the path leakage (see definition \ref{def:path_leakage}), we set $l_{v_0} = l_{u_0} = l_{v_{n+1}} = l_{u_{m+1}}=0$ in the expressions above.  )
\item Let the leakage value at each vertex be such that the parameters $\alpha$ and $\beta$ satisfy the following inequality, 
$1 \geq \frac{\lt}{\lo} \geq 1- \frac{\cue}{4 \cdot r}\cdot\min\left( \frac{\fflow{s}}{ \bflow{d}},\frac{\bflow{d}}{ \bflow{s}}\right)$, where $\cue \defeq g(r)-r>0$. As the forward and backward flow is fixed, the above constraint on the leakage parameters only depend on $g$ and therefore satisfy the conditions of the theorem.
\end{itemize}
We show by induction that the above inequalities on flow and pheromone levels hold for all time $t\geq 0$ and therefore the system does not converge to the minimum leakage path.
\paragraph{Hypothesis:} At time $t\geq 0$, pheromone values on the edges $(s,s_1),(s,s_2), (d_1, d),(d_2, d)$ are such that the normalized pheromone level on $(s,s_2),(d_2, d)$ is $\geq r$, that is $\npsst(t),\nbp{d_2d}(t) \geq r$. Further the flow values on the edges satisfy the following,
\begin{equation}
    \fflow{v_{i}v_{i+1}}(t) \leq \fflow{s} \cdot (1-r)\cdot\prod_{j \leq i} (1-l_{v_{j}}) \text{ and }\bflow{v_{i}v_{i+1}}(t) \leq \bflow{d} \cdot(1-r)\cdot\prod_{j \geq i+1} (1-l_{v_{j}})
\end{equation}
\begin{equation}
    \fflow{u_{i}u_{i+1}}(t) \geq \fflow{s} \cdot r\cdot\prod_{j \leq i} (1-l_{u_{j}}) \text{ and }\bflow{u_{i}u_{i+1}}(t) \geq \bflow{d} \cdot r\cdot\prod_{j \geq i+1} (1-l_{u_{j}})~.
\end{equation}

\paragraph{Base case:} The conditions trivially hold at time $0$ because of the initial setting of flow and pheromone levels described above.

\paragraph{Induction Step:} To prove the hypothesis for time $t+1$, all we need to show is that, $\npsst(t+1),\npddt(t+1) \geq r$, $\fflow{u_0u_1}(t+1)=\fflow{ss_2}(t+1) \geq \fflow{s} \cdot  r$ and $\bflow{u_{m}u_{m+1}}(t+1)=\bflow{d_2d}(t+1) \geq \bflow{d} \cdot r$; all the remaining inequalities follows from these basic inequalities. Also, from the definition of our case, that is $g(r)>r$, we get that the inequalities $\npsst(t+1),\npddt(t+1) \geq r$ further imply $\fflow{ss_2}(t+1) \geq \fflow{s}\cdot r$ and $\bflow{d_2d}(t+1) \geq \bflow{d} \cdot r $. 

To see this, note that 
$$\fflow{ss_2}(t+1) = \fflow{s} \cdot g(\npsst(t+1)) \geq  \fflow{s}\cdot g(r) \geq \fflow{s}\cdot r~,$$ 
when $\npsst(t+1) \leq 1/2$
and,
$$\fflow{ss_2}(t+1) = \fflow{s} \cdot (1 - g(1 -\npsst(t+1))) \geq \fflow{s} \cdot (1 - g(1/2)) = \fflow{s}\cdot \frac{1}{2} \geq  \fflow{s}\cdot r~,$$
when $\npsst(t+1) \geq 1/2$, where we used monotonically non-decreasing property of $g$ for the above inequalities. Also, we used $g(1/2) = 1/2$ and $r \leq 1/2$ for the last inequality. Here, we had to consider two cases because function $g$ takes as input the minimum of the two normalized pheromone levels. Similarly we can show $\bflow{d_2d}(t+1) \geq \bflow{d} \cdot r~.$

Therefore it is enough to show that $\npsst(t+1),\npddt(t+1) \geq r$. As the proof for the bound on $\npddt(t+1)$ is analogous to that of $\npsst(t+1)$, in the remainder we focus our attention on the proof for $\npsst(t+1)$. Recall the definition of $\npsst(t+1)$,
\begin{align}
\npsst(t+1)=\frac{\psst(t+1)}{\psso(t+1)+\psst(t+1)} &= \frac{ \psst(t) + 
    \fsst(t) +  \fsts(t)}{ \psso(t) + \psst(t)+ \fsso(t) +\fsst(t) +  \fsos(t)+\fsts(t)}
\end{align}
We know by the induction step that,
\begin{align}\label{eq:P411}
\frac{\psst(t)}{\psso(t)+\psst(t)} \geq  r~. 
\end{align}
Also note that,
\begin{equation}\label{eq:P412}
\frac{\fsts(t)}{ \fsos(t)+\fsts(t)}=\frac{\frac{\fsts(t)}{\fsos(t)}}{ 1+\frac{\fsts(t)}{\fsos(t)}} \geq \frac{\frac{r \beta }{(1-r)\alpha}}{1+\frac{r \beta }{(1-r)\alpha}} =\frac{r \beta}{\alpha-r(\alpha -\beta)}\geq \frac{\beta}{\alpha} \cdot r~.
\end{equation}
In the above we used the monotonically non-decreasing property of $x/(1+x)$ and the conditions provided by the induction step at time $t$, that is $\fsts(t)\geq \bflow{d} \cdot  r \cdot \beta$ and $\fsos(t)\leq \bflow{d} \cdot (1-r)\cdot \alpha$; which implies $\frac{\fsts(t)}{\fsos(t)} \geq \frac{r \beta }{(1-r)\alpha}$. In the fourth inequality, we used $\alpha-\beta \geq 0$ and $1 \geq r\geq 0$.

Let $\cue \defeq g(r)-r>0$.  We have the following upper bound on the flow value,
\begin{equation}\label{eq:P413}
    \fsst(t)  \geq   \fflow{s}\cdot(r+\cue)~,
\end{equation}
To see this, note that
\begin{equation*}
    \fsst(t) = \fflow{s} \cdot g(\npsst(t)) \geq  \fflow{s}\cdot g(r) = \fflow{s}\cdot(r+\cue)~,
\end{equation*}
when $\npsst(t) \leq 1/2$. In the above we used $\npsst(t) \geq r$ and monotonicity of decision rule $g$. And
\begin{equation*}
    \fsst(t) = \fflow{s} \cdot (1- g(1-\npsst(t))) \geq \fflow{s} \cdot (1- g(1/2)) = \fflow{s} \cdot \frac{1}{2} = \fflow{s} \cdot g(1/2)  \geq \fflow{s}\cdot g(r) = \fflow{s}\cdot(r+\cue)~,
\end{equation*}
when $\npsst(t) > 1/2$. In the above, we used $\npsst(t) > 1/2$, monotonicity of decision rule $g$, $g(1/2) = 1/2$ and $r \leq 1/2$.

Using these bounds, we provide an upper bound for the normalized pheromone level.
\begin{align}
\npsst(t+1)&= \frac{  \psst(t)+ 
    \fsst(t) +  \fsts(t)}{ \psso(t) + \psst(t)+ \fsso(t) +\fsst(t) +  \fsos(t)+\fsts(t)},\\
    & \geq \frac{ r\cdot (\psso(t)+\psst(t)) + 
    \fsst(t) +  \frac{\beta}{\alpha}\cdot r \cdot (\fsos(t)+\fsts(t))}{ \psso(t) + \psst(t)+ \fsso(t) +\fsst(t) +  \fsos(t)+\fsts(t)}\\
    & = r+\frac{(\frac{\lt}{\lo}-1)\cdot r\cdot(\fsts(t)+\fsos(t))+
    \left( \fsst(t)- r \cdot (\fsso(t) +  \fsst(t))\right)}{ \psso(t) + \psst(t)+ \fsso(t) +\fsst(t) +  \fsos(t)+\fsts(t)},\\
    & \geq r+\frac{(\frac{\lt}{\lo}-1)\cdot r \cdot 2 \bflow{d}+
    \left( \fsst(t)- r\cdot \fflow{s}\right)}{ \psso(t) + \psst(t)+ \fsso(t) +\fsst(t) +  \fsos(t)+\fsts(t)},\\
    & \geq r+\frac{(\frac{\lt}{\lo}-1)\cdot r \cdot 2 \bflow{d}+
    \fflow{s}\left( r+\cue-r\right)}{ \psso(t) + \psst(t)+ \fsso(t) +\fsst(t) +  \fsos(t)+\fsts(t)},\\
    & = r+\frac{(\frac{\lt}{\lo}-1)\cdot r \cdot 2 \bflow{d}+\cue
    \fflow{s}}{ \psso(t) + \psst(t)+ \fsso(t) +\fsst(t) +  \fsos(t)+\fsts(t)},\\
    &\geq r+\frac{\frac{1}{2}\cue
    \fflow{s}}{ \psso(t) + \psst(t)+ \fsso(t) +\fsst(t) +  \fsos(t)+\fsts(t)}~,\\
    &\geq r~.
\end{align}
In the second inequality we used \Cref{eq:P411,eq:P412}. In the third equality, we rearranged terms. In the fourth inequality, we used $\frac{\beta}{\alpha}\leq 1$ and $\fsts(t),\fsos(t) \leq \bflow{d}$, $\fsso(t) +  \fsst(t) \leq \fflow{s}$ inequalities. In the fifth inequality, we used \Cref{eq:P413}.
The seventh inequality holds because $(\frac{\lt}{\lo}-1)\cdot r \cdot 2 \bflow{d} \geq -\frac{1}{2}\cue \fflow{s}$, which is equivalent to $\frac{\lt}{\lo} \geq 1- \frac{\cue}{4 \cdot r} \frac{\fflow{s}}{\bflow{d}}$. Note that this constraint is satisfied by our choice of values for leakage parameters. Further the previous derivation gives us,
$$\npsst(t+1) \geq r~,$$
and we conclude the second case.
\end{proof}

\subsection{Necessity of the linear rule for convergence to the shortest path (\Cref{prop:nl-incflow})}

\thirdprop*
\begin{proof}[Proof of \Cref{prop:nl-incflow}]
Let $n$ and $m$ be the number of vertices between $s$ and $d$ on paths $P_1$ and $P_2$ respectively. Without loss of generality we let $P_{1}$ to be the shortest path (that is, $n < m$). We name the vertices from left to right on path $P_1$ by $v_{0}$ to $v_{n+1}$ and $P_2$ by $u_0$ to $u_{m+1}$, with the convention $v_0=u_0=s$ and $v_{n+1}=u_{m+1}=d$. We also let $s_1=v_1,s_2=u_1,d_1=v_n$ and $d_2=u_m$. We divide the analysis into two cases.

For any non-linear $g \in \setflb$, there exists an $r \in (0,1/2)$ such that $g(r)\neq r$. For such an $r$, one of the following two conditions holds:
$g(r)< r$ or $g(r)>r$.
\paragraph{Case 1:} Suppose $g(r)< r$, then we pick the following initial configuration and flow values:
\begin{itemize}
\item Assign pheromone values on the edges $(s,s_1),(s,s_2), (d_1, d),(d_2, d)$ such that the normalized pheromone level on $(s,s_1),(d_1, d)$ is $\leq r$, that is $\npsso(0),\nbp{d_1d}(0) \leq r$. Further assign the flow values on the edges such that they satisfy the following,
\begin{equation}
    \fflow{v_{i}v_{i+1}}(0) \leq \fflow{s}(0) \cdot r \text{ and }\bflow{v_{i}v_{i+1}}(0) \leq \bflow{d}(0) \cdot r
\end{equation}
\begin{equation}
    \fflow{u_{i}u_{i+1}}(0) \geq \fflow{s}(0) \cdot(1-r) \text{ and }\bflow{u_{i}u_{i+1}}(0) \geq \bflow{d}(0) \cdot (1-r)~.
\end{equation}
\item Let $\fflow{s}(0) = \bflow{d}(0) = 1$. We specify the values of flow $\fflow{s}(t),\bflow{d}(t)$ for times $t>0$ later in the proof. 
\end{itemize}
We show by induction that the above inequalities on flow and pheromone levels hold for all time $t\geq 0$ and therefore the system does not converge to the shortest path.
\paragraph{Hypothesis:} At time $t\geq 0$, pheromone value on the edges $(s,s_1),(s,s_2), (d_1, d),(d_2, d)$ are such that the normalized pheromone level on $(s,s_1),(d_1, d)$ is $\leq r$, that is $\npsso(t),\nbp{d_1d}(t) \leq r$. Further the flow values on the edges satisfy the following,
\begin{equation}
    \fflow{v_{i}v_{i+1}}(t) \leq \fflow{s}(t-i) \cdot r \text{ and }\bflow{v_{i}v_{i+1}}(t) \leq \bflow{d}(t-(n-i)) \cdot r
\end{equation}
\begin{equation}
    \fflow{u_{i}u_{i+1}}(t) \geq \fflow{s}(t-i) \cdot(1-r) \text{ and }\bflow{u_{i}u_{i+1}}(t) \geq \bflow{d}(t-(m-i)) \cdot (1-r)~.
\end{equation}
In the above equations, $\fflow{s}(t')\defeq\fflow{s}(0)$ and $\bflow{d}(t')\defeq\bflow{d}(0)$ for all $t' <0$.
\paragraph{Base case:} The conditions trivially hold at time $0$ because of the initial setting of flow and pheromone levels described above.

\paragraph{Induction Step:} To prove the hypothesis for time $t+1$, all we need to show is that, $\npsso(t+1),\npddo(t+1) \leq r$, $\fflow{v_0v_1}(t+1)=\fflow{ss_1}(t+1) \leq \fflow{s}(t+1) \cdot  r$ and $\bflow{v_{n}v_{n+1}}(t+1)=\bflow{d_1d}(t+1) \leq \bflow{d}(t+1) \cdot r$; all the remaining inequalities follow from these basic inequalities. Also note that from the definition of our case, that is $g(r)<r$, we get that the inequalities $\npsso(t+1),\npddo(t+1) \leq r$ further imply the following,
$$\fflow{ss_1}(t+1) = \fflow{s}(t+1) \cdot g(\npsso(t+1)) \leq  \fflow{s}(t+1)\cdot g(r) < \fflow{s}(t+1)\cdot r~,$$ 
and similarly, 
$$\bflow{d_1d}(t+1) < \bflow{d}(t+1) \cdot r~.$$

In the above we used monotonically non-decreasing property of decision rule $g$.
Therefore it is enough to show that $\npsso(t+1),\npddo(t+1) \leq r$. As the proof for the bound on $\npddo(t+1)$ is analogous to that of $\npsso(t+1)$, in the remainder we focus our attention on the proof for $\npsso(t+1)$. Recall the definition of $\npsso(t+1)$,
\begin{align}
\npsso(t+1)=\frac{\psso(t+1)}{\psso(t+1)+\psst(t+1)} &= \frac{ \psso(t) + 
    \fsso(t) +  \fsos(t)}{ \psst(t) + \psso(t)+ \fsst(t) +\fsso(t) +  \fsts(t)+\fsos(t)}
\end{align}
We know by the induction step that,
\begin{align}\label{eq:P51}
\frac{\psso(t)}{\psst(t)+\psso(t)} \leq  r~. 
\end{align}
Also note that,
\begin{equation}\label{eq:P52}
\frac{\fsos(t)}{ \fsts(t)+\fsos(t)} = \frac{\frac{\fsos(t)}{\fsts(t)}}{ 1+\frac{\fsos(t)}{\fsts(t)}} \leq \frac{\frac{r\cdot \bflow{d}(t-n)}{(1-r)\cdot \bflow{d}(t-m)}}{1+\frac{r\cdot \bflow{d}(t-n)}{(1-r)\cdot \bflow{d}(t-m)}}= \frac{r\cdot \bflow{d}(t-n)}{\bflow{d}(t-m)+r\cdot (\bflow{d}(t-n)-\bflow{d}(t-m))}\leq  \frac{\bflow{d}(t-n)}{\bflow{d}(t-m)}\cdot r~.    
\end{equation}
In the above we used the monotonically non-decreasing property of $x/(1+x)$ and the conditions provided by the induction step at time $t$, that is $\fsos(t)\leq  r \cdot \bflow{d}(t-n)$ and $\fsts(t)\geq (1-r)\cdot \bflow{d}(t-m)$; which implies that $\frac{\fsos(t)}{\fsts(t)} \leq \frac{r\cdot \bflow{d}(t-n)}{(1-r)\cdot \bflow{d}(t-m)}$. In the final inequality, we used $0 \leq r \leq 1$ and $\bflow{d}(t-n)-\bflow{d}(t-m) \geq 0$. Note that $\bflow{d}(t-n)-\bflow{d}(t-m) \geq 0$ follows because $n \leq m$ and the fact that the incoming flow is non-decreasing.

\renewcommand{\lo}{\bflow{d}(t-n)}
\renewcommand{\lt}{\bflow{d}(t-m)}
Let $\cue \defeq  r-g(r)>0$, then we have following upper bound on the flow value,
\begin{equation}\label{eq:P53}
    \fsso(t) = \fflow{s}(t) \cdot g(\npsso(t)) \leq  \fflow{s}(t)\cdot g(r) = \fflow{s}(t)\cdot(r-\cue)~,
\end{equation}
In the above we used $\npsso(t) \leq r$ and monotonicity of decision rule $g$.
Using these bounds, we provide an upper bound for the normalized pheromone level.
\begin{align}
\npsso(t+1)&= \frac{ \psso(t) + 
    \fsso(t) +  \fsos(t)}{ \psst(t) + \psso(t)+ \fsst(t) +\fsso(t) +  \fsts(t)+\fsos(t)},\\
    & \leq \frac{ r\cdot (\psst(t) + \psso(t)) + 
    \fsso(t) +  \frac{\lo}{\lt}\cdot r\cdot(\fsts(t)+\fsos(t))}{ \psst(t) + \psso(t)+ \fsst(t) +\fsso(t) +  \fsts(t)+\fsos(t)},\\
    & =  r+\frac{(\frac{\lo}{\lt}-1)\cdot r\cdot(\fsts(t)+\fsos(t))+
    \left( \fsso(t)- r(\fsso(t) +  \fsst(t))\right)}{ \psst(t) + \psso(t)+ \fsst(t) +\fsso(t) +  \fsts(t)+\fsos(t)},\\
    & \leq  r+\frac{(\frac{\lo}{\lt}-1)\cdot r\cdot(\bflow{d}(t-m)+\bflow{d}(t-n))+
    \fflow{s}(t)\left( r-\cue-r\right)}{ \psst(t) + \psso(t)+ \fsst(t) +\fsso(t) +  \fsts(t)+\fsos(t)},\\
    & \leq  r+\frac{(\frac{\lo}{\lt}-1)\cdot r\cdot(\bflow{d}(t-m)+\bflow{d}(t-n))-\cue
    \fflow{s}(t)}{ \psst(t) + \psso(t)+ \fsst(t) +\fsso(t) +  \fsts(t)+\fsos(t)},\\
    &\leq  r-\frac{\frac{1}{2}\cue
    \fflow{s}(t)}{ \psst(t) + \psso(t)+ \fsst(t) +\fsso(t) +  \fsts(t)+\fsos(t)}~,\\
    &\leq  r~.
\end{align}
In the second inequality we used \Cref{eq:P51,eq:P52}. In the third equality, we rearranged terms. In the fourth inequality, we used the following inequalities: $\frac{\lo}{\lt}-1 \geq 0$, $\fsts(t)\leq \bflow{d}(t-m)$, $\fsos(t) \leq \bflow{d}(t-n)$, $\fsso(t) +  \fsst(t) \leq \fflow{s}(t)$ and \Cref{eq:P53}. In the fifth inequality, we simplified the expression.
The sixth inequality holds if $(\frac{\lo}{\lt}-1)\cdot r\cdot(\bflow{d}(t-m)+\bflow{d}(t-n)) \leq \frac{1}{2}\cue \fflow{s}(t)$. 

We will set functions $\fflow{s}(t)$ and $\bflow{d}(t)$ such that the this condition is satisfied. Suppose we set $\fflow{s}(t) = \bflow{d}(t) = (1 + \alpha)^t$ for some $\alpha \geq 0$. Then the 
condition above is satisfied if $\alpha$ satisfies
\begin{align}
   (1+\alpha)^{m-n} \leq 1+\frac{\cue}{2 \cdot r} \frac{1}{(1+\alpha)^{-n}+(1+\alpha)^{-m}}=1+\frac{\cue}{2 \cdot r} \frac{(1+\alpha)^{m}}{1+(1+\alpha)^{m-n}} 
\end{align}

Suppose we set $\alpha$ small enough such that $(1+\alpha)^{m-n} <1+ \frac{1}{3}\min(\frac{\cue}{2\cdot r},1)$, then $\alpha$ satisfies

% For the symmetric case, that is $\fflow{s}(t)=\bflow{d}(t)=f(t)$ for all $t \geq 0$ the above condition simplifies to $\frac{f(t-n)}{f(t-m)} \leq 1+\frac{\cue}{2\cdot r} (\frac{f(t)}{f(t-n)+f(t-m)})$. \kiran{Shivam read this} Let $f(t)=(1+\alpha)^{t}$ for all $t\geq 0$ and $\alpha \geq 0$. The condition $\frac{f(t-n)}{f(t-m)} \leq 1+\frac{\cue}{2\cdot r} (\frac{f(t)}{f(t-n)+f(t-m)})$ is met if $\alpha$ satisfies the following, $(1+\alpha)^{m-n} \leq 1+\frac{\cue}{2 \cdot r} \frac{1}{(1+\alpha)^{-n}+(1+\alpha)^{-m}}=1+\frac{\cue}{2 \cdot r} \frac{(1+\alpha)^{m}}{1+(1+\alpha)^{m-n}}$. Suppose $(1+\alpha)^{m-n} <1+ \frac{1}{3}\min(\frac{\cue}{2\cdot r},1)$, then note that,
\begin{align}
(1+\alpha)^{m-n}\leq 1+ \frac{\cue}{2\cdot r} \cdot \frac{1}{3} \leq 1+  \frac{\cue}{2\cdot r} \cdot \frac{1}{1+(1+\alpha)^{m-n}} \leq 1+  \frac{\cue}{2\cdot r} \cdot \frac{(1+\alpha)^{m}}{1+(1+\alpha)^{m-n}}~.
\end{align}
In the first inequality we used $(1+\alpha)^{m-n} <1+ \frac{1}{3} \cdot \frac{\cue}{2\cdot r}$. In the second inequality we used $(1+\alpha)^{m-n} <1+ \frac{1}{3} \leq 2$. In the third inequality we used $\alpha \geq 0$, which implies $(1+\alpha)^{m} \geq 1$. Therefore there exists a setting of $\alpha$  which depends only on $g$, $n$ and $m$ such that when the incoming flow increases by a factor of $1 + \alpha$ at each time step,
 the previous analysis goes through and we get 
$$\npsso(t+1) \leq r~.$$
Therefore, by induction, the dynamics never converges to the shortest path $P_1$. We conclude the first case.

\paragraph{Case 2:} Suppose $g(r)> r$, then we pick the following initial configuration:
\begin{itemize}
\item Assign pheromone value on the edges $(s,s_1),(s,s_2), (d_1, d),(d_2, d)$ such that the normalized pheromone level on $(s,s_2),(d_2, d)$ is $\geq r$, that is $\npsst(0),\nbp{d_2d}(0) \geq r$. Further assign the flow values on the edges such that they satisfy the following,
\begin{equation}
    \fflow{v_{i}v_{i+1}}(0) \leq \fflow{s}(0)\cdot (1-r) \text{ and }\bflow{v_{i}v_{i+1}}(0) \leq \bflow{d}(0) \cdot(1-r)
\end{equation}
\begin{equation}
    \fflow{u_{i}u_{i+1}}(0) \geq \fflow{s}(0) \cdot r \text{ and }\bflow{v_{i}v_{i+1}}(0) \geq \bflow{d}(0) \cdot r~.
\end{equation}
\item Let $\fflow{s}(0) = \bflow{d}(0) = 1$. We specify the values of flow $\fflow{s}(t),\bflow{d}(t)$ for times $t>0$ later in the proof.
\end{itemize}
We show by induction that the above inequalities on flow and pheromone levels hold for all time $t\geq 0$ and therefore the system does not converge to the shortest path.
\paragraph{Hypothesis:} At time $t\geq 0$, pheromone value on the edges $(s,s_1),(s,s_2), (d_1, d),(d_2, d)$ are such that the normalized pheromone level on $(s,s_2),(d_2, d)$ is $\geq r$, that is $\npsst(t),\nbp{d_2d}(t) \geq r$. Further the flow values on the edges satisfy the following,
\begin{equation}
    \fflow{v_{i}v_{i+1}}(t) \leq \fflow{s}(t-i) \cdot(1-r) \text{ and }\bflow{v_{i}v_{i+1}}(t) \leq \bflow{d}(t-(n-i)) \cdot(1-r)
\end{equation}
\begin{equation}
    \fflow{u_{i}u_{i+1}}(t) \geq \fflow{s}(t-i) \cdot r \text{ and }\bflow{v_{i}v_{i+1}}(t) \geq \bflow{d}(t-(m-i)) \cdot r~.
\end{equation}
In the above equations, $\fflow{s}(t')=\fflow{s}(0)$ and $\bflow{d}(t')=\bflow{d}(0)$ for all $t' <0$.
\paragraph{Base case:} The conditions trivially hold at time $0$ because of the initial setting of flow and pheromone levels described above.

\paragraph{Induction Step:} To prove the hypothesis for time $t+1$, all we need to show is that, $\npsst(t+1),\npddt(t+1) \geq r$, $\fflow{u_0u_1}(t+1)=\fflow{ss_2}(t+1) \geq \fflow{s}(t+1) \cdot  r$ and $\bflow{u_{m}u_{m+1}}(t+1)=\bflow{d_2d}(t+1) \geq \bflow{d}(t+1) \cdot r$; all the remaining inequalities follows from these basic inequalities. Also note that from the definition of our case, that is $g(r)>r$, we get that the inequalities $\npsst(t+1),\npddt(t+1) \geq r$
% further imply the following,
% $$\fflow{ss_2}(t+1) = \fflow{s}(t+1) \cdot g(\npsst(t+1)) \geq  \fflow{s}(t+1)\cdot g(r) > \fflow{s}(t+1)\cdot r~,$$ 
% and similarly, 
% $$\bflow{d_2d}(t+1) > \bflow{d}(t+1) \cdot r~.$$
% In the above we used monotonically non-decreasing property of decision rule $g$.
further imply $\fflow{ss_2}(t+1) \geq \fflow{s}(t+1)\cdot r$ and $\bflow{d_2d}(t+1) \geq \bflow{d}(t+1) \cdot r $. 

To see this, note that 
$$\fflow{ss_2}(t+1) = \fflow{s}(t+1) \cdot g(\npsst(t+1)) \geq  \fflow{s}(t+1)\cdot g(r) \geq \fflow{s}(t+1)\cdot r~,$$ 
when $\npsst(t+1) \leq 1/2$
and,
$$\fflow{ss_2}(t+1) = \fflow{s}(t+1) \cdot (1 - g(1 -\npsst(t+1))) \geq \fflow{s}(t+1) \cdot (1 - g(1/2)) = \fflow{s}(t+1)\cdot \frac{1}{2} \geq  \fflow{s}(t+1)\cdot r~,$$
when $\npsst(t+1) \geq 1/2$, where we used monotonically non-decreasing property of $g$ for the above inequalities. Also, we used $g(1/2) = 1/2$ and $r \leq 1/2$ for the last inequality. Here, we had to consider two cases because function $g$ takes as input the minimum of the two normalized pheromone levels. Similarly we can show $\bflow{d_2d}(t+1) \geq \bflow{d}(t+1) \cdot r~.$

Therefore it is enough to show that $\npsst(t+1),\npddt(t+1) \geq r$. As the proof for the bound on $\npddt(t+1)$ is analogous to that of $\npsst(t+1)$, in the remainder we focus our attention on the proof for $\npsst(t+1)$. Recall the definition of $\npsst(t+1)$,
\begin{align}
\npsst(t+1)=\frac{\psst(t+1)}{\psso(t+1)+\psst(t+1)} &= \frac{ \psst(t) + 
    \fsst(t) +  \fsts(t)}{ \psso(t) + \psst(t)+ \fsso(t) +\fsst(t) +  \fsos(t)+\fsts(t)}
\end{align}
We know by the induction step that,
\begin{align}\label{eq:P511}
\frac{\psst(t)}{\psso(t)+\psst(t)} \geq  r~. 
\end{align}
Also note that,
\begin{equation}\label{eq:P512}
\frac{\fsts(t)}{ \fsos(t)+\fsts(t)}=\frac{\frac{\fsts(t)}{\fsos(t)}}{1+\frac{\fsts(t)}{\fsos(t)}} \geq \frac{\frac{r \cdot \bflow{d}(t-m)}{(1-r)\cdot \bflow{d}(t-n)}}{1+\frac{r \cdot \bflow{d}(t-m)}{(1-r)\cdot \bflow{d}(t-n)}}=\frac{r \cdot \bflow{d}(t-m)}{\bflow{d}(t-n)-r\cdot (\bflow{d}(t-n)-\bflow{d}(t-m))} \geq \frac{\bflow{d}(t-m)}{\bflow{d}(t-n)}\cdot r~.
\end{equation}
In the above we used the monotonically non-decreasing property of $x/(1+x)$ and the conditions provided by the induction step at time $t$, that is $\fsts(t)\geq r \cdot \bflow{d}(t-m)$ and $\fsos(t)\leq (1-r)\cdot \bflow{d}(t-n)$; which implies $\frac{\fsts(t)}{\fsos(t)} \geq \frac{r \cdot \bflow{d}(t-m)}{(1-r)\cdot \bflow{d}(t-n)}$. In the final inequality we used $0 \leq r \leq 1$ and $\bflow{d}(t-n)-\bflow{d}(t-m)) \geq 0$, which follows because $n \leq m$ and the incoming flow is non-decreasing.

% Let $\cue \defeq g(r)- r>0$, then we have following lower bound on the flow value,
% \begin{equation}\label{eq:P513}
%     \fsst(t) = \fflow{s}(t) \cdot g(\npsst(t)) \geq  \fflow{s}(t)\cdot g(r) = \fflow{s}(t)\cdot(r+\cue)~,
% \end{equation}
% In the above we used $\npsst(t) \geq r$ and monotonicity of decision rule $g$.

Let $\cue \defeq g(r)-r>0$.  We have the following upper bound on the flow value,
\begin{equation}\label{eq:P513}
    \fsst(t)  \geq   \fflow{s}(t)\cdot(r+\cue)~,
\end{equation}
To see this, note that
\begin{equation*}
    \fsst(t) = \fflow{s}(t) \cdot g(\npsst(t)) \geq  \fflow{s}(t)\cdot g(r) = \fflow{s}(t)\cdot(r+\cue)~,
\end{equation*}
when $\npsst(t) \leq 1/2$. In the above we used $\npsst(t) \geq r$ and monotonicity of decision rule $g$. And
\begin{align*}
    \fsst(t) &= \fflow{s}(t) \cdot (1- g(1-\npsst(t))) \geq \fflow{s}(t) \cdot (1- g(1/2))\\ &= \fflow{s}(t) \cdot \frac{1}{2} = \fflow{s}(t) \cdot g(1/2)  \geq \fflow{s}(t)\cdot g(r) = \fflow{s}(t)\cdot(r+\cue)~,
\end{align*}
when $\npsst(t) > 1/2$. In the above, we used $\npsst(t) > 1/2$, monotonicity of decision rule $g$, $g(1/2) = 1/2$ and $r \leq 1/2$.

Using these bounds, we provide an upper bound for the normalized pheromone level.
\begin{align}
\npsst(t+1)&= \frac{ \psst(t) + 
    \fsst(t) +  \fsts(t)}{ \psso(t) + \psst(t)+ \fsso(t) +\fsst(t) +  \fsos(t)+\fsts(t)},\\
    &\geq  \frac{ r \cdot (\psso(t) + \psst(t))+   \fsst(t)+ \frac{\lt}{\lo}\cdot r\cdot(\fsts(t)+\fsos(t))
    }{ \psso(t) + \psst(t)+ \fsso(t) +\fsst(t) +  \fsos(t)+\fsts(t)},\\
    & = r+\frac{(\frac{\lt}{\lo}-1)\cdot r\cdot(\fsts(t)+\fsos(t))+
    \left( \fsst(t)- r\cdot(\fsso(t) +  \fsst(t))\right)}{ \psso(t) + \psst(t)+ \fsso(t) +\fsst(t) +  \fsos(t)+\fsts(t)},\\
    & \geq r+\frac{(\frac{\lt}{\lo}-1)\cdot r\cdot(\bflow{d}(t-m)+\bflow{d}(t-n))+
    \fflow{s}(t)\cdot\left( r+\cue-r\right)}{ \psso(t) + \psst(t)+ \fsso(t) +\fsst(t) +  \fsos(t)+\fsts(t)},\\
    & \geq r+\frac{(\frac{\lt}{\lo}-1)\cdot r\cdot(\bflow{d}(t-m)+\bflow{d}(t-n))+\cue
    \fflow{s}(t)}{ \psso(t) + \psst(t)+ \fsso(t) +\fsst(t) +  \fsos(t)+\fsts(t)},\\
    &\geq r+\frac{\frac{1}{2}\cue
    \fflow{s}(t)}{ \psso(t) + \psst(t)+ \fsso(t) +\fsst(t) +  \fsos(t)+\fsts(t)}~,\\
    &\geq r~.
\end{align}
In the second inequality we used \Cref{eq:P511,eq:P512}. In the third equality, we rearranged terms. In the fourth inequality, we used the following inequalities: $\frac{\lt}{\lo}-1\leq 0$, $\fsts(t)\leq \bflow{d}(t-m)$, $\fsos(t) \leq \bflow{d}(t-n)$,  $\fsso(t) +  \fsst(t) \leq \fflow{s}(t)$ and \Cref{eq:P513}. In the fifth inequality, we simplified the expression.
The sixth inequality holds if $(\frac{\lt}{\lo}-1)\cdot r \cdot (\bflow{d}(t-m)+\bflow{d}(t-n)) \geq -\frac{1}{2}\cue
    \fflow{s}(t)$.

We will set functions $\fflow{s}(t)$ and $\bflow{d}(t)$ such that the this condition is satisfied. Suppose we set $\fflow{s}(t) = \bflow{d}(t) = (1 + \alpha)^t$ for some $\alpha \geq 0$. Then the 
condition above is satisfied if $\alpha$ satisfies    
\begin{align}
    (1+\alpha)^{n-m} \geq 1-\frac{\cue }{2 \cdot r} \frac{1}{(1+\alpha)^{-m}+(1+\alpha)^{-n}} =1-\frac{\cue }{2 \cdot r} \frac{(1+\alpha)^{n}}{(1+\alpha)^{n-m}+1}.
\end{align}
% For the symmetric case, that is $\fflow{s}(t)=\bflow{d}(t)=f(t)$ for all $t \geq 0$ the above condition simplifies to $\frac{f(t-m)}{f(t-n)} \geq 1-\frac{\cue}{2 \cdot r} \cdot \frac{f(t)}{f(t-m)+f(t-n)}$. Let $f(t)=(1+\alpha)^{t}$ for all $t\geq 0$ and $\alpha \geq 0$. The condition $\frac{f(t-m)}{f(t-n)} \geq 1-\frac{\cue}{2 \cdot r} \cdot \frac{f(t)}{f(t-m)+f(t-n)}$ is met if $\alpha$ satisfies the following,
%     $(1+\alpha)^{n-m} \geq 1-\frac{\cue }{2 \cdot r} \frac{1}{(1+\alpha)^{-m}+(1+\alpha)^{-n}} =1-\frac{\cue }{2 \cdot r} \frac{(1+\alpha)^{n}}{(1+\alpha)^{n-m}+1}$. 
Recall that $n \leq m$. Suppose we set $\alpha$ small enough such that $(1+\alpha)^{n-m} \geq 1-\frac{1}{2} \min(\frac{\cue }{2 \cdot r},1)$, then note that, 
\begin{align}
    (1+\alpha)^{n-m} \geq 1-\frac{1}{2} \cdot \frac{\cue }{2 \cdot r} \geq 1- \frac{\cue }{2 \cdot r} \cdot \frac{1}{1+(1+\alpha)^{n-m}} \geq 1- \frac{\cue }{2 \cdot r} \cdot \frac{(1+\alpha)^{n}}{1+(1+\alpha)^{n-m}}
\end{align}
    In the first inequality, we used $(1+\alpha)^{n-m} \geq 1-\frac{1}{2} \min(\frac{\cue }{2 \cdot r},1)  \geq 1-\frac{1}{2} \frac{\cue }{2 \cdot r}$. In the second inequality we used, $\alpha \geq 0$, $n \leq m$; which implies $(1+\alpha)^{n-m} \leq 1$. In the final inequality, we used $\alpha \geq 0$, which implies $(1+\alpha)^{n} \geq 1$.
    Therefore there exists a setting of $\alpha$  which depends only on $g$, $n$ and $m$ such that when the incoming flow increases by a factor of $1 + \alpha$ at each time step,
 the previous analysis goes through and we get 
$$\npsst(t+1) \geq r~.$$
Therefore, by induction, the dynamics never converges to the shortest path $P_1$. We conclude the second case.
\end{proof}

\section{Simulation Details}
\label{sec:app_sim_details}
In Section \ref{sec:results}, we discussed various simulation results for linear and non-linear decision rules. We give more details of these simulations here.

\paragraph{Graph families considered.} We ran the simulations for three kinds of directed graphs: graphs sampled from the $G(n, p)$ model, graphs sampled from  the $G(n, p)$ model with the additional locality constraint that an edge can exist between vertex $i$ and $j$ only if $|i - j| \leq k$ for some parameter $k$ (with the source vertex and the destination the vertex numbered 1 and $n$ respectively), and the grid graph. $G(n, p)$ model with the additional constraint ensures that the graph has long paths between the source and the destination. We only considered instances where there was at least one path from the source to the destination.

For the $G(n, p)$ model and its locally constrained version, we consider two kinds of graph families: one where an edge is allowed from $i$ to $j$  only if $i < j$ (resulting in a DAG), and the other where the edges can go both ways. We will use $G(n, p)$ and $G(n, p)$ local to denote the $G(n, p)$ graph and it's locally constrained version where edges can go both ways, and $G(n, p)$ DAG and $G(n, p)$ local DAG to denote their acyclic versions respectively.

\subsection{Linear Decision Rule}
In Section \ref{sec:parallel_paths_ub}, we discussed that with linear decision rule, in the fixed incoming flow setting, the dynamics converges to the path with the minimum leakage. With increasing flow, and no leakage, the dynamics converges to the shortest path. We generated a large number of instances with different values for $n$, $p$, $k$ and other parameters and observed the desired convergence in all the simulated instances.  Below, we describe the details of the parameter settings we considered.

For all the simulations, the decay parameter $\delta$ was set to $0.9$.  The initial forward flow level at $s$ and backward flow level at $d$ was chosen uniformly at random from $(0.5,1)$, and the initial flow at all other vertices was set to $0$.  We ran each simulation instance with two settings for initial pheromone level: (i) setting initial pheromone levels at all edges to $1$, (ii) choosing initial pheromone level on the edges uniformly at random from $(0, 1)$.
%Leakage chosen randomly from (0, 1)
%Decay parameter delta fixed for all the edges, and chosen randomly from (0,1)
%n, p, k, number of iterations

\subsubsection{Leakage with fixed flow}

For all the randomly generated instance, we chose the leakage value of all the vertices uniformly at random from $(0,0.1)$. For all these instances, the dynamics converged to the path with the minimum leakage. Below, we describe the parameter values for different graph families.
\begin{enumerate}
    
    \item $G(n, p)$ and $G(n,p)$ DAG: We ran 1000 instances each with $(n,p)$ set to $(100, 0.05), (100, 0.1)$ and $(100, 0.5)$. We ran 100 instances each with $(n,p)$ set to $(1000, 0.01)$ and $(1000, 0.005)$.
    
    \item $G(n, p)$ local and $G(n, p)$ local DAG: We ran 1000 instances with $(n, p) = (100, 0.5)$ and window size of 10, and 10 instances with $(n, p) = (1000, 0.5)$ and window size of 40.
 
    \item Grid graph: We ran the simulations on a $10X10$ grid graph with source and destination at the diagonally opposite extreme vertices. We ran $100$ instances. The graph structure was same across all these instances, but other parameter values such as leakage levels, initial pheromone levels etc. were chosen randomly as described above.
\end{enumerate}

\subsubsection{Increasing flow with no leakage}
In this case, the leakage was set to $0$ for all the vertices. We increase the forward flow level at $s$ and backward flow level at $d$ by a factor of $1.1$ every time step. To avoid floating point overflow, instead of multiplying the forward and backward flow values at $s$ and $d$ by $1.1$, we keep these two values the same, and divide all other flow values and pheromone levels in the graph by a factor of $1.1$. For any decision rule that only depends on the normalized pheromone levels (which holds for the linear rule), it is not difficult to see that only the relative values of flow and pheromone level matter, and this gives rise to exactly the same dynamics (up to scaling) as when we multiply the forward  flow  at $s$ and backward flow at $d$ by $1.1$.  When the value of pheromone level or flow at any edge became too small (smaller than the minimum allowed value  for a 64 bit floating point number which is  $\approx 10^{-323}$), we rounded it to zero.

For all the graphs, we planted a short path in the graph so as to ensure that the shortest path is unique. In all the simulations, we observed that the dynamics converges to this shortest path. Below, we describe the details for different graph types.

\begin{enumerate}
    \item $G(n, p)$ and $G(n, p)$ DAG: We ran 1000 instances each with $(n,p)$ set to $(100, 0.05), (100, 0.1)$ and $(100, 0.5)$. We ran 100 instances each with $(n,p)$ set to $(1000, 0.01)$ and $(1000, 0.005)$. For all the instances, we plant a randomly generated shortest path of length 1 smaller than the previous shortest path, so that there is unique shortest path in the graph so obtained.

    \item $G(n, p)$ local and $G(n, p)$ local DAG: We ran 1000 instances with $(n, p) = (100, 0.5)$ and window size of 10, and 100 instances with $(n, p) = (1000, 0.5)$ and window size of 40. We planted the shortest path having edges $(1,k+r), (k+r, 2(k+r))
    ,(2(k+r),3(k+r))$   and so on. Here, $1$ is the source vertex and $n$ is the destination vertex, and $r$ is a random integer in $[1, \text{window size}]$.
    
    \item Grid graph: We ran the simulations on a $10X10$ grid graph with source and destination at the diagonally opposite extreme vertices.  We ran $100$ instances. For each instance, we planted a randomly generated shortest path of length randomly chosen from $\{9, 10, 11\}$.
\end{enumerate}

\subsection{Nonlinear Decision Rules}
In Section \ref{sec:general_rules_lb}, we showed that for a reasonably general family of  decision rules, the linear decision rule is its unique member with guaranteed convergence to the path with minimum leakage or the shortest path. Note that these results only suggest that linear decision rule is necessary for \emph{guaranteed convergence} to the shortest or the minimum leakage path. Can it be the case that even with non-linear decision rules, the forces of increasing flow and leakage still help in finding shorter or smaller leakage paths respectively, compared to the paths found in the absence of these forces? To understand this question, we ran simulations for various non-linear decision rules.
We observe that within each graph family, for a large fraction of graph instances, the path found in the presence of these forces has length (respectively leakage) smaller than or equal to the length (respectively leakage) of the path found in their absence. These simulations suggest that the usefulness of the forces of 
leakage and increasing flow is not limited to the linear decision rule. Below, we discuss more details of these simulations.

\subsubsection{Decision rules details}
We consider the following non-linear decision rules:
\begin{itemize}
    \item \textbf{Quadratic}: This rule divides the flow in proportion to the square of the pheromone levels.
    \item \textbf{1.1 power}: This rule divides the flow in proportion to the $1.1^{\text{th}}$ power of the pheromone levels. We study this rule to understand the effect of strength of non-linearity. Since this rule is closer to the linear rule compared to the quadratic rule, we would expect that the forces of leakage and increasing flow are more effective with this rule compared to the quadratic rule.
    \item \textbf{Quadratic-with-offset}: This rule is a slight variant of the quadratic rule and has been used previously \cite{ deneubourg1990self} to model ant behaviour. This rule adds a fixed positive constant $c$ to the pheromone levels, and divides the flow in proportion to the square of the offsetted pheromone levels. The parameter $c$ was added to encourage exploration.
   
   \item \textbf{Rank-edge}: This rule was introduced in \cite{chandrasekhar2018distributed}. This rule ranks the edges from highest to lowest pheromone levels. If multiple edges have the same pheromone level, they get the same rank. Let $q \in (0, 1)$ be some parameter. This rule sends $(1-q)$ fraction of flow to the first ranked edge, $q(1-q)$ fraction of flow to the second ranked edge, $q^2(1-q)$ fraction of flow to the third ranked edge and so on. In general, the $i^{\text{th}}$ ranked edge gets $q^{i-1}(1-q)$ fraction of flow.  If there are $k$ edges at the $i^{\text{th}}$ rank, then each of them gets $\frac{q^{i-1}(1-q)}{k}$ fraction of flow. One exception to this rule is the lowest ranked edge which gets all the remaining flow that has not been assigned to any other edge. If there are multiple edges at the lowest rank, then this remaining flow gets divided equally among them.
 \end{itemize}  
\subsubsection{Results} 
As in the linear decision rule case, we considered the $G(n, p)$, $G(n, p)$ local, $G(n, p)$ DAG, $G(n, p)$ local DAG and the grid graph for our simulations with non-linear decision rules. In the increasing flow case, similar to the linear decision rule setting, we add a random shortest path to these graphs so that the shortest path is unique. For each non-linear decision rule and each graph family, we ran simulations with 100 random graph instances.

% Similarly, Figure 4b shows the path leakage obtained by the quadratic decision rule with and without leakage present at the vertices, and by the linear decision rule with leakage present. There is a subtle distinction here between path leakage as an objective function and leakage as a process affecting the dynamics. For each graph instance, we assign leakage values to vertices (see Appendix C for details). This gives us a path leakage objective function which we measure in all the three cases. However, in the case of the quadratic rule without leakage, the leakage process is not applied at vertices during the dynamics.

In Table \ref{tab:leakage-non_lin-effect}, we show the effect of leakage with non-linear decision rules. There is a subtle distinction here between path leakage as an objective function and leakage as a process affecting the dynamics. For each graph instance, we assign leakage values to vertices.  This gives us a path leakage objective function for each graph instance. For each graph instance, we compare two settings, the baseline setting in which the leakage process is not applied at vertices during the dynamics and the main setting in which the leakage process is applied. Now for each graph instance, we compare the path leakage objective function of the path found in the baseline setting and the main setting. We report the fraction of instances where the path leakage objective in the main setting is unchanged or smaller as compared to the baseline setting. We also report the average percentage difference in the path leakage objective between the baseline and the main setting,  that is
\[
 \frac{1}{n} \sum_{i=1}^n 100*\frac{\text{path\_leakage}(\text{main}_i) - \text{path\_leakage}(\text{baseline}_i)}{\text{path\_leakage}(\text{baseline}_i)}.
\]
Here, $\text{path\_leakage}(\text{main}_i)$ and $\text{path\_leakage}(\text{baseline}_i)$ denote the path leakage objective obtained in the main setting and the baseline setting respectively, for the $i^{\text{th}}$ graph instance.
In the last column, we report the average percentage difference in path leakage for the baseline compared to the optimum (minimum) path leakage, that is
\[
 \frac{1}{n} \sum_{i=1}^n 100*\frac{\text{path\_leakage}(\text{opt}_i) - \text{path\_leakage}(\text{baseline}_i)}{\text{path\_leakage}(\text{baseline}_i)}.
\]
Here, $\text{path\_leakage}(\text{opt}_i)$ denotes the optimum path leakage for the $i^{\text{th}}$ graph instance.
Due to symmetry in the grid graph, the dynamics does not converge to a path in the baseline setting, therefore we do not consider grid graphs while studying the effect of leakage with non-linear decision rules.

% For each graph instance, we run simulations under two settings: the baseline setting in which there is no leakage at vertices, and the second setting in which there is leakage present at each vertex. Although there is no leakage at any vertex in the baseline setting, we still use the leakage values of the second setting to compute the leakage of the path found in this setting. The leakage of the path found in the baseline setting is then compared to the leakage of the path found in the second setting. We report the fraction of instances where the path leakage in the second setting is unchanged or smaller as compared to the baseline setting. We also report the average percentage difference in the leakage of path found between the baseline and the second setting. In the last column, we report the average percentage difference in path leakage for the baseline compared to the optimum (minimum) leakage path. Due to symmetry in the grid graph, the dynamics does not converge to a path in the baseline setting, therefore we do not consider grid graphs while studying the effect of leakage with non-linear decision rules.

In Table \ref{tab:inc_flow-non_lin-effect}, we show the effect of increasing flow with non-linear decision rules. Again, we consider the baseline setting where there is no increasing flow, and the main setting with increasing flow. We report the fraction of instances where the path length in the main setting is unchanged or smaller as compared to the baseline setting.
We also report the average percentage difference in the length of path found between the baseline and the main setting, that is
\[
 \frac{1}{n} \sum_{i=1}^n 100*\frac{\text{path\_length}(\text{main}_i) - \text{path\_length}(\text{baseline}_i)}{\text{path\_length}(\text{baseline}_i)}.
\]
Here, $\text{path\_length}(\text{main}_i)$ and $\text{path\_length}(\text{baseline}_i)$ denote the path length obtained in the main setting and the baseline setting respectively, for the $i^{\text{th}}$ graph instance.
In the last column, we report the average percentage difference in path length for the baseline compared to the shortest length path, that is 
\[
 \frac{1}{n} \sum_{i=1}^n 100*\frac{\text{path\_length}(\text{opt}_i) - \text{path\_length}(\text{baseline}_i)}{\text{path\_length}(\text{baseline}_i)}.
\]
Here, $\text{path\_length}(\text{opt}_i)$ denotes the shortest path length for the $i^{\text{th}}$ graph instance.

% \usepackage{multirow}
\begin{table}[]
\begin{tabular}{clcccc}
\hline
\multicolumn{1}{l}{\textbf{\begin{tabular}[c]{@{}l@{}}Decision rule\end{tabular}}} &
  \textbf{Graph family} &
  \multicolumn{1}{l}{\textbf{\begin{tabular}[c]{@{}l@{}}\% instances with\\ unchanged path leak.\end{tabular}}} &
  \multicolumn{1}{l}{\textbf{\begin{tabular}[c]{@{}l@{}}\% instances with\\ smaller path leak.\end{tabular}}} &
  \multicolumn{1}{l}{\textbf{\begin{tabular}[c]{@{}l@{}}Avg. \%\\ change\end{tabular}}} &
  \multicolumn{1}{l}{\textbf{\begin{tabular}[c]{@{}l@{}}Avg. \%\\ change opt.\end{tabular}}} \\ \hline
  \multirow{4}{*}{1.1 power} &
  $G(n, p)$ local &
  1 &
  98 &
  -27.5 &
  -43.9 \\
 &
  \multicolumn{2}{l}{$G(n, p)$ local DAG \hspace{34pt}     0} &
  100 &
  -38.1 &
  -47.7 \\
 &
  $G(n, p)$ &
  70 &
  30 &
  -12.3 &
  -26.1 \\
 &
  $G(n, p)$ DAG &
  88 &
  12 &
  -3.9 &
  -17.7 \\ \hline
\multirow{4}{*}{Quadratic} &
  $G(n, p)$ local &
  27 &
  67 &
  -8.2 &
  -44.3 \\
 &
  \multicolumn{2}{l}{$G(n, p)$ local DAG \hspace{31pt}     24} &
  73 &
  -12.5 &
  -47.2 \\
 &
  $G(n, p)$ &
  83 &
  16 &
  -4.5 &
  -24.8 \\
 &
  $G(n, p)$ DAG &
  94 &
  6 &
  -2.4 &
  -22.3 \\ \hline
\multirow{4}{*}{\begin{tabular}[c]{@{}c@{}}Quadratic\\ -with-\\ offset\end{tabular}} &
  $G(n, p)$ local &
  2 &
  79 &
  -10.5 &
  -43.1 \\
 &
  \multicolumn{2}{l}{$G(n, p)$ local DAG \hspace{34pt}        1} &
  93 &
  -17.6 &
  -47.6 \\
 &
  $G(n, p)$ &
  84 &
  14 &
  -5.2 &
  -25.4 \\
 &
  $G(n, p)$ DAG &
  93 &
  6 &
  -2.7 &
  -22.6 \\ \hline
\multirow{4}{*}{Rank-edge} &
  $G(n, p)$ local &
  67 &
  32 &
  -2.8 &
  -40.6 \\
 &
  \multicolumn{2}{l}{$G(n, p)$ local DAG \hspace{31pt}         63} &
  31 &
  -2.5 &
  -40.2 \\
 &
  $G(n, p)$ &
  78 &
  22 &
  -4.1 &
  -21.2 \\
 &
  $G(n, p)$ DAG &
  81 &
  19 &
  -4.1 &
  -15.2 \\ \hline
\end{tabular}
\caption{Effect of leakage with non-linear decision rules.}
\label{tab:leakage-non_lin-effect}
\end{table}
% \usepackage{multirow}
\begin{table}[]
\begin{tabular}{clcccc}
\hline
\multicolumn{1}{l}{\textbf{\begin{tabular}[c]{@{}l@{}}Decision rule\end{tabular}}} &
  \textbf{Graph family} &
  \multicolumn{1}{l}{\textbf{\begin{tabular}[c]{@{}l@{}}\% instances with\\ unchanged path len.\end{tabular}}} &
  \multicolumn{1}{l}{\textbf{\begin{tabular}[c]{@{}l@{}}\% instances with\\ smaller path len.\end{tabular}}} &
  \multicolumn{1}{l}{\textbf{\begin{tabular}[c]{@{}l@{}}Avg. \%\\ change\end{tabular}}} &
  \multicolumn{1}{l}{\textbf{\begin{tabular}[c]{@{}l@{}}Avg. \%\\ change opt.\end{tabular}}} \\ \hline
  \multirow{5}{*}{1.1 power} &
  $G(n, p)$ local &
  77 &
  23 &
  -4.6 &
  -4.6 \\
 &
  \multicolumn{2}{l}{$G(n, p)$ local DAG \hspace{28pt}     22} &
  78 &
  -29 &
  -29 \\
 &
  Grid &
  56 &
  44 &
  -8.3 &
  -16.4 \\
 &
  $G(n, p)$ &
  100 &
  0 &
  0 &
  0 \\
 &
  $G(n, p)$ DAG &
  96 &
  4 &
  -1.3 &
  -1.3 \\ \hline
\multirow{5}{*}{Quadratic} &
  $G(n, p)$ local &
  83 &
  17 &
  -3.3 &
  -4.4 \\
 &
  \multicolumn{2}{l}{$G(n, p)$ local DAG \hspace{28pt}     55} &
  45 &
  -13.8 &
  -19.6 \\
 &
  Grid &
  85 &
  12 &
  -1.7 &
  -15 \\
 &
  $G(n, p)$ &
  100 &
  0 &
  0 &
  0 \\
 &
  $G(n, p)$ DAG &
  99 &
  1 &
  -0.3 &
  -2.4 \\ \hline
\multirow{5}{*}{\begin{tabular}[c]{@{}c@{}}Quadratic\\ -with-\\ offset\end{tabular}} &
  $G(n, p)$ local &
  72 &
  28 &
  -5.7 &
  -6.5 \\
 &
  \multicolumn{2}{l}{$G(n, p)$ local DAG  \hspace{28pt}       14} &
  85 &
  -38.8 &
  -42 \\
 &
  Grid &
  27 &
  71 &
  -29.2 &
  -39.6 \\
 &
  $G(n, p)$ &
  100 &
  0 &
  0 &
  0 \\
 &
  $G(n, p)$ DAG &
  92 &
  8 &
  -3.9 &
  -5.4 \\ \hline
\multirow{5}{*}{Rank-edge} &
  $G(n, p)$ local &
  99 &
  1 &
  -0.1 &
  -0.1 \\
 &
  \multicolumn{2}{l}{$G(n, p)$ local DAG  \hspace{28pt}       95} &
  5 &
  -0.9 &
  -0.9 \\
 &
  Grid &
  83 &
  17 &
  -2.7 &
  -2.7 \\
 &
  $G(n, p)$ &
  100 &
  0 &
  0 &
  0 \\
 &
  $G(n, p)$ DAG &
  100 &
  0 &
  0 &
  0 \\ \hline
\end{tabular}
\caption{Effect of increasing flow with non-linear decision rules.}
\label{tab:inc_flow-non_lin-effect}
\end{table}

We discuss our major observations below:

\begin{itemize}
    \item \textbf{Usefulness of forces of leakage and increasing flow not limited to the linear rule}: For a large fraction of graph instances in each graph family, the path found in the presence of these forces has length (respectively leakage) smaller than or equal to the length (respectively leakage) of the path found in their absence. This holds for all non-linear decision rules considered. 
    
    \item \textbf{Forces of leakage and increasing flow more effective with weaker non-linearity}:
    One would intuitively hope that the closer a non-linear decision rule is to the linear rule, the more effective the forces of leakage and increasing flow would be. To study this, we  compare the quadratic decision rule with 1.1 power rule. For all graph families, we observe that the percentage of instances with strictly smaller path leakage (path length resp.) is higher for 1.1 power rule. We also observe that the average percentage change in path leakage (path length resp.) is more negative for 1.1 power rule. This shows that the forces of leakage and increasing flow are more effective for 1.1 power rule (which is closer to the linear rule) compared to the quadratic rule.
    
    \item \textbf{Effectiveness of the forces of leakage and increasing flow varies across graph families}: We observe that the effectiveness of leakage and increasing flow varies across graph types. For instance, across all non-linear rules considered, we observe that the percentage of instances with strictly smaller path leakage (path length resp.) is higher for $G(n, p)$ local and $G(n, p)$ local DAG graphs compared to $G(n, p)$ and $G(n, p)$  DAG graphs. We also observe that there are a few instances where these forces end up increasing the path leakage (path length resp.). For instance, with 1.1 power rule and for $G(n, p)$ local graph family, 1\% of instances end up with increased leakage compared to the baseline. It is an interesting direction for future research to understand what graph properties affect the effectiveness of leakage and increasing flow with non-linear decision rules.
    
    \item \textbf{Non-linear decision rules can prefer short paths even in the absence of increasing flow}: We observe that non-linear decision rules considered end up finding relatively short paths even without increasing flow. Such a behaviour is consistent with past observations \cite{goss1989self,chandrasekhar2018distributed} where various non-linear decision rules have been shown to find relatively short paths .  For instance, with $G(n, p)$ graphs, we observe that all the non-linear rules considered end up finding the shortest path for all graph instances.   However, increasing flow does help in nudging the dynamics towards shorter paths (when the dynamics does not already converge to the shortest path).
    
\item \textbf{Faster convergence with non-linear decision rules compared to linear decision rules}: For the non-linear decision rules we considered, we observe that simulations take significantly fewer iterations to converge compared to linear decision rules. This is true for settings involving leakage as well as increasing flow. Moreover, we observe that convergence is faster with stronger non-linearity. For instance, we observe that convergence with quadratic rule is faster compared to 1.1 power rule. Recall that we also observed that the forces of leakage and increasing flow were more effective with weaker non-linearity. In this sense, the strength of non-linearity can be thought of as a useful knob to balance the trade-off between convergence time and effectiveness of forces of leakage and increasing flow.
\end{itemize}    

\subsubsection{Parameter setting and implementation details}
In this section, we describe how we set various parameters for the non-linear decision rule simulations:

\begin{itemize}
    \item \textbf{Graph parameter setting}: For the $G(n, p)$ and $G(n, p)$ DAG graphs, we set $n = 100$ and $p = 0.1$. For $G(n, p)$ local and $G(n, p )$ local DAG, we set $n = 100, \ p=0.5$ and window size $k= 10$. For the grid graph, we consider a $10 X 10$ grid. Similar to the linear decision rule case, we plant a random shortest path to the graphs in increasing flow case so that the shortest path is unique.
    
    \item \textbf{Decay parameter and initial flow}: The decay parameter $\delta$ was set to $0.9$.  The initial forward flow level at $s$ and backward flow level at $d$ was chosen uniformly at random from $(0.5,1)$, and the initial flow at all other vertices was set to $0$. 
    
    \item \textbf{Leakage}:  For the case when there is leakage at vertices, the leakage at each vertex is chosen uniformly at random from $(0, 0.1)$.
    
    \item \textbf{Increasing flow rate}: In the case of increasing flow, we increase the incoming forward and backward flow by a factor of $1.1$ at each time step. Similar to the linear decision rule case, to avoid floating point  overflow, we implement this by decreasing all other flow and pheromone levels by a factor of $1.1$ except the incoming forward and backward flow. For any decision rule that only depends on the normalized pheromone levels , it is not difficult to see that only the relative values of flow and pheromone level matter, and our implementation gives rise to exactly the same dynamics (up to scaling) as when we increased the incoming forward and backward flow by a factor of $1.1$. All the non-linear decision rules we consider here only depend on the normalized pheromone levels except the quadratic-with-offset rule. For the quadratic-with-offset rule, we additionally also scale down parameter $c$ by $1.1$ in each step, which makes the dynamics equivalent (up to scaling) to the case where incoming flow is multiplied by $1.1$. When the value of pheromone level or flow at any edge becomes too small (smaller than the minimum allowed value  for a 64 bit floating point number $\approx 10^{-323}$), we round it to zero.
    
    \item \textbf{Decision rule parameters}: The quadratic and 1.1 power rules do not have any parameters to be set. Quadratic-with-offset rule has a parameter $c$. Note that when $c$ is large, the dynamics would not converge to a single path. We set $c$ to be small enough such that the dynamics for most instances end up converging to a single path (see definition of convergence below). For the leakage simulations, we set the value of $c = 0.25$ for all graph instances . For the increasing flow simulations, we set the value of $c = 1$ for the grid graph instances and $c = 0.5$ for all other graph instances.  Rank-edge decision rule also has a parameter $q$. We set $q = 0.01$ for all our simulations. 
    
    \item \textbf{Convergence criterion}: For all simulations (including linear rule simulations) except the quadratic-with-offset rule simulations, we consider the dynamics to have converged to some path when for each vertex on the path, at least 99\% of the forward (backward) flow entering it passes through its outgoing (incoming) edge on the path. For linear rule, this corresponds to each edge on the path having backward and forward normalized pheromone level at least 0.99.
    With quadratic-with-offset rule, we keep the threshold slightly lower to 95\%. If the threshold is too high for this rule, most instances would only converge to a path when $c$ is fairly small. However, with small $c$, this rule would be similar to the quadratic rule. To differentiate the dynamics with this rule from the quadratic rule, we keep the threshold slightly lower to 95\% which allows us to choose a slightly larger value of $c$.
    
    \item \textbf{Initial Pheromone levels}: We set initial pheromone level on all edges equal to one. Recall that with the linear rule, our simulations worked as expected even when the pheromone level at each edge is chosen uniformly at random from $(0, 1)$. However, for non-linear rules, we observed that for most instances, the simulations do not converge to a path when pheromone level is randomly initialized. In that sense, the linear rule seems more robust to pheromone initialization. Therefore, we run our simulations with non-linear rules with pheromone levels uniformly set to one so that the dynamics converges to a path. %Note that this way of pheromone initialization also seems more biologically realistic compared to random initialization.
    
    \item \textbf{Ignoring instances that do not converge to a path}: We discussed that we set the initial pheromone level uniformly to one so that the dynamics converges to a path. But even in this case, with non-linear decision rules, we observe that for a few instances the dynamics does not converge to a path. In our simulations, we ignored such instances, and only considered the instances where the dynamics converges to a path in both the baseline setting and the setting with leakage or increasing flow.
\end{itemize}

\end{document}